%% file: main_2026_new.tex
\documentclass[english, 12pt]{article}
\usepackage[utf8]{inputenc}
\usepackage{enumerate}
\usepackage[toc,page,header]{appendix}
\usepackage{minitoc}
\usepackage{amsopn,amsfonts,amsmath,amsthm,amssymb,mathrsfs,amstext}
\usepackage[pdftex]{graphicx}
\usepackage{verbatim}
\usepackage[width=\textwidth]{caption}
\usepackage[english]{babel}
\usepackage{subfigure}
\usepackage[]{float}
\usepackage[]{placeins}
\usepackage{flafter}
\usepackage[]{longtable}
\usepackage{dsfont}
\usepackage{array}	
\usepackage{natbib}
\usepackage{xr-hyper}
\usepackage{color}
\usepackage[flushleft]{threeparttable}
\usepackage{pdflscape}
\usepackage{makecell}
\usepackage{figsize}
\usepackage{multicol}
\usepackage[stable]{footmisc}

\usepackage{adjustbox} 
\usepackage{textcomp}
\usepackage{parskip}
\usepackage{tabularx}

\usepackage{graphicx}
\usepackage{booktabs,makecell}
\usepackage{array}
\usepackage{delarray}
\usepackage{dcolumn}
\usepackage{float}
\usepackage{psfrag}
\usepackage{multirow,hhline}
\usepackage{fancyhdr}
\usepackage{setspace}
\usepackage{booktabs,caption,fixltx2e}
\usepackage{xr}
\usepackage{enumitem}
\usepackage{arydshln}
\usepackage{xfrac}
\usepackage{adjustbox}
\usepackage{siunitx}

\usepackage{titling}
\setlist[enumerate]{labelsep=0.7pc, leftmargin=1.5pc}

\newtheorem{definition}{Definition}
\newtheorem{ass}{Assumption}
\usepackage[usenames,dvipsnames,svgnames,table]{xcolor}
\usepackage[colorlinks=true,linkcolor=blue,urlcolor=blue,filecolor=blue,citecolor=gray]{hyperref}
\usepackage{layout}

\usepackage{trimclip,xcolor}

\newcommand{\ind}[1]{\mathds{1}\left\{#1\right\}}

\usepackage[top=1.0in, bottom=1in, left=1in, right=1in]{geometry}

\def\indic{{\rm {\large 1}\hspace{-2.3pt}{\large
l}}}
\usepackage{pgfplots}
\pgfplotsset{compat=1.15}
\usepackage{mathrsfs}
\usetikzlibrary{arrows}
\definecolor{zzttqq}{rgb}{0.6,0.2,0}
\definecolor{ududff}{rgb}{0.30196078431372547,0.30196078431372547,1}

\usepackage{colortbl}
\usepackage{pgfplots}
\usepackage{pgfplotstable}

\newtheorem{thm}{Theorem}[section]

\newtheorem{prop}[thm]{Proposition}

\usepackage{setspace}

\def\PE{France Travail}

\def\Pbs{$\mathcal{U}$-\textsc{rec}}
\def\Mix{\textsc{Mix}}
\def\MixQuarter{\Mix-$\sfrac{1}{4}$}
\def\MixHalf{\Mix-$\sfrac{1}{2}$}
\def\MixThreeQuarter{\Mix-$\sfrac{3}{4}$}

\def\x{\mbox{\bf x}}
\def\y{\mbox{\bf y}}

\def\V{{\sc Vadore}}

\def\Vo{{\sc Vadore}}
\def\Vu{{\sc Vadore.1}}
\def\Vdt{{\sc Vadore.2}}
\def\Hu{{H.1.Hiring}}
\def\Hd{{H.2.Application}}

\pgfplotstableset{
    /color cells/min/.initial=0,
    /color cells/max/.initial=1000,
    /color cells/textcolor/.initial=,
    color cells/.code={%
        \pgfqkeys{/color cells}{#1}%
        \pgfkeysalso{%
            postproc cell content/.code={%
                \begingroup
                \pgfkeysgetvalue{/pgfplots/table/@preprocessed cell content}\value
                \ifx\value\empty
                \endgroup
                \else
                \pgfmathfloatparsenumber{\value}%
                \pgfmathfloattofixed{\pgfmathresult}%
                \let\value=\pgfmathresult
                \pgfplotscolormapaccess
                [\pgfkeysvalueof{/color cells/min}:\pgfkeysvalueof{/color cells/max}]%
                {\value}%
                {\pgfkeysvalueof{/pgfplots/colormap name}}%
                \pgfkeysgetvalue{/pgfplots/table/@cell content}\typesetvalue
                \pgfkeysgetvalue{/color cells/textcolor}\textcolorvalue
                \toks0=\expandafter{\typesetvalue}%
                \xdef\temp{%
                    \noexpand\pgfkeysalso{%
                        @cell content={%
                            \noexpand\cellcolor[rgb]{\pgfmathresult}%
                            \noexpand\definecolor{mapped color}{rgb}{\pgfmathresult}%
                            \ifx\textcolorvalue\empty
                            \else
                            \noexpand\color{\textcolorvalue}%
                            \fi
                            \the\toks0 %
                        }%
                    }%
                }%
                \endgroup
                \temp
                \fi
            }%
        }%
    }
}

\DeclareMathOperator*{\argmax}{\arg\!\max}

\title{\vspace{-0.6cm} A Job I Like or a Job I Can Get: Designing Job Recommender Systems Using Field Experiments\thanks{\footnotesize{Bied, Ghent University, IDLab: guillaume.bied@ugent.be; Caillou, UPSaclay/LISN/ INRIA: caillou@lri.fr; Cr\'epon: CREST, crepon@ensae.fr; Gaillac: University of Geneva, GSEM: christophe.gaillac@unige.ch; P\'erennes, CREST / France Travail: elia.perennes@ensae.fr; Sebag, UPSaclay/LISN/INRIA/CNRS: sebag@lri.fr. This paper is the result of a partnership with France Travail, the Public Employment Service in France. We thank Camille Qu\'er\'e, Chantal Vessereau, Cyril Nouveau, Paul Beurnier, Yann De Coster, Sebastien Robidou and Thierry Foltier for their operational support. 
The first experiment received IRB approval 2021-026 from the PSE-IRB and the second one IRB approval 2021-026-amendment. They were registered at the AEA's Registry for RCTs (respectively \url{https://doi.org/10.1257/rct.8998-1.3} and \url{https://doi.org/10.1257/rct.16650-1.0}).
This research was supported by the DATAIA convergence institute as part of the « Programme d’Investissement d’Avenir », (ANR-17-CONV-0003) operated by CREST and LISN. C. Gaillac acknowledges support from ERC POEMH 337665 and ANR-17-EURE-0010. Authors retained full intellectual freedom throughout this process, all errors are our own.  We thank Michele Belot, Morgane Hoffmann, Philipp Kircher, Rafael Lalive, Barbara Petrongolo, as well as seminar participants at Oxford, IZA/CREST Conference, ENSAI Economics Days, the 7th Lindau meeting, University of Bologna, 2023 EALE Conference, LISN, and the HiJoS Workshop in Innsbruck for useful comments and suggestions.}
 }}

\date{\normalsize \today}
\author{Guillaume Bied, Philippe Caillou, Bruno Cr\'epon, \\ Christophe Gaillac, Elia P\'erennes, Mich\`ele Sebag}

\begin{document}

\maketitle
 
\vspace{-0.6cm}
\begin{abstract} 
\footnotesize

Recommendation systems (RSs) are increasingly used to guide job seekers on online platforms, yet the algorithms currently deployed are typically optimized for predictive objectives such as clicks, applications, or hires, rather than job seekers' welfare. We develop a job-search model with an application stage in which the value of a vacancy depends on two dimensions: the utility it delivers to the worker and the probability that an application succeeds. The model implies that welfare-optimal RSs rank vacancies by an expected-surplus index combining both, and shows why rankings based solely on utility, hiring probabilities, or observed application behavior are generically suboptimal, an instance of the inversion problem between behavior and welfare. We test these predictions and quantify their practical importance through two randomized field experiments conducted with the French public employment service. The first experiment, comparing existing algorithms and their combinations, provides behavioral evidence that both dimensions shape application decisions. Guided by the model and these results, the second experiment extends the comparison to an RS designed to approximate the welfare-optimal ranking. The experiments generate exogenous variation in the vacancies shown to job seekers, allowing us to estimate the model, validate its behavioral predictions, and construct a welfare metric. Algorithms informed by the model-implied optimal ranking substantially outperform existing approaches and perform close to the welfare-optimal benchmark. Our results show that embedding predictive tools within a simple job-search framework and combining it with experimental evidence yields recommendation rules with substantial welfare gains in practice.\\\

\noindent \textbf{JEL Classification:} J64, J68, L86, C78, C55, C61

\noindent \textbf{Keywords:} Job Recommender Systems, Matching, Experiments, Machine Learning.

\end{abstract}

\newgeometry{top=1in, bottom=1in, left=1in, right=1in}

\onehalfspacing
\doparttoc 
\faketableofcontents 

\section{Introduction}

Recommendation systems (RSs) are transforming how job seekers interact with online labor-market platforms. Their adoption is accelerating, and many public employment services (PES) are considering integrating such tools \citep[see][]{organisation2023artificial}. Yet designs vary widely, and evidence on their relative effectiveness or their ability to generate meaningful labor-market improvements at scale remains limited, even if they may benefit specific subgroups. This paper develops a framework providing a welfare criterion for comparing designs, and uses two field experiments to show that algorithms informed by this criterion substantially outperform existing approaches.

\medskip

A large data-science literature develops job RSs by training algorithms on historical data to predict user behavior such as clicks, applications, or hires, and evaluating performance using predictive metrics \citep{freire2021recruitment, de2021job, mashayekhi2022challenge}.  While highly effective at forecasting observed behavior, their normative interpretation is less clear: predicting observed interactions does not imply that recommended opportunities maximize job seekers' expected welfare.

\medskip 

At the same time, a growing empirical literature in economics evaluates job recommendations in the field, often through randomized experiments conducted in partnership with PESs in several countries. Many of these interventions provide occupational or employer-level recommendations designed to broaden search or redirect job seekers toward labor markets with better prospects \citep[see, e.g.,][]{belot2019providing, altmann2022direct, behaghel2024potential, belot2025advising, bachli2025helping}. These studies provide valuable causal evidence, but they typically evaluate a small number of pre-specified recommendation rules against a limited set of outcomes, and offer little guidance on how to compare such rules against a common welfare criterion.

\medskip 

These two literatures have largely developed in parallel, though recent contributions connect them. \cite{roland2022} evaluate a collaborative filtering RS based on clicking behavior; \cite{SuBayoumiJoachims2022} consider social welfare maximization. Other work questions the practice of treating observed behavior as a proxy for welfare \citep[see, e.g.,][]{agan2023automating, kleinberg2022challenge, kleinberg2024inversion, mullainathan2025economics}. \citet{kleinberg2024inversion} formalize the inversion problem: when behavior is generated under frictions, predicting behavior is not equivalent to recovering the latent objective guiding optimal decisions.

\medskip 

This paper develops a job-search framework to discipline the design of RSs. The framework incorporates the application stage and highlights two objects: the utility a job seeker associates with a vacancy and the probability that an application results in a hire. The welfare-optimal rule ranks vacancies by an expected-surplus index combining both. A central implication is that neither utility or hiring outcomes alone provide a sufficient basis for optimal recommendations: the optimal algorithm requires a specific combination of utility and hiring probability that reflects the expected gains from applying. This is an instance of the inversion problem: recommendation rules that optimize predictive objectives such as predicted utility, observed application behavior, or observed hiring outcomes alone will generally miss it, even in a simple and frictionless job-search environment. An important empirical question is whether this inversion problem is economically significant in practice, and whether welfare-oriented rankings deliver substantial improvements over robust predictive rules in realistic environments.

\medskip 

We bring this framework to the data through close collaboration with the French PES. Starting from two existing RSs, one reflecting stated preferences, another a state-of-the-art ML system \citep[][]{vadoreijcai} predicting hiring outcomes, we conduct a sequence of two randomized field experiments. These experiments are not designed to evaluate the large-scale causal impact of deploying recommendations per se. Rather, they are conceived as beta-tests in a learning cycle \citep[in the spirit of][]{athey2018impact}, aimed at comparing alternative algorithmic designs and iteratively improving them.

\medskip 
A first experiment compares recommendations based on existing algorithms separately and in combination. We find that job seekers respond more favorably to recommendations combining information about preferences and hiring prospects, suggesting that neither dimension alone is sufficient to shape search behavior.

Guided by these findings and the model's structure, we design a new class of algorithms that better approximate the model-implied optimal ranking by incorporating explicit proxies for application behavior and preference-related signals into hiring predictions. This yields a family of recommendation rules including the original algorithms, an application-based algorithm, and an approximation of the welfare-optimal algorithm that combines information on preferences, applications, and hiring.

We evaluate these alternative designs in a second randomized experiment. The newly introduced algorithms, particularly the application-based algorithm and the approximation of the welfare-optimal rule, substantially outperform the initial approaches, especially in terms of clicks and applications. While they remain distinct from the theoretically optimal rule, their strong empirical performance highlights the practical gains from incorporating application behavior and preference-related signals into recommendation design.

A central goal of the paper is to use the conceptual structure of the job-search model to formally compare alternative RSs. Because optimality is defined within a behavioral model and a welfare criterion, this comparison requires that the model’s implications for application behavior be empirically plausible. We therefore exploit the exogenous variation generated by the random assignment of recommendation algorithms to estimate a structural model of application behavior. The results validate the model's core behavioral predictions: both the utility score and the inverse hiring probability are highly significant predictors of application decisions, with coefficients stable across specifications. This supports using the model as a conceptual basis for welfare comparisons.

The experimental data allow us to identify the predictive components needed to reconstruct the model-implied optimal recommendation rule. Using observed outcomes under different algorithms, we estimate predictions for both applications and hiring, and combine them to construct an estimate of the optimal recommendation score. This score is then used to identify the vacancies that would have been recommended under the optimal algorithm and to rank the tested algorithms against a common welfare benchmark. The application-based algorithm and our approximation of the welfare-optimal rule strongly outperform the initial algorithms. The latter further improves upon the application-based algorithm, albeit by a smaller margin, and performs remarkably close to the model-implied optimum.

Taken together, our results point to an important implication: RSs that target predictive objectives such as clicks, applications, or hiring outcomes are generally not aligned with job seekers' welfare. At the same time, the predictive tasks underlying these systems provide essential building blocks for welfare-relevant recommendation design. By combining structural modeling with experimental evidence that disciplines its behavioral assumptions, our approach illustrates how predictive tools can be used to construct, evaluate, and iteratively improve recommendation rules.

\medskip 

Our paper speaks to a growing empirical literature in economics that studies job RSs and related forms of automated advice, often through randomized field experiments conducted in partnership with public employment services \citep[see, e.g.,][]{belot2019providing, altmann2022direct, behaghel2024potential, belot2025advising, bachli2025helping, roland2022}. A central feature of this literature is the diversity of objectives implicitly targeted by recommendation rules. Some interventions rely on observed transitions or predicted hiring probabilities to steer job seekers toward occupations or employers with higher employment prospects \citep[e.g.,][]{belot2019providing, altmann2022direct, behaghel2024potential, belot2025advising}; others emphasize measures of fit or expressed interest inferred from search behavior or skills profiles \citep[e.g.,][]{bachli2025helping, roland2022}. While these approaches provide valuable causal evidence on the effects of specific recommendation designs, they do not offer a general framework to compare alternative algorithmic objectives or relate them to a common welfare criterion.

We contribute to this literature by providing such a framework and comparing different RSs using experiments. Within a simple job-search model that explicitly incorporates the application stage, we show that two dimensions are central to recommendation design: the utility that a job seeker associates with a vacancy and the probability that an application results in a hire. Existing approaches can be interpreted as emphasizing one of these dimensions in isolation, but neither is sufficient on its own. An economically meaningful ranking must combine both into a single expected-surplus index, which allows us to place diverse recommendation designs on a common footing and to evaluate their relative performance.

Our analysis also relates to a growing literature at the intersection of machine learning and economics that questions the normative interpretation of observed behavior on digital platforms \citep[see, e.g.,][]{agan2023automating, kleinberg2022challenge, kleinberg2024inversion, mullainathan2025economics}. In our setting, the inversion problem arises because application decisions only reveal whether applying is privately profitable, hiring outcomes capture only part of the expected gains, whereas welfare-relevant ranking of vacancies depends on the magnitude of expected gains.  

Our contribution is to make this insight operational in a labor-market environment. By combining experimental variation with a structural model of application behavior, we characterize job seekers' welfare-relevant objective, distinguish it from commonly used behavior-based rankings, and quantify the importance of the inversion problem. While the model implies that welfare-relevant rankings differ from application-based recommendations, our results show that the gap is positive but quantitatively modest, an assessment that would be difficult to obtain without jointly leveraging experimental evidence and structural modeling.

Our methodology relates closely to work emphasizing experimentation and economic modeling in the design of algorithmic decision rules \citep[e.g.,][]{athey2018impact}. Rather than evaluating the large-scale causal impact of a fixed RS, we use sequential beta-tests to compare designs, feed the results back into the model, and construct improved algorithms. The structural model disciplines which signals to incorporate and how to combine them, while the experiments in turn discipline the model’s behavioral assumptions and provide the variation needed to estimate the welfare metric. Empirically, RSs informed by both utility and hiring probabilities substantially outperform approaches based on either alone, yielding sizeable gains relative to algorithms currently used in practice, while the inversion problem has modest quantitative implications in this setting.


The paper proceeds as follows. Section \ref{sec:search} presents the job-search model and derives the optimal recommendation rule. Section \ref{sec:2RS} describes two representative RSs and their underlying scores. Section \ref{sec:betatestso} presents the design and results of the randomized experiments. Section \ref{sec:sec6} estimates the structural model and compares alternative algorithms using the proposed metric. Section \ref{sec:conclusion} concludes with implications for the design of job RSs in practice.

\section{Job search model with recommender systems}\label{sec:search}

\paragraph{Overview.}
We develop a job-search model in which job seekers encounter vacancies sequentially and decide whether to apply. Vacancies are lotteries characterized by (i) the utility they would deliver to the worker and (ii) the probability that an application succeeds. We then introduce a RS as a technology that (a) restricts the set of vacancies a job seeker is exposed to based on a score and (b) potentially increases the rate at which vacancies can be processed. The model delivers (i) a model-implied value index for vacancies, (ii) conditions under which an RS improves welfare, and (iii) an optimal ranking rule in the benchmark case of myopic job seekers. Finally, we discuss how these results inform RS design and clarify why ranking by utility alone, by hiring probability alone, or by predicted application behavior is generally suboptimal.

\subsection{Environment and primitives}\label{subsec:primitives}
Our job search model builds on the following environment and primitives.
\paragraph{\textbf{Job seekers and vacancies.}} Job seekers are indexed by characteristics $x$ and vacancies by characteristics $y$. An unemployed job seeker receives flow utility $u(b)$.
\paragraph{\textbf{Payoffs.}} A vacancy $y$ yields utility $U(x,y)+\varepsilon_{i,y}$ to job seeker $i$ of type $x$, where $\varepsilon$ is an idiosyncratic taste shock observed by the job seeker. We assume $\varepsilon$ follows a logistic distribution with scale parameter $\sigma$ and cumulative distribution function $F_{\varepsilon}(\cdot)=F(\cdot/\sigma)$.

\paragraph{\textbf{Hiring probability.}} Conditional on applying, the job seeker is hired with probability $p(x,y)$, which is known to the job seeker at the application decision stage.
\paragraph{\textbf{Vacancy distribution and reparametrization.}} Vacancies are drawn from $F_0(y)$. For a given type $x$, the induced distribution of $(p(x,y),U(x,y))$ is denoted $F_0(p,U)$.\footnote{with $x$ suppressed in the notation.}
\paragraph{\textbf{Other primitives.}} Applications entail a cost $k$ and rejection entails a psychological cost $R$. Matches separate at rate $q$ and future utility is discounted at rate $r$.

\subsection{Baseline sequential search without an RS}\label{subsec:baseline}

This section characterizes job seekers’ behavior and the value of unemployment in the baseline sequential search environment without an RS. The key friction is that job seekers explore vacancies sequentially, drawing from the distribution $F_0(p,U)$ at rate $\alpha_0$, and cannot simultaneously compare all available opportunities. Upon encountering a vacancy, the job seeker observes its characteristics $(p,U,\varepsilon)$ and decides whether to apply. All formal derivations and proofs are relegated to Appendix~\ref{app:model}.

\subsubsection{Reservation utility and surplus}\label{subsubsec:baseline_surplus}

Let $V_0(x)$ denote the discounted value of unemployment for a job seeker of type $x$ in the absence of a RS. Each time the job seeker encounters a vacancy with characteristics $(p,U,\varepsilon)$, she faces a binary decision: apply to it or continue searching.

If the job seeker does not apply, she remains unemployed and retains continuation
value $V_0(x)$. If she applies, she pays the application cost $k$; with probability
$p$ the application succeeds and she transitions into employment with discounted present value of utility $V_e(x,U+\varepsilon)$, while with probability
$1-p$ she is rejected and continues searching, incurring the rejection cost $R$. Conditional on applying, the discounted expected payoff is therefore
\[
p\,V_e(x,U+\varepsilon)
+
(1-p)\big(V_0(x)-R\big)
-
k.
\]
To characterize this apply-or-not decision, it is useful to introduce a reservation
utility:
\begin{equation}\label{eq:res_u}
U_0^*(x,p)
=
rV_0(x)
-
\overline{R}
+
\frac{\overline{k}+\overline{R}}{p},
\end{equation}
where $\overline{R}=(r+q)R$ and $\overline{k}=(r+q)k$. The quantity $U_0^*(x,p)$ represents the minimum utility level that makes a vacancy
with hiring probability $p$ worth applying to.

We then define the surplus associated with a vacancy as
\begin{equation}\label{eq:def_surplus}
\Delta(x,p,U)
:=
U
-
U_0^*(x,p).
\end{equation}
This surplus compares the utility provided by the vacancy to the reservation utility
associated with its hiring probability.
Vacancies with lower hiring probabilities must offer higher utility in order to be
attractive, reflecting the costs associated with unsuccessful applications. To streamline notation, we henceforth suppress the dependence on $x$ whenever this does
not create ambiguity.

\subsubsection{Application behavior and vacancy value}\label{subsubsec:baseline_value}

The job seeker applies to a vacancy whenever the realized surplus is positive:
\begin{equation}\label{eq:appli}
\Delta(p,U)+\varepsilon>0.
\end{equation}
This rule admits a natural interpretation: a job seeker applies if the realized utility
$U+\varepsilon$ exceeds the reservation utility $U_0^*(p)$.

Under the assumed logistic distribution for $\varepsilon$, this decision rule implies
that the probability of applying to a vacancy $(p,U)$ is $ p_a(p,U) =
F\!\left(\Delta(p,U)/\sigma\right)$.

The discounted value of unemployment satisfies
\begin{equation}\label{eq:Bellman1_pi}
rV_{0}
=
u(b)
+
\frac{\alpha_0}{r+q}\,
\mathbb{E}\!\left(\Gamma(p,U)\right),
\end{equation}
where the expectation is taken over the distribution of vacancies and
\begin{equation}\label{eq:gamma_mplus}
\Gamma(p,U)
:=
p
\int
\left(\Delta(p,U)+\varepsilon\right)
\ind{\Delta(p,U)+\varepsilon>0}
\,dF_{\varepsilon}(\varepsilon).
\end{equation}
The index $\Gamma(p,U)$ represents the expected contribution of encountering a vacancy
with characteristics $(p,U)$ to the job seeker’s continuation value.
It aggregates the probability of applying, the probability of being hired conditional
on application, and the surplus generated by a successful match conditional on application.

Under the logistic assumption on $\varepsilon$, $\Gamma(p,U)$ admits the closed-form expression
\begin{equation}\label{eq:gamma_plus}
\Gamma(p,U)
=
p\,\sigma
\log\!\left(1+e^{\Delta(p,U)/\sigma}\right).
\end{equation}
\begin{prop}[Application rule and vacancy value]\label{prop_V_u0}
In the absence of a RS, application behavior is governed by the surplus $\Delta(p,U)$ through the rule~\eqref{eq:appli}, while the welfare-relevant value of a vacancy is summarized by the index $\Gamma(p,U)$. 
\end{prop}

\begin{proof}
See Appendix~\ref{app:proof}.
\end{proof}

Proposition \ref{prop_V_u0} makes explicit the distinction between application behavior and vacancy value. This non-equivalence arises because application decisions are governed by a binary profitability condition, whereas vacancy values depend on the expected magnitude of the gains. A vacancy with a high hiring probability $p$ contributes more to the job seeker's continuation value than one with low $p$, even if both exceed the application threshold. Rankings based on observed applications therefore do not generally coincide with rankings that maximize job seekers' welfare. This distinction is central to the design of RSs, and Section~\ref{subsec:implications} characterizes it precisely.

\subsection{Recommender systems: selection and exposure}\label{subsec:rs}

A recommender system (RS) scores the pool of available vacancies and pre-selects a subset to be shown to the job seeker, potentially also increasing the rate at which vacancies can be processed. 
While decisions remain sequential on the worker side, the RS transforms sequential search over vacancies into sequential applications over a curated set.

In our framework, RSs affect job search through two channels:
\begin{enumerate}[label=(\roman*)]
    \item \textbf{Selection.} The RS restricts the job seeker’s consideration set to vacancies
    whose score exceeds a threshold, \emph{i.e.}, the top $s$ fraction according to a score.
    \item \textbf{Exposure.} The RS may change the effective rate at which vacancies can be
    processed, from $\alpha_0$ to $\alpha_1$.
\end{enumerate}

\subsubsection{Scores and the induced distribution of considered vacancies}\label{subsubsec:truncation}

For a job seeker with characteristics $x$ and a vacancy described by characteristics $y$, the RS computes a score $S(x,y)$. To keep the notation light, we suppress the dependence on $x$ and $y$  whenever it is unambiguous.

For a given $x$, the joint distribution of vacancy characteristics and scores is the distribution of $\big(p(x,y),\,U(x,y),\,S(x,y)\big)$, 
when $y \sim F_0$, which we denote by $F_0(p,U,S)$.\footnote{As in Section~\ref{subsec:baseline}, this notation suppresses the dependence on $x$. Formally, $F_0(p,U,S)$ is the distribution of $(p(x,y),U(x,y),S(x,y))$ induced by $F_0(y)$ conditional on $x$.}

We model RS selection as recommending the vacancies whose score lies above the
$(1-s)$-quantile of the score distribution.
Let $\overline{q}_S(s)$ denote the quantile of order $1-s$ of $S$ under $F_0(p,U,S)$.
Then the RS induces the truncated distribution
\begin{equation}\label{eq:trunc}
dF_1(p,U,S)
=
\frac{\ind{S>\overline{q}_S(s)}}{s}\,dF_0(p,U,S),
\end{equation}
which is the distribution of $(p,U,S)$ among recommended vacancies.

\subsubsection{Exposure: arrival/processing rate}\label{subsubsec:rs_channels}

In addition to changing the composition of vacancies considered, RSs may also change the intensity of exposure to vacancies. We capture this by allowing the effective vacancy-processing rate under recommendations to be $\alpha_1$ rather than $\alpha_0$. This reduced-form formulation accommodates multiple mechanisms: recommendations may lower cognitive and search costs, facilitate navigation, or simply deliver a fixed number of suggestions over a given period.

\subsubsection{Myopic benchmark: value under recommendations}\label{subsubsec:myopic_value}

We first consider the case of \emph{myopic} job seekers, who do not adjust their reservation utility in response to the introduction of the RS. They continue to evaluate vacancies using the baseline reservation utility $U_0^*(p)$ derived in Section~\ref{subsec:baseline}.
Accordingly, $\Gamma^m(p,U)$ is defined by the same expression as $\Gamma(p,U)$, but evaluated at the baseline continuation value $V_0$. In particular, $\Gamma^m(p,U)=\Gamma(p,U)$ in the absence of an RS.

Under an RS characterized by $(S,s,\alpha_1)$, the discounted value for a myopic unemployed job seeker is
\begin{equation}\label{eq:myopic}
rV^m_{1}(s,\alpha_1)
=
u(b)
+
\frac{\alpha_1}{r+q}\,
\frac{\mathbb{E}\!\left(\Gamma^m(p,U)\ind{S>\overline{q}_S(s)}\right)}{s},
\end{equation}
where the expectation is taken with respect to $F_0(p,U,S)$.

\subsection{Optimal RS and improvement conditions}
\label{subsec:optimal}

\subsubsection{Objective and definition of an optimal RS}
\label{subsubsec:def_opt}

We now define what it means for a RS to be optimal.
Since the RS is designed and implemented by the PES,
we must specify the objective it pursues.
Throughout the paper, we assume that the PES aims to maximize the welfare of job seekers.
We abstract from firms’ outcomes and from broader social objectives.

RSs rank vacancies using a score $S$ and recommend the top fraction $s$ of vacancies
according to that score. An optimal RS is therefore one that, for any recommendation intensity $s$,
maximizes the expected value of the recommended vacancies.

\begin{definition}\label{def:rec}
For myopic job seekers, an \emph{optimal RS} is a measurable scoring rule $S^*$ of $(p,U)$ that solves,
for each $s\in[0,1]$,
\[
S^*
\in
\argmax_{S }
\;
\mathbb{E}\!\left[
\Gamma^{m}(p,U)\;
\indic\big\{ S(p,U)>\overline{q}_{S}(s) \big\}
\right].
\]
\end{definition}
The requirement that optimality holds for all $s$ emphasizes that the RS defines a global ranking, rather than being tailored to a specific cutoff or number of recommendations.

\subsubsection{Characterization and sufficient conditions}
\label{subsubsec:prop_myopic}

Characterizing the optimal RS in the myopic case is straightforward and yields
sharp implications for practice.

\begin{prop}\label{prop:myopic}
Consider a RS based on a score $S(p,U)$ that selects the top fraction $s$
of vacancies.
\begin{enumerate}
    \item\label{2.2.1} The optimal ranking is obtained by $S(p,U)=\Gamma^m(p,U)$.
    \item\label{2.2.2} If $\alpha_1\geq\alpha_0$, a sufficient condition for an $S$-based RS to improve
    job seekers’ welfare relative to baseline search is that
    $z\mapsto \mathbb{E}\!\left(\Gamma^m(p,U)\mid S=z\right)$ is increasing.
    In particular, the RS based on $S(p,U)=\Gamma^m(p,U)$ strictly improves welfare whenever $\Gamma^m$ is non-degenerate.
\end{enumerate}
\end{prop}

\begin{proof}
See Appendix~\ref{app:proof}.
\end{proof}

\paragraph{Interpretation.}
Proposition~\ref{prop:myopic} establishes that the optimal RS solves a \emph{global ranking problem} over the entire distribution of vacancies, whereas application behavior reflects a binary profitability condition. Together with Proposition~\ref{prop_V_u0}, this creates an \emph{inversion problem}: algorithms that target observed applications or hiring outcomes are not generally aligned with job seekers' welfare-relevant ranking. Section~\ref{subsec:implications} characterizes the nature and magnitude of this gap for several natural recommendation rules.

\subsection{Implications for the design of recommender systems}
\label{subsec:implications}

This section studies the gap between the optimal RS characterized in Proposition~\ref{prop:myopic} and alternative recommendation rules that may appear reasonable \emph{a priori}. We consider rankings based on perceived utility $U$, on the probability of hire $p$, on observed application behavior $p_a$, and on observed hires. The latter are of special interest, since applications are the observable outcome of job seekers’ decisions and hiring is a key outcome of the search process, making it natural to ask whether recommendation rules that replicate observed applications or hires can be normatively justified.

Building on the model and welfare criterion introduced above, we compare these alternative rules to the welfare-relevant objective. The key result of this section is a unifying decomposition of the welfare-relevant score, from which the non-optimality of several natural recommendation rules follows directly.

\subsubsection{Decomposing the welfare-relevant score}
\label{subsubsec:decomp}

In the myopic case, the welfare-relevant index $\Gamma(p,U)$ admits the 
decomposition:
\begin{equation}\label{eq:decomposition_short}
\Gamma(p,U)
=
p \times p_a(p,U) \times
\mathbb{E}\!\left[\Delta(p,U)+\varepsilon \mid \Delta(p,U)+\varepsilon>0, p, U\right].
\end{equation}
This expression makes clear that the optimal ranking is the product of three distinct components:
(i) the probability of being hired conditional on applying ($p$), (ii) the probability of applying ($p_a(p,U)$), and (iii) the expected surplus generated by a successful application, conditional on applying.

Since $p_a(p,U)=\mathbb{P}(\Delta(p,U)+\varepsilon>0 \mid p, U)$ is a monotone function of $\Delta$ under mild conditions, the conditional expectation in
\eqref{eq:decomposition_short} can be expressed as a function of the application
probability.
Let $m(p_a(p,U)) :=
\mathbb{E}\!\left[\Delta(p,U)+\varepsilon \mid \Delta(p,U)+\varepsilon>0, p, U\right]$, 
which allows us to rewrite the optimal score as $ \Gamma(p,U) = p \times p_a(p,U) \times m(p_a(p,U))$.

This decomposition is generic and does not rely on a specific parametric assumption on
the distribution of the idiosyncratic shock $\varepsilon$.
Under the logistic assumption adopted in the baseline model, $\Gamma(p,U)$ admits the
closed-form expression given in Equation~\eqref{eq:gamma_plus}, which can be written as
\begin{equation}\label{eq:gamma-end}
\Gamma(p,U)
=
p \times p_a(p,U)
\times
\left[-\log\!\big(1-p_a(p,U)\big)\big/p_a(p,U)\right].    
\end{equation}
The last term, $m(p_a)=-\log(1-p_a)/p_a$, is a convex and increasing function of $p_a$, with $\lim_{p_a\to 0}m(p_a)=1$ and $\lim_{p_a\to 1}m(p_a)=+\infty$. Thus,  for small $p_a$, a second-order Taylor expansion gives $m(p_a)\approx 1+p_a/2$, so that $\Gamma(p,U)\approx p_h(p,U)\times(1+p_a(p,U)/2)$, where $p_h$ is the unconditional probability of being hired,
\begin{equation}\label{eq:ph_def}
p_h(p,U) := p \times p_a(p,U),
\end{equation}
i.e., the product of the hiring probability conditional on application and the application probability. When application probabilities are empirically small, $\Gamma$ is therefore well approximated by $p_h$, which rationalizes the strong empirical performance of hiring-based rankings documented in Section~\ref{sec:comparisonRS}.

More generally, one can show that if the distribution of $\varepsilon$ is log-concave,
then the conditional surplus $m(p_a)$ is an increasing function of the application
probability.
The precise shape of this function depends on the distributional assumption.
Appendix Figure~\ref{graph:app:lois_epsilon} illustrates this additional term for three
standard cases: logistic, Gumbel (EV1), and normal distributions, all normalized to have
unit variance.

The decomposition \eqref{eq:gamma-end} immediately implies that neither $U$, $p$, $p_a$, nor $p_h$ alone is sufficient to recover the optimal ranking: each captures only a subset of the dimensions jointly determining $\Gamma(p,U)$.

\subsubsection{Why ranking by application probability is not optimal}
\label{subsubsec:why_application}

Ranking vacancies according to observed application behavior is a natural benchmark,
since applications are directly observed and summarize job seekers’ choices.
Within the model, however, the Bellman equation makes clear that application behavior
and welfare-relevant vacancy values rely on fundamentally different objects.

When a job seeker applies to a vacancy $(p,U)$, the contribution of this vacancy to her
continuation value is proportional to $p\big(U - U_0^*(x,p) + \varepsilon\big)$, 
where $p$ is the probability of being hired conditional on application.
By contrast, the application decision itself is governed by a threshold rule: a job seeker applies whenever this expression is positive.

As a result, application behavior only reveals whether applying is privately profitable,
but abstracts from the probability $p$ that scales the contribution of the vacancy to expected welfare. Two vacancies may generate the same surplus $U - U_0^*(p) + \varepsilon$ and therefore
induce the same application decision, yet differ substantially in their welfare contribution because they are associated with different hiring probabilities $p$.

This wedge follows directly from the dynamic structure of the search problem and provides a structural explanation for the inversion problem emphasized in the algorithmic fairness literature: a score that reproduces application behavior generally fails to recover the welfare-relevant ranking of vacancies.

\subsubsection{Why ranking by hiring probability is not optimal}
\label{subsubsec:why_ph}

A closely related benchmark is to rank vacancies according to the probability of hire $p_h(p,U)$ (see Equation \ref{eq:ph_def}).
Using the decomposition in \eqref{eq:gamma-end}, the optimal score can be written as $$\Gamma(p,U) = p_h(p,U) \times m\!\big(p_a(p,U)\big).$$
This expression shows that $p_h$ is still an incomplete measure of vacancy value.
Under log-concavity of the distribution of $\varepsilon$, the function $m(p_a)$ is
increasing, implying that $p_h$ does not place sufficient weight on perceived utility
and application behavior. Ranking by $p_h$ underweights vacancies that generate large conditional surpluses.

That said, for standard distributions such as the logistic or Gumbel (EV1), the function
$m(p_a)$ varies relatively slowly when $p_a$ is small.
As a result, in environments where application probabilities are low, rankings based on
$p_h$ may perform reasonably well in practice, even though they are not theoretically
optimal.

\subsubsection{Two implementable routes to welfare-oriented recommendation}
\label{sec:two-routes}

Proposition~\ref{prop:myopic} characterizes the welfare-optimal ranking through the score
$\Gamma(p,U)$. In practice, two routes can be followed to implement such a
ranking.

\paragraph{Route A: structural implementation.}
A first approach is to recover the primitives entering $\Gamma(p,U)$ by separately estimating a utility component $U_{i,j}$ and the probability of recruitment $p_{i,j}$, and then computing the welfare score implied by the model. While conceptually direct, this route requires specifying how observed platform signals map into the latent utility object $U_{i,j}$.

\paragraph{Route B: reduced-form welfare index.}
An alternative is to construct the welfare score directly from observable transition probabilities. Let
\[
p_{a,i,j} = \mathbb{P}(C_{i,j}=1 \mid x_i,y_j), 
\qquad
p_{i,j} = \mathbb{P}(H_{i,j}=1 \mid C_{i,j}=1,  x_i,y_j)
\]
denote the probability of applying and the probability of recruitment conditional on application. Under the logistic specification discussed above, the welfare score can be written as \eqref{eq:gamma_plus}-\eqref{eq:gamma-end}, which yields
$\Gamma_{i,j}
=
p_{i,j}\,[-\log(1-p_{a,i,j})]$, up to a positive scale normalization.

This second route is directly implementable since it relies only on predicting applications and hires. The empirical strategy developed below follows this approach: we estimate $\widehat p_{a,i,j}$ and $\widehat p_{i,j}$ using the available signals, construct $\widehat{\Gamma}_{i,j}$, and use it to form welfare-oriented recommendation sets.


Implementing such a welfare-oriented ranking relies on the behavioral structure linking applications to both the surplus component $U_{i,j}$ and the probability of recruitment $p_{i,j}$. Establishing the empirical relevance of this structure is therefore a natural prerequisite for algorithm design.\footnote{While the experiments reported below provide evidence consistent with this behavioral structure, they do not identify the exact distribution of taste shocks.} Random assignment of recommendation algorithms generates exogenous variation in the signals observed by job seekers. This variation allows us to identify how application decisions respond to both utility-related signals and recruitment probabilities.

\subsection{Identification of the objective $\Gamma$ from application data}
\label{subsec:identification}

\subsubsection{Identification of $\Delta(p,U)$ and reconstruction of $\Gamma(p,U)$}
\label{subsubsec:identify_delta}

Under the logit taste-shock specification, application decisions identify the surplus
index $\Delta(p,U)$ that governs job seekers’ application behavior.
As Equation~\eqref{eq:gamma_plus} shows, the welfare-relevant vacancy value
$\Gamma(p,U)$ is a function of this surplus index.
Given identification of $\Delta(p,U)$, the structural assumptions of the model allow us
to reconstruct the welfare-relevant objective $\Gamma(p,U)$ and thus the optimal
ranking.
Identification of $\Gamma(p,U)$ in this framework therefore primarily hinges on the
ability to identify $\Delta(p,U)$ from application data.

We denote by $C_{i,j}=1$ if job seeker $i$ applies to vacancy $j$, and $0$ otherwise.
Let $\tilde U_{i,j} \equiv U_{i,j}-U_{0,i}^*(1)$ denote utility net of the continuation
value of unemployment.

\subsubsection{A sufficient parametric identification result}
\label{subsubsec:identify_prop}

We provide a sufficient parametric identification result under a logistic
specification.

\begin{prop}\label{prop:identification}
Let $\varepsilon$ be distributed as a logistic random variable with scale parameter
$\sigma$.
Assume that sequences of individual $i$'s application decisions on vacancies $j$ are
observed, together with $p_{i,j}$ and $\tilde U_{i,j}$.
Then, the parameters $\alpha$,\footnote{This $\alpha$ is not to be confused with the vacancy-arrival rates $\alpha_0, \alpha_1$.} $\beta$, and $\gamma$ in the binary choice model
\begin{equation}\label{eq:identification}
\mathbb{P}(C_{i,j}=1 \mid p_{i,j},\tilde U_{i,j})
=
F\!\left(
\alpha \tilde U_{i,j}
-\beta/p_{i,j}
+\gamma
\right)
\end{equation}
are identified. Moreover, for generic values $(p,U)$, the structural objects $\sigma$, $\overline{k}+\overline{R}$,
$\Delta(p,U)$, and $\Gamma(p,U)$ in equations~\eqref{eq:def_surplus} and~\eqref{eq:gamma_mplus} are identified as follows:
\begin{itemize}[label=--]
    \item[(i)] $1/\alpha$ identifies $\sigma$;
    \item[(ii)] $\beta/\alpha $ identifies $\overline{k}+\overline{R}$;
    \item[(iii)] $U-(\beta/\alpha)/p+\gamma/\alpha$ identifies $\Delta(p,U)$;
    \item[(iv)] $p\log\!\left(1+e^{\alpha U+\gamma-\beta/p}\right)/\alpha$ identifies
    $\Gamma(p,U)$.
\end{itemize}
\end{prop}

\begin{proof}
See Appendix~\ref{app:proof}.
\end{proof}

In Proposition \ref{prop:identification}, the assumption that $\tilde U_{i,j}$ is directly observed can be relaxed by allowing
for an individual-specific error term $\nu_i$, so that utilities take the form
$\tilde U_{i,j}+\nu_i$.
In this case, identification of $\sigma$, $\overline{k}+\overline{R}$,
$\Delta(p,U)$, and $\Gamma(p,U)$ follows from a fixed-effects panel logit specification
adapted from Equation~\eqref{eq:identification}.

\subsection{Extensions and scope}
\label{subsec:extensions}

This section discusses extensions of the baseline framework and clarifies the scope of the analysis. Our objective is not to provide a full treatment of these extensions in the main text, but rather to explain how the core insights extend beyond the baseline model and to indicate where additional complexities arise. 

\subsubsection{Forward-looking job seekers}
\label{subsubsec:nonmyopic}

The analysis in the main text focuses on myopic job seekers.
In Online Appendix~\ref{sec:nonmyopicth}, we extend the framework to allow for
forward-looking behavior. We address two issues.

First, we examine how to evaluate the gains associated with the use of a RS when job seekers are not myopic. We show that, using observable data, it is possible to adjust the welfare evaluation derived under the myopia assumption to account for the endogenous adjustment of the reservation utility implied by forward-looking behavior.

Second, we consider the implications for RS design. While the optimal rule under myopia is no longer exactly optimal in this setting, recommendation sets constructed under this assumption can still be improved. Although deriving a fully optimal rule with forward-looking job seekers is challenging, the RS remains effective in steering job seekers toward vacancies with higher expected value.

\subsubsection{Imperfect knowledge about the distribution of opportunities}
\label{subsubsec:misperceptions}

The baseline model abstracts from potential misperceptions about the distribution of
available vacancies.
In practice, job seekers may hold biased beliefs about the types of opportunities they
are likely to encounter and form expectations accordingly.

To fix ideas, consider the case of myopic job seekers.
Absent recommendations, job seekers base their search decisions on a subjective
distribution of vacancies, which gives rise to a continuation value denoted
$V_0^{JS}$.
By contrast, the true distribution of opportunities would imply a value $V_0^{True}$. Introducing a RS provides access to vacancies drawn from the true distribution and yields a continuation value $V_1^m$, without modifying the model-implied optimal ranking rule $\Gamma^m$ for recommended vacancies.

The overall effect of the RS can therefore be decomposed as:
\[
\underbrace{V_1^m - V_0^{JS}}_{\text{Full RS effect}}
=
\underbrace{V_1^m - V_0^{True}}_{\text{Pure RS effect}}
+
\underbrace{V_0^{True} - V_0^{JS}}_{\text{Information effect}}.
\]

This decomposition highlights that, in addition to alleviating search frictions, RSs may
generate value by correcting job seekers’ misperceptions about the distribution of
opportunities.\footnote{Formally, subjective beliefs can be represented by assuming that
job seekers draw vacancies from a subjective distribution $dF_0^{JS}$, while
recommended vacancies are drawn from the true distribution $dF_0^{True}$.
The relation between the two can be expressed, for example, as
$dF_0^{True} = \phi(p,U)\, dF_0^{JS}$, where $\phi(\cdot)$ captures systematic belief
distortions.
Under this formulation, recommendations shift job seekers from draws based on
$dF_0^{JS}$ to draws based on $dF_0^{True}$, without altering the model-implied optimal
ranking.
For forward-looking job seekers, the same logic applies, although the interaction with
the endogenous reservation utility makes the analysis more involved without altering
the qualitative insight.}

\subsubsection{Competition and congestion}
\label{subsubsec:congestion}

The model abstracts from competition among job seekers and from congestion effects.
At the recommendation stage, vacancies are ranked independently for each job seeker, so
the RS effectively solves a collection of individual optimization problems.

In environments with many job seekers, this approach may lead multiple individuals to
be recommended the same vacancy, potentially exacerbating congestion and altering
effective hiring probabilities.
Accounting for such interactions would require formulating a global assignment problem
that imposes constraints on how often a vacancy can be recommended.

From an algorithmic perspective, this can be implemented as a post-processing step on
top of a proximity matrix between job seekers and vacancies, for instance using optimal
transport methods that explicitly trade off match quality and congestion
\citep[see, e.g.,][]{bied2021congestion,mashayekhi2023recon}.
In this paper, we deliberately abstract from congestion and focus on recommendation
rules that ignore these interactions.

 \section{Two representative job RSs}\label{sec:2RS}
Section~\ref{sec:search} shows that welfare-optimal recommendations must combine two primitives: the utility $U(x,y)$ a vacancy delivers to a job seeker and the probability $p(x,y)$ that an application succeeds. This section presents the two RSs at the core of the paper's experimental design, chosen precisely because each emphasizes one of these two components. The first, a state-of-the-art ML algorithm trained on realized hires, primarily targets $p$; the second, a knowledge-based matching algorithm derived from job seekers' stated search criteria, primarily targets $U$. Neither is welfare-optimal in isolation (Proposition~\ref{prop_V_u0}), but they provide the building blocks for the welfare-approximating algorithms constructed in Section~\ref{sec:betatestso}. Before describing these two systems, we briefly situate them within the broader RS landscape.
 \subsection{Many forms of RSs}\label{sec:sec1}

All RSs operate according to a common principle: they rely on the computation of a matching score that summarizes information about the expected value of a job seeker–vacancy match. Specifically, for an individual $i$ described by a set of characteristics $x_i$ and a vacancy $j$ described by characteristics $y_j$, the system computes a score $S_{i,j}$ depending on $x_i$ and $y_j$. Higher values of  $S_{i,j}$ indicate a stronger match and are therefore preferred. Once these scores are computed, generating recommendations is straightforward. For a given job seeker $i_0$, vacancies are ranked according to their scores  $S_{i_0,j}$ from the most to the least desirable.  To make $k$ recommendations to $i_0$, an intuitive solution is to pick the $k$ vacancies with the highest scores.\footnote{\label{foot:rank}Throughout the paper, we define the rank of a vacancy as its 
position in this ordering, with rank $1$ corresponding to the highest score.}

\medskip
 
Although this underlying principle is shared across systems, job recommendation approaches vary widely in both the computer science literature and real-world applications. As surveyed by \cite{freire2021recruitment, de2021job, mashayekhi2022challenge}, this diversity reflects a multitude of application contexts, data availability, and algorithmic strategies. 
Importantly, these approaches also differ in the type of information they exploit and in the outcomes they are designed to predict.
Table \ref{tab:recsys} summarizes the main families of approaches and highlights their key characteristics.

\medskip

\textit{Knowledge-based} RSs leverage expert ontologies of occupations, skills, and locations to match workers to vacancies based on assessed fit. A prominent example is WCC ELISE, used by several PESs and private entities.\footnote{See the dedicated websites \href{https://www.roberthalf.com/us/en/find-jobs}{\textcolor{blue}{Robert Half}} and \href{https://www.wcc-group.com/employment/products/elise-job-matching-search-and-match/}{\textcolor{blue}{ELISE}}.} An alternative class of approaches leverages machine learning techniques. \textit{Collaborative filtering} relies on interaction histories to infer similarity patterns, as in the click-based algorithm studied by \cite{roland2022} at the Swedish PES.  By contrast, \textit{content-based} RSs exploit observable characteristics  (\textit{e.g.} occupation, education and skills) to predict interaction probabilities; CareerBuilder provides an example \cite[][]{zhao2021embedding}.

\medskip

\textit{Hybrid} RSs combine these approaches. Examples include the RecSys 2017 winner \cite[][]{Volkovs2017} or LinkedIn's system, which predicts matches based on user and vacancy characteristics, incorporating individual and recruiter-level fixed effects when sufficient interaction data are available \cite[][]{shi2022generalized}. Finally, Indeed's RS \citep[see][]{ma2022jobs} combines collaborative filtering and content-based methods, with a final hybrid stage involving deep learning and a rule-based engine.  

\medskip

Beyond algorithmic design, an important dimension of heterogeneity across RSs concerns the choice of target variable. Table \ref{tab:recsys} summarizes selected contributions from data science and economics, as well as the algorithms studied in this paper. The table highlights substantial variation in outcome measures, ranging from clicks and applications to hires. Notably, data science contributions predominantly focus on intermediate outcomes such as clicks or applications, whereas economic studies typically emphasize hires. Our paper implements an experimental comparison of RSs with different designs within the same setting, evaluated against a common welfare metric derived from the model. This contrasts with the existing literature, where studies typically evaluate one or two pre-specified algorithms targeting different outcomes in different populations, making cross-study comparisons difficult.

\begin{table}
\centering
\caption{Examples of different RSs}
\scalebox{0.7}{
\begin{tabular}{lcccccc}
\toprule
References & Setting & Knowledge & Collaborative  & Content& Target & Type of\\
    &  & based &  Filtering & based & variable &  Recom.  \\
\midrule
WCC Elise & National PESs, & x &  & & &\\
& Robert Half &  &  & & & \\
\cite{zhao2021embedding} & CareerBuilder &  &  & x & Applications  & Vacancies\\ 
\cite{shi2022generalized} & LinkedIn &  & x & x & Applications, ``save" & Vacancies\\ 
\cite{Volkovs2017} & Xing challenge & & x & x & Impressions, clicks&Vacancies \\
\cite{ma2022jobs} & Indeed & x & x & x & Clicks, Applications&Vacancies \\ 
\midrule
\multicolumn{7}{c}{RSs tested in Economics}\\
\midrule
\cite{roland2022} & Swedish PES &  & x & & Clicks& Vacancies \\
\cite{altmann2022direct} & Danish PES &   &  &  x & Hires & Occupations \\
\cite{behaghel2024potential} & French PES & x &  & x & Hires& Firms\\
\cite{belot2025advising} & Dutch PES & x &  &  x & Hires& Occupations\\
\cite{bachli2025helping} & Swiss PES & x  &  &  & Skills profile fit&  Occupations\\
\midrule
\multicolumn{7}{c}{Our tested RSs }\\
\midrule
\Pbs  & French PES &  x &  &  & Utility &Vacancies \\
\V.0 & French PES &  &  & x &   Hires &Vacancies \\
\text{\textsc{XGBoost}} & French PES &  &  & x &   Hires &Vacancies \\
\text{\textsc{Application}} & French PES &  &  & x & Applications  &Vacancies\\
\text{\textsc{Vadore}.2}  & French PES &  &  & x &  Applications, Hires& Vacancies\\
\text{\textsc{Mix}}\sfrac{1}{4}, \text{\textsc{Mix}}\sfrac{1}{2}, \text{\textsc{Mix}}\sfrac{3}{4}   & French PES &  x &  & x & Utility, Hires &Vacancies\\
\MixHalf(\V .2)   & French PES &  x &  & x & Utility, Hires, Appl. & Vacancies\\
\bottomrule
\end{tabular}
}
\label{tab:recsys}
\end{table}

\subsection{A state-of-the-art RS based on hiring predictions}\label{sec:RSP}

We present the baseline ML-based RS we initially developed \cite[see][for architectural details and related literature]{bied2021congestion,vadoreijcai}. This is a state-of-the-art content-based RS building on the insights of the winning algorithm of the RecSys 2017 challenge \citep{Volkovs2017}. In the framework developed in Section~\ref{sec:search}, this algorithm can be interpreted as primarily targeting the probability of successful matching $p(x,y)$. It is trained on data from job seekers who found employment, using the vacancy that led to a hire as the positive example for each job seeker, and a set of vacancies that did not result in a hire as negative examples. The training objective seeks to rank the realized match above the alternatives for each job seeker, it does not directly maximize reemployment rates, but rather learns a score that separates successful from unsuccessful job seeker, vacancy pairs.  

\medskip

Model performance is naturally evaluated using recall@$k$: the share of job seekers $i$ hired a given week for whom the algorithm ranks the realized vacancy among the top  $k$ recommendations available that week, where $k$ is usually 10, 20, 50, or 100. While the recall@$k$ provides a meaningful evaluation metric, it is intractable for direct optimization.  We therefore follow the \textit{learning to rank} literature, and learn a similarity score $S_{i,j}$ between job seeker $i$ and vacancy $j$. The objective is to ensure that, for any job seeker $i$, the score associated with the realized match $j^*(i)$ exceeds that of any alternative vacancy $j'$. This leads to minimizing the \textit{triplet margin loss}  corresponding to the following objective:
\begin{equation}\label{eq:triplet}
    \min_{S} \sum_i \sum_{j' \neq j^*(i)} [S_{i, j^*(i)} - S_{i, j'} + \eta]_+,
\end{equation} 
where $\eta > 0$ is a scalar hyperparameter, $[x]_+ = \max (x, 0)$, the outer sum ranges over all job seekers with matches, the inner one over all vacancies \citep[see, \emph{e.g.},][]{weinberger2009distance}. The objective enforces a separation of at least $\eta$ between the scores of matched and unmatched vacancies for each job seeker.

\medskip

Given job seeker and vacancies characteristics, respectively denoted by $X_i$ and $Y_j$, the score $S_{ij}$ is defined as:
\[ S_{i,j}(X_{i,j}) = \phi(X_i)^{\top} A \psi(Y_j), \]
where $X_{i,j} = (X_i,Y_j)$, $\phi, \psi$ are feed-forward neural networks with several layers, and $A$ is an affinity matrix. Feed-forward neural networks provide flexible, differentiable representations well suited to high-dimensional inputs and large datasets.\footnote{The interested reader may consult \cite{goodfellow2016deeplearningbook} for a textbook treatment.}

\medskip

In this context, $\phi(X_i)$ and $\psi(Y_j)$ can be interpreted as latent representations of job seeker $i$ and vacancy $j$. The matrix $A$ captures cross-dimensional affinities: the parameter $A_{k,l}$ measures the complementarity between dimension $k$ of the job seekers' latent space and dimension $l$ of the vacancy’s latent space. Both latent spaces have dimension 872.   Table \ref{feature_list} lists the observable characteristics used on the job seeker and vacancy sides to predict hires.

Importantly, $\phi$, $\psi$ and $A$ are given a structure which incorporates three main blocks corresponding to geography, skills, and all remaining features. This design explicitly incorporates key dimensions emphasized in the job recommendation and labor market literature, most notably location and (soft) skills \citep[see][]{belot2019providing,belot2025advising,altmann2022direct,bachli2025helping}, while leveraging the power of ML methods to detect the most promising interactions and transitions.

Formally (see Figure \ref{fig:nn}), the similarity score decomposes as 
\[ S_{i,j}(X_{i,j}) = \sum_{b \in \{``geography", ``skills" , ``other features" \} }\phi_{b}(X_{i})^{\top} A_b \psi_b(Y_{j}),
\]
where $\phi_b$, $\psi_b$ and $A_b$ are block-specific transformations and affinity matrices.

\medskip

The parameters that are optimized during training include the neural network weights defining the mappings to the latent spaces, as well as the affinity matrices $A_b$. The resulting non-convex objective is minimized using mini-batch stochastic gradient descent. For computational tractability, non-matching job seeker–vacancy pairs are heavily and uniformly subsampled. Finally, for a given job seeker $i_0$, we define the  $\mathcal{P}$-ranking as the ordering of vacancies induced by the score $S_{i_0,j}$. 

\paragraph{Data used to train algorithms.}

Three features of the data are particularly relevant for this study. First, administrative records on job seekers and vacancies are matched to behavioral data (clicks, applications, and hiring outcomes) on the same platform, providing a joint view of both sides of the market and the interactions between them. Second, we observe click data alongside application data, allowing us to distinguish early-stage interest from actual application decisions; the use of click data to measure job seekers' expressed interest at this early stage of online search is novel in this context. Third, the scale of the data, over 1.1 million job-seeker search sessions, 516,776 unique vacancies, and 75,744 observed hires, enables reliable training of high-dimensional ML models and estimation of the structural application model.

We use rich historical administrative data from the Public Employment Service (PES) to train and evaluate several job recommendation systems (RSs). Table \ref{table:desc_data} summarizes the type of data used in this paper. This includes descriptions of vacancies posted on the PES, job seekers' characteristics and search parameters, as well as user interactions on the PES website, such as clicks, applications, and subsequent hiring outcomes. Our analysis focuses on the former French region of Rhône-Alpes, which offers substantial economic and geographic diversity while remaining sufficiently contained for detailed empirical analysis.

\medskip

The PES website is open to all employers and job seekers and constitutes the largest platform for vacancy postings in the French labor market.\footnote{Using the same application data, \cite{le2021gender} estimate that vacancies posted on this website represented 60\% of all vacancies in 2010 (see their Section VI.A.)} The data provide extensive information on job postings, including publication date, the occupation at several levels of granularity, the posted wage, required experience; contract type (permanent, temporary, or fixed-term); workplace location; weekly working hours; and required qualifications. Importantly, the data contain information on desired hard and soft skills, textual descriptions of both the vacancy and the firm, firm size, the number of applications received by the vacancy and by the firm over the previous six months, and the time elapsed since the vacancy was posted.

\medskip

Information on job seekers is drawn from administrative records on unemployment spells (for example, the \textit{fichier historique}, FH, of the PES). These records include demographic characteristics and detailed job search histories, such as date of registration, geographic location, experience, skills, unemployment duration, applications in the last six months. They also contain various individual and postal code level socio-demographic characteristics. Importantly, we observe job search parameters declared at registration, including reservation wage, maximum commuting time, desired occupation, desired type of contract (temporary vs. long-term), and working hours (full-time vs. part-time). 

\medskip

This comprehensive information on both sides of the market is complemented by detailed data on user behavior on the PES website, notably clicks on vacancies and subsequent applications. While application data from the PES have been used in previous work \citep[see][]{marinescu2021unemployment,glover2019job,algan2020active}, the use of click data to measure job seekers’ interest at the early stages of online search appears to be novel in this context. Applications are observed through three of the PES channels: applications submitted directly by job seekers,  potential matches initiated by firms, and applications suggested by PES caseworkers. Finally, we also exploit the final outcome of these interactions, whether a hire occurs, which is recorded by caseworkers.\footnote{As noted in \cite{algan2020active}, these hiring data are subject to measurement error, notably because hires occurring outside the PES may not be observed. To mitigate this limitation, we complement the PES records with comprehensive administrative data on all hires (\textit{Déclarations préalable à l'embauche}) whenever a hire can be linked to an identifiable PES posted vacancy.} 

\subsection{An RS based on search criteria}\label{sec:sdr}

The French PES relies on a matching algorithm based on WCC Elise to recommend relevant vacancies to job seekers. Like most knowledge-based RSs, this algorithm starts from a comparison between characteristics of the job desired by job seekers and those of the available vacancies and can be interpreted, in the framework developed in Section~\ref{sec:search}, as primarily capturing the utility component of a job seeker–vacancy match. For each characteristic considered, a sub-score is determined, ranging from zero to one, reflecting the degree of compatibility between the job seeker’s preferences or profile and the vacancy’s requirements. The sub-scores are then aggregated to form a global score. Aggregation is primarily based on a weighted average. 

\medskip

For the purpose of this study, we construct a RS inspired by the PES algorithm and based on the same set of criteria. Each criterion is matched exactly with its counterpart on the recruiter’s side (\emph{i.e.}, the requirements specified in the vacancy and the characteristics of the job offered). For each characteristic $k$, we define a consistency measure $c_{k,i,j}\in [0,1]$, which captures the degree to which characteristic $k$ of job seeker $i$ is compatible with that of vacancy $j$.

The characteristics entering the definition of the score and their associated weights are as follows: Occupation (0.332),  Skills in occupation (0.332), Geographic mobility (0.1),  Reservation wage (0.066), Diploma (0.033), Working hours (0.033), Driving license (0.033), Languages (0.033),  Years of experience in occupation (0.033),  Duration and type of contract (0.003). The resulting matching score is defined as:\footnote{\label{foot:U}Throughout the paper, we use the formula from Equation \eqref{eq_PES} together with the weights used by the PES. The score $\mathcal{U}^{PES}$ actually used by the PES and the score \Pbs \ used in our experiments are not identical. The exact criteria used at the PES share the same principles but allow for smoother definitions of several sub-criteria and incorporates additional nonlinearities, such as censoring based on geographic fit. We abstract from these features here in order to preserve transparency and interpretability.}

\begin{equation}
    \label{eq_PES}
\mathcal{U}_{i,j}=\sum_{k=1}^K w_k c_{k,i,j}.
\end{equation}

We refer to the ordering of vacancies for a given job seeker $i_0$ according to the criterion $\mathcal{U}_{i_0,j}$ as the $\mathcal{U}$ ranking. As for the $\mathcal{P}$-based RS, recommendations are the vacancies ranked highest (i.e., with the largest score) under this ordering.

\begin{table}[htbp]
\centering
\footnotesize
\caption{Information used by the search criteria and machine-learning RSs}
\label{feature_list}
\renewcommand{\arraystretch}{1.02}
\setlength{\tabcolsep}{4pt}

\scalebox{0.9}{
\begin{tabular}{>{\raggedright\arraybackslash}p{4cm}
                >{\raggedright\arraybackslash}p{6cm}
                >{\raggedright\arraybackslash}p{7.5cm}}
\toprule
& \textbf{Search criteria RS} & \textbf{Machine-learning RS} \\
\midrule

\textbf{Job seeker data} & & \\

\textit{Skills and qualifications}
& \makecell[l]{Skills; Diploma;\\ Languages}
& \makecell[l]{Skills (\textbf{SVD/embeddings}); Diploma;\\   Languages; \textbf{Soft skills (14)}} \\
\cmidrule{2-3}
\textit{Labor market history}
& \makecell[l]{Experience; Occupation\\ (3-digit)}
& \makecell[l]{Experience; Occupation\\ (1--3 digit)} \\
\cmidrule{2-3}
\textit{Constraints and preferences}
& \makecell[l]{Location; Working hours;\\ Reservation wage and mobility}
& \makecell[l]{Location; Working hours;\\ Reservation wage and mobility (\textbf{history})} \\
\cmidrule{2-3}
\textit{Demographics and search behavior}
& ---
& \makecell[l]{\textbf{Age; Sex; Children; Search type};\\ \textbf{Constraints; Applications}} \\
\cmidrule{2-3} \cmidrule{2-3}
\addlinespace[2pt]
\textbf{Vacancy data} & & \\

\textit{Job requirements}
& \makecell[l]{Skills; Diploma; Experience}
& \makecell[l]{Skills (\textbf{SVD/embeddings}); Diploma;\\  \textbf{Qualification}; \textbf{Soft skills (14)}} \\
\cmidrule{2-3}
\textit{Contract and pay}
& \makecell[l]{Contract type; Working hours; Wage}
& \makecell[l]{Contract type; Working hours\\  Wage; \textbf{Duration}} \\
\cmidrule{2-3}
\textit{Firm and posting characteristics}
& Location
& \makecell[l]{Location; \textbf{Firm size and status; Posting age};\\ \textbf{Applications}} \\
\cmidrule{2-3}
\textit{Textual information}
& ---
& \makecell[l]{Job and firm \textbf{descriptions}} \\
\bottomrule
\end{tabular}
}

\vspace{0.1cm}
\begin{minipage}{0.95\textwidth}
\footnotesize
\textit{Notes:} Differences in terms of data used are highlighted in \textbf{bold}.
\end{minipage}
\end{table}

\subsection{Understanding the two RSs}\label{sec:2RS3}

\textbf{The $\mathcal{U}$ score as a utility signal}. Our previous analysis shows that any recommendation algorithm score can be interpreted as a combination of two components: a signal about utility, $U_{i,j}-U_{0,i}^*(1)$ and a signal about the probability of recruitment, $p_{i,j}$. Consequently, the scores associated with the different RSs generally mix information about both dimensions. We interpret the $\mathcal{U}$ score primarily as a signal of surplus utility $U_{i,j}-U_{0,i}^*(1)$. Salary is the job attribute most directly linked to utility, but the other job characteristics included in the definition of $\mathcal{U}$ such as occupation, required skills, geographic location, working hours, and contract type also plausibly affect job seekers’ utility. Similarly, requirements related to diplomas, experience, driving licenses, and languages describe the set of skills that a job seeker would deploy in the position. Skill matching is not only relevant for productivity but is also associated with job satisfaction, personal fulfillment, and the accumulation of human capital, all of which are valued by job seekers. The weights used to aggregate these characteristics are those set by the PES. Nevertheless, as shown in Appendix \ref{app:modelobservational}, these weights can alternatively be estimated from the data at our disposal using application behavior.

\textbf{The $\mathcal{P}$ score as a match probability signal}. Analogously, we interpret the ranking induced by the $\mathcal{P}$ score as primarily reflecting the probability of recruitment, $p_{i,j}$. Appendix \ref{sec:calibration} provides empirical support for this interpretation. Following the approach of \citet{chernozhukov2018generic}, we use the predictive content of the ML algorithm's score to build a generic best logistic predictor of the matching probability. Specifically, we exploit the \textit{history of sequential job applications} to estimate a model linking the probability of a successful application to the matching score between job seeker $i$ and the vacancies $j(i)$ to which they applied. The results in Table \ref{tab:calibration} strongly validate the association between the score $S_{i,j}$ and recruitment outcomes. This procedure also allows us to map the score $S_{i,j}$ into an estimated probability of success, $p_{i,j}$. As a result, the rankings induced by the $S$-score and $\mathcal{P}$-score, which we call $\mathcal{P}$-ranking, can be interpreted as rankings based on recruitment probabilities.\footnote{This approach provides us with an estimate of $p_{i,j}$ for each application–vacancy pair. This quantity is typically unobserved, but is revealed here through the machine-learning-based estimation.} 

\textbf{Two different scores}. Figure~\ref{fig:ranks} compares the sets of top-ranked vacancies recommended to a given job seeker under the $\mathcal{U}$ and $\mathcal{P}$ rankings. The overlap between the vacancies ranked highest according to each criterion is very limited. On average, the vacancy that maximizes $\mathcal{P}$ is ranked 2{,}027th in the $\mathcal{U}$ ordering, while the vacancy that maximizes $\mathcal{U}$ is ranked 4{,}403rd in the $\mathcal{P}$ ordering. These large rank reversals highlight that the two scores emphasize markedly different
dimensions of job seeker–vacancy matches.

Appendix~\ref{app:comparison} further documents these differences by comparing the distributions of $\mathcal{U}$ and $\mathcal{P}$ among the top-ranked vacancies under each criterion (see Figure~\ref{fig:dist_all}). For example, the median probability of success of a vacancy recommended under the $\mathcal{U}$ ranking is 0.02, compared with 0.06 for a vacancy recommended under the $\mathcal{P}$ ranking.

Taken together, these results support the intuitive interpretation that the $\mathcal{U}$-based RS primarily captures the utility dimension, while the $\mathcal{P}$-based RS primarily captures the probability of successful matching. This distinction is useful for expositional clarity. However, neither score was originally designed to isolate a single dimension.  There is no guarantee that $\mathcal{U}_{i,j}$ perfectly identifies $U_{i,j}-U_{0,i}^*(1)$ and that $\mathcal{P}_{i,j}$ perfectly identifies $p_{i,j}$. More plausibly, both scores are functions of the two underlying components, such that $\mathcal{U}_{i,j}=\mathcal{U}(U_{i,j}-U_{0,i}^*,p_{i,j})$ and $\mathcal{P}_{i,j}=\mathcal{P}(U_{i,j}-U_{0,i}^*,p_{i,j})$. Even if $\mathcal{U}$ predominantly reflects utility and $\mathcal{P}$ predominantly reflects recruitment probabilities, each score is likely to contain information about both dimensions. The framework of Section~\ref{sec:search} identifies the theoretically correct combination, and the experiments of Section~\ref{sec:betatestso} are designed to test whether algorithms that move in this direction, by enriching hiring-based predictions with preference-related signals, deliver welfare gains in practice. Importantly, $\mathcal{U}$ and $\mathcal{P}$ were developed independently of this interpretation and were not designed to be combined; the model provides the basis for doing so.

\FloatBarrier

\section{Using field experiments to design recommender systems}\label{sec:betatestso}

This section reports two randomized field experiments that are distinctive in three respects. First, six algorithms are compared head-to-head in the same setting, using the same population and the same outcome measures, a direct experimental comparison of this breadth does not exist in the prior literature. Second, the set of algorithms spans the full spectrum from welfare-misaligned (pure $\mathcal{P}$ or pure $\mathcal{U}$) to welfare-approximating (\V.2), so the experiments trace the welfare gains from progressively better-aligned designs. Third, the design of the second experiment is derived from the theoretical predictions of the model and guided by the results of the first experiment, as part of an iterative learning cycle.

The first field experiment (beta-test~1) tests whether job seekers respond more favorably to recommendations that combine hiring-based and preference-based rankings than to recommendations based on either dimension alone, the core prediction of Proposition~\ref{prop_V_u0}. Section~\ref{sec:2RS} and Appendix~\ref{app:comparison} establish that the $\mathcal{P}$ and $\mathcal{U}$ rankings differ substantially, so this comparison has empirical bite. The results of this experiment motivate a redesign of the algorithm that more directly incorporates preference information while maintaining hiring prediction as the primary objective.

The second field experiment (beta-test~2) evaluates a new family of algorithms designed in response to these results, including an application-based algorithm (\textsc{Application}) and an enriched hiring-prediction algorithm (\V.2) that incorporates both preference-related and application-related signals. \textsc{Application} plays a dual role: it is a theoretically meaningful benchmark closely related to job seekers’ decision rules, and a building block used to improve hiring predictions in \V.2.

The two field experiments we conducted follow the same protocol, which is described in detail in Appendix \ref{sec:betatests}. Table \ref{table:experiments2}  summarizes these experiments. The eligible population consists of job seekers registered at \PE\ in the Auvergne-Rhône-Alpes region who were actively seeking employment. In each experiment, a single email was sent to a randomly selected sample of job seekers among this population, respectively 102,314 for the March 2022 experiment and 150,000 for the June 2023 one, providing access to a list of job recommendations. Job seekers who clicked the consent link and viewed the list were enrolled, resulting in 18,947 and 30,973 participants, respectively. There is no control group; instead, participants were randomly assigned to recommendation lists generated by different algorithms. The first experiment focuses on combining the two generic algorithms presented in Section \ref{sec:2RS}. The second experiment emphasizes new algorithms developed based on insights from the first experiment. In both experiments, we analyze the $\mathcal{U}$ and $\mathcal{P}$ scores of recommended vacancies, as well as clicks, applications, and hires.

\subsection{Beta-test~1: testing ordinal combinations of $\mathcal{U}$ and $\mathcal{P}$}

\subsubsection{Experimental design and treatments}

The experiment randomized two dimensions of the intervention: (i) the \textit{algorithm used to generate job recommendations}, and (ii) the \textit{display of additional information}. Job seekers were randomly assigned to one of ten treatment arms, corresponding to five algorithmic variants crossed with two display conditions.

\paragraph{Recommendation algorithms}

All recommendations are drawn from a \textit{consideration set}, a pool of vacancies that includes job vacancies highly ranked by at least one of the two base algorithms (\V.0\ or \Pbs). This ensures that the recommended vacancies are relevant according to either match likelihood or stated preferences.

Each job seeker was randomly assigned to one of five algorithms, which differed in how they selected and ranked job vacancies from within this consideration set:\footnote{Independently of the recommendation algorithm, job seekers were also randomly assigned to different information display conditions. Some participants were shown additional performance indicators (star ratings summarizing preference match and predicted hiring probability) for the first two recommended vacancies. In the main analysis, we pool all display conditions. Appendix Table ~\ref{tab:reduced_all} shows that accounting for this variation does not affect our conclusions.}

\begin{itemize}
    \item[-] \textbf{\V.0 (ML-based recommendations):} The RS ranking vacancies using the $\mathcal{P}$ score (detailed in Section~\ref{sec:RSP}), designed to predict successful matches.
    
    \item[-] \textbf{\Pbs \ (Preference-based recommendations):} The RS using the $\mathcal{U}$ score (see Section~\ref{sec:sdr}), which ranks vacancies based on the job seeker's stated preferences. On top of this baseline score, it also includes a final censoring step based on the geographic fit (see footnote \ref{foot:U}).
    
    \item[-] \textbf{\Mix\ algorithms (ordinal combinations of \V.0 and \Pbs):} Three hybrid variants combine the rankings from \V.0\ and \Pbs, relying solely on ordinal information. The steps are as follows (see Appendix~\ref{sec:treatments} for details):
    \begin{enumerate}
        \item Rank vacancies in the consideration set according to \V.0 ($\mathcal{P}$);
        \item Filter vacancies in the consideration set:
        \begin{itemize}
            \item[-] \textbf{\MixQuarter:} Retains the top 25\% according to $\mathcal{P}$ ranking;
            \item[-] \textbf{\MixHalf:} Retains the top 50\% according to $\mathcal{P}$ ranking;
            \item[-] \textbf{\MixThreeQuarter:} Retains the top 75\% according to $\mathcal{P}$ ranking;
            \end{itemize}
        \item Re-rank the filtered vacancies using the $\mathcal{U}$ ranking;
        \item Select the 10 first vacancies in this ranking.
    \end{enumerate}
 
    We refer to these three variants collectively as the \textit{\Mix} group.
\end{itemize}

Clearly, the corresponding algorithms, ordered as \V.0, \MixQuarter, \MixHalf, \MixThreeQuarter, and \Pbs, progressively shift the weight from $\mathcal{P}$ to $\mathcal{U}$ ranking.

\paragraph{Survey protocol and data description.}

Each of the 102,314 job seekers randomly selected for invitation to the experiment was first assigned to one of ten randomization groups. They received an email containing a link to a Qualtrics survey in which job ads were listed. Of the 102,314 individuals invited, 100,879 successfully received the email (the remainder were affected by technical issues), and 18,947 (18.6\%) opened the survey, thereby enrolling in the experiment. Each job seeker was shown the top 10 ads corresponding to their profile according to the algorithm to which they were assigned.

\subsubsection{Reduced-form estimates}

We estimate the following specification at the job seeker--vacancy pair level using ordinary least squares (OLS):

\begin{equation}\label{eq:reduced}
    Y_{ij} =  \sum_{a\in\mathcal{A}_1} \beta_a G_{a,i} + \gamma^{\top} Z_{i} + \epsilon_{ij},
\end{equation}

where \( Y_{ij} \) denotes one of the following outcomes for job seeker \( i \) and vacancy \( j \): the hiring score (\( \mathcal{P} \)), the matching score (\( \mathcal{U} \)), whether the job seeker clicked or applied to the vacancy, or the subjective rating given to the vacancy. The indicator \( G_{a,i} \) denotes assignment of individual \( i \) to algorithm \( a \in \mathcal{A}_1 = \{\text{\textsc{Vadore}.0}, \text{\textsc{Mix}}-\sfrac{1}{4}, \text{\textsc{Mix}}-\sfrac{1}{2}, \text{\textsc{Mix}}-\sfrac{3}{4}, \text{\Pbs}\} \). The vector \( Z_i \) includes a set of indicators for the position of the vacancy in the displayed list (slots 1 to 10). Standard errors are clustered at the job seeker level.

Figure \ref{fig:coef_levels_all_O} (see also Table~\ref{tab:reduced}, taking \Pbs\ as reference) presents the results using data from the 18,947 job seekers randomly assigned to receive job recommendations from one of the five algorithms.\footnote{Preregistered outcomes for this beta-test are: ratings, clicks and applications on recommended vacancies. We also preregistered broader labor market outcomes related to job search but do not use them here.}

Figures \ref{fig:coef_levels_all_O}-(a) and \ref{fig:coef_levels_all_O}-(b) use the full set of 10 job recommendations provided to each participant (for a total of 189,470 observations) to examine how the assigned algorithm affects the distribution of vacancies by predicted hiring probability and matching score. As intended by the experimental design, assignment to algorithms with higher weight on \V.0\ results in higher average hiring scores (\( \mathcal{P} \)) but lower matching scores (\( \mathcal{U} \)). Specifically, compared to vacancies recommended by \Pbs, recommendations from \V.0\ have an average hiring score that is 0.046 points higher (against a baseline mean of 0.054), representing nearly a doubling in expected hiring probability. All estimated coefficients are statistically significant, and we observe a large jump in hiring scores between \MixThreeQuarter\ and \MixHalf. The three algorithms with greater weight on \V.0\ (\V.0, \MixQuarter, and \MixHalf) yield similar outcomes on this dimension.

Conversely, matching scores (\( \mathcal{U} \)) decline markedly as the weight on \V.0\ increases. Vacancies recommended by \V.0\ have an average matching score 0.19 points lower than those recommended by \Pbs\ (baseline: 0.773). The pattern of decreasing \( \mathcal{U} \) scores mirrors the increase in \( \mathcal{P} \) scores, and again, the first three algorithms (\V.0, \MixQuarter, \MixHalf) yield relatively close estimates. 

Figures \ref{fig:coef_levels_all_O}-(c) uses data on subjective ratings provided by job seekers for the first two recommended vacancies (36,668 observations). These ratings range from 0 to 10 and were elicited via the question: \textit{“Overall, what rating out of 10 would you give this job vacancy?”}. In Table~\ref{tab:reduced}, all treatment coefficients are positive and statistically significant, indicating that vacancies recommended by algorithms in \( \mathcal{A}_1\setminus\{\text{\Pbs}\} \) are rated more favorably than those recommended by \Pbs. Importantly, the highest coefficient is not associated with \V.0, but—as anticipated from the model presented in the previous section, with hybrid algorithms that combine the rankings of \Pbs\ and \V.0. The \MixHalf\ algorithm yields the highest average rating.

Figures \ref{fig:coef_levels_all_O}-(d) and \ref{fig:coef_levels_all_O}-(e) present results for clicks and applications using again the full set of ten recommendations per job seeker. Both outcomes are relatively rare, particularly applications. The lowest rates are observed under \Pbs: 4.2\% for clicks and 0.45\% for applications. Click behavior follows a pattern consistent with the subjective ratings: algorithms (\V.0, \MixQuarter, \MixHalf) yield relatively close and higher estimates, the difference with \Pbs \ being positive 
and statistically significant. The \MixHalf\ algorithm again produces the largest increase, with a click rate 0.64 percentage points higher than \Pbs, an increase of approximately 15\%. Application rates are also higher under \V.0\ and \MixHalf, with the former effect being significant at the 10\% level and representing roughly 16\% increases relative to the \Pbs\ baseline. Since we only observed three hires based on the recommendations made in this first experiment, the results regarding the effects on hiring are not significant and are therefore not reported.

\begin{figure}[!htbp]
    \centering

    \begin{minipage}{0.43\textwidth}
        \centering
        \includegraphics[width=\linewidth]{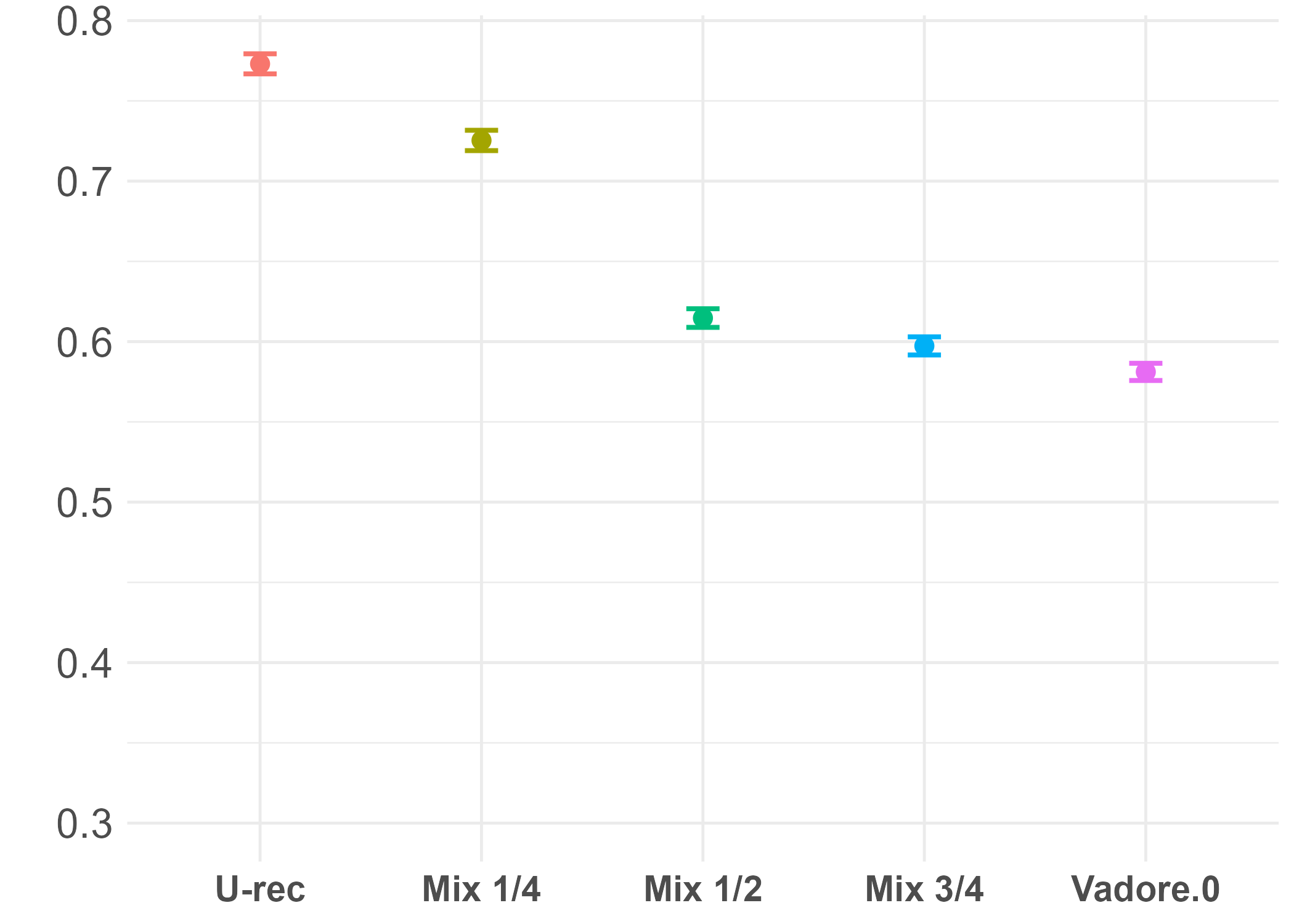}\\
        {\footnotesize F-test: F=899.6, p=0.000} \\
        {\small (a) $\mathcal{U}$}
    \end{minipage}\hfill
    \begin{minipage}{0.43\textwidth}
        \centering
        \includegraphics[width=\linewidth]{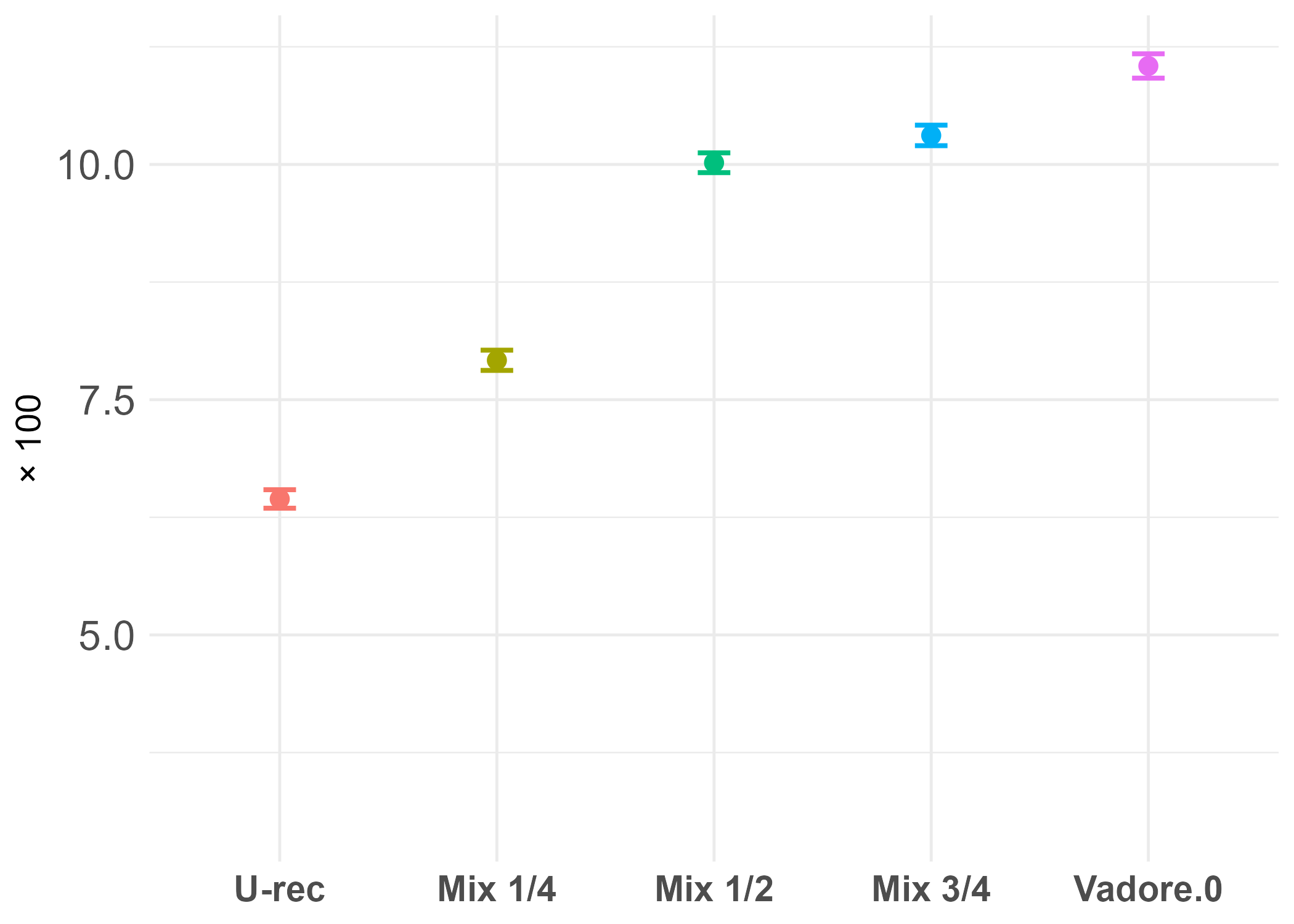}\\
                {\footnotesize F-test: F=1423.1, p=0.000} \\
        {\small (b) $\mathcal{P}$}
    \end{minipage}

    \vspace{0.35cm}

    \begin{minipage}{0.43\textwidth}
        \centering
        \includegraphics[width=\linewidth]{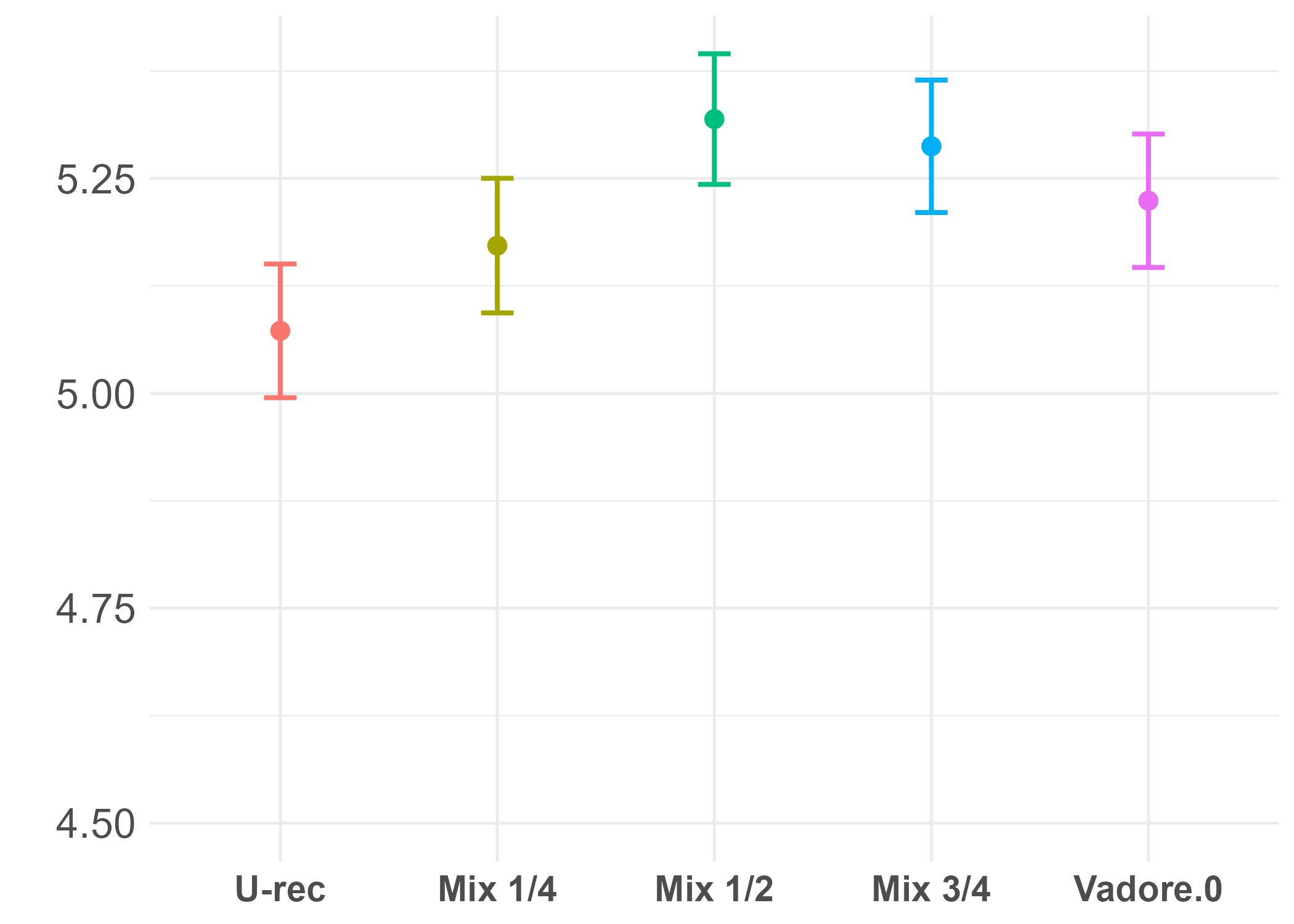}\\
                {\footnotesize F-test: F=6.73, p=0.000} \\
        {\small (c) Rating}
    \end{minipage}\hfill
    \begin{minipage}{0.43\textwidth}
        \centering
        \includegraphics[width=\linewidth]{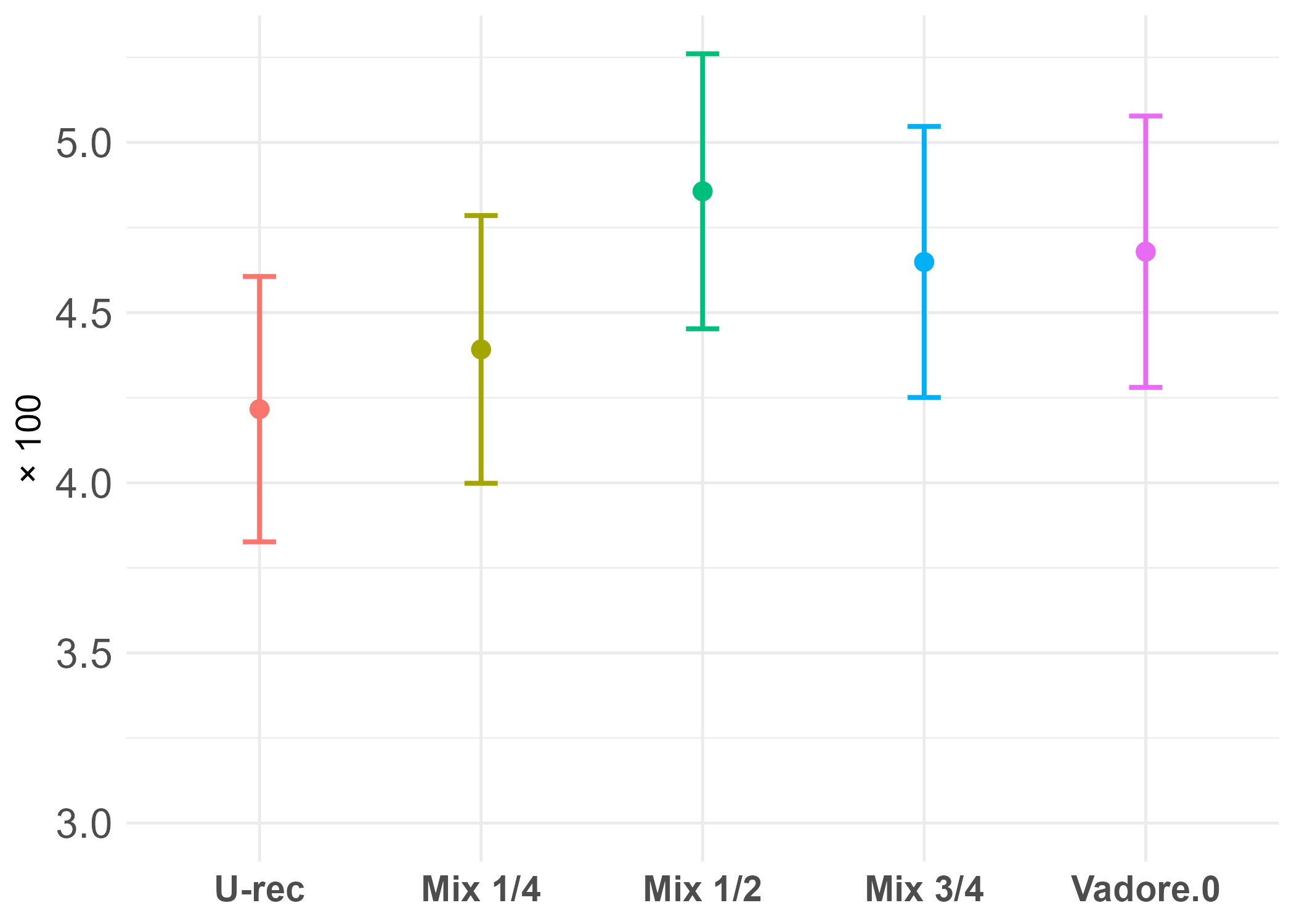}\\
                {\footnotesize F-test: F=2.917, p=0.02} \\
        {\small (d) Clicks}
    \end{minipage}

    \vspace{0.35cm}

    \begin{minipage}{0.43\textwidth}
        \centering
        \includegraphics[width=\linewidth]{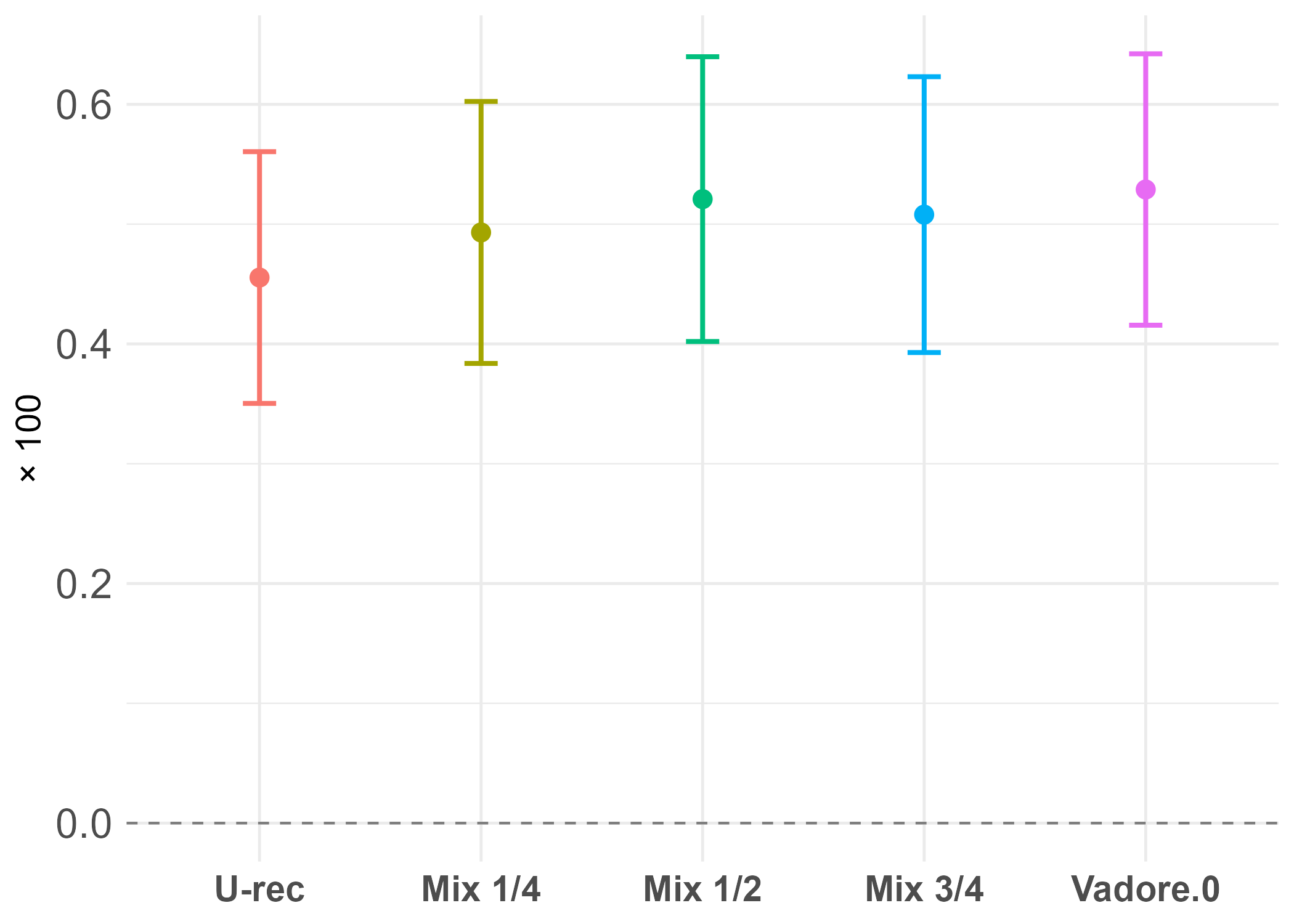}\\
                {\footnotesize F-test: F=1.123, p=0.344} \\
        {\small (e) Applications}
    \end{minipage}\hfill

    \caption{Beta-test~1: Effects of algorithm assignment on vacancy characteristics and job-seeker responses}
    \label{fig:coef_levels_all_O}

    \vspace{0.2cm}
    \begin{minipage}{0.95\textwidth}
        \footnotesize
        \textit{Notes:} Each panel reports coefficient estimates from separate regressions of the indicated outcome on recommendation-treatment indicators, controlling for ad fixed effects. Points represent estimated coefficients and bars denote 95\% confidence intervals. Standard errors are clustered at the job-seeker level. Outcomes $\mathcal{P}$, Clicks, Applications, and Hirings are rescaled as indicated in the figure labels.  F-tests correspond to the F-statistic and p-value from a joint test of significance of all estimated coefficients, with \Pbs\ taken as the reference group.
    \end{minipage}
\end{figure}

\subsection{Designing new algorithms based on the experimental results}

The key finding of beta-test~1 is that the hybrid algorithm \MixHalf, which combines both $\mathcal{P}$ and $\mathcal{U}$ rankings, outperforms the two pure strategies on ratings and clicks. This is the empirical counterpart of Proposition~\ref{prop_V_u0}: since $\Delta(p,U)$ and $\Gamma(p,U)$ are not ordinally equivalent, neither utility-based nor hiring-based rankings alone suffice for welfare-optimal recommendations. These results therefore motivate the exploration of additional recommendation principles consistent with the model.

Building on this insight, we consider two complementary directions for extending the initial recommendation scores.
First, the model highlights application behavior as a key behavioral object.
As shown in Section~\ref{sec:search}, application decisions identify the surplus index $\Delta(p,U)$ that governs job seekers’ choices.
A recommendation score based on predicted applications therefore constitutes a theoretically meaningful benchmark, even though it does not coincide with the welfare-relevant score $\Gamma$.
We accordingly train a RS that predicts applications rather than hirings, which we denote \textsc{Application}.
Beyond its role as a benchmark, this score provides an empirical proxy for perceived job utility, which is not directly observed.

Second, the model implies that welfare-relevant rankings depend not only on hiring probabilities but also on job seekers’ preferences.
To operationalize this dimension within a hiring-based RS, we enrich the original hiring-prediction algorithm by incorporating the preference components $c_{k,i,j}$ used in the $\mathcal{U}$ score (see Equation~\eqref{eq_PES}) as additional predictors.
This modification yields an intermediate version of the algorithm, denoted \V.1.
We then combine these two extensions by introducing the application-based score as an additional input into the hiring-prediction architecture used for \V.1.
This results in the final algorithm, \V.2, which augments hiring predictions with information on both job seekers’ preferences and application behavior.
Figure~\ref{fig:nn1} schematically illustrates how these components are integrated within the algorithm.

Importantly, the algorithms evaluated in the second field experiment also allow us to assess empirically a natural approximation of the welfare-relevant score characterized in Section~\ref{sec:search}.
As shown in the model, the optimal score can be written as $\Gamma(p,U)=p_h(p,U)\times m\!\big(p_a(p,U)\big)$, where the function $m(\cdot)$ is increasing in the application probability under fairly general conditions.
While $p_a(p,U)$ is not directly observed, both application-based predictions and preference-based scores provide empirical proxies that are positively related to job seekers’ propensity to apply.

From this perspective, combining a hiring-based score with additional information on job utility constitutes a conceptually grounded way of enriching $p_h(p,U)$ in the direction suggested by the model.
In particular, mixing the hiring score produced by \V.2\ with the $\mathcal{U}$ score amounts to testing whether reinforcing the utility content of a hiring-based RS improves performance, in line with the monotonicity properties implied by the model.
Although alternative combinations, such as directly mixing $p_h$ with predicted application probabilities, would also be consistent with this logic, the mixtures considered here allow us to assess whether incorporating job utility into hiring-based recommendations moves the algorithm in the direction predicted by the theory.

\subsection{Beta-test~2: evaluating the model-guided algorithm family}
\subsubsection{Experimental design and treatments}
The second field experiment is designed to evaluate the relative performance of the new algorithms developed in response to the results of the first beta-test.
Its objective is twofold: first, to compare these newly designed algorithms against one another; second, to benchmark them against the two reference algorithms based on hiring predictions ($\mathcal{P}$) and stated preferences ($\mathcal{U}$) that motivated the initial analysis.

Among the six algorithms tested in this second experiment, two were already included in the March~2022 study and serve as benchmarks: \Pbs\ and \V.0.\footnote{Although their architecture is unchanged, both algorithms were retrained using more recent data to reflect labor market conditions prevailing at the time of the June~2023 experiment.}
They are complemented by four additional algorithms that reflect different ways of incorporating preference-related information into hiring-based recommendations.
These include the enhanced hiring-based algorithm \V.2, which integrates the design improvements described above; the \textsc{Application} algorithm, which predicts application behavior and serves both as a theoretically meaningful benchmark and as an input into \V.2; and a hybrid ranking that combines \V.2\ and \Pbs, denoted $\text{\textsc{Mix}}\sfrac{1}{2}(\text{\textsc{Vadore}.2})$, constructed in the same spirit as the mixtures tested in the first experiment.
Finally, we include \textsc{XGBoost}, a benchmark algorithm based on gradient boosting that mirrors the objective of \V.0\ (predicting hires) but replaces neural networks with a standard gradient-boosting architecture. Its purpose is to disentangle the effect of the recommendation \emph{objective} from the effect of the \emph{architecture}: since \V.0\ and \textsc{XGBoost} target the same outcome but differ in their ML method, any performance difference between them reflects architectural choices rather than the welfare alignment of the objective. Conversely, differences between \textsc{XGBoost} and the welfare-enriched algorithms (\V.2, \textsc{Application}) reflect objective alignment rather than architecture.

The full set of algorithms evaluated in this second experiment is therefore:\footnote{
In addition to algorithmic variation, the experiment also randomized the type of information displayed alongside recommendations for a subset of algorithms (including \V.2, \textsc{Application}, and $\text{\textsc{Mix}}\sfrac{1}{2}(\text{\textsc{Vadore}.2})$), resulting in multiple display conditions.
As in the first experiment, our baseline analysis abstracts from this dimension and aggregates all display variants for a given algorithm.
Appendix Table~\ref{tab:reduced3} shows that accounting explicitly for display variation does not affect our main conclusions.
} $$\mathcal{A}_2
=
\{
\text{\textsc{Vadore}.0},\ 
\text{\textsc{Vadore}.2},\ 
\text{\textsc{Mix}}\sfrac{1}{2}(\text{\textsc{Vadore}.2}),\ 
\text{\textsc{Application}},\ 
\text{\textsc{XGBoost}},\ 
\mathcal{U}\text{-}\textsc{rec}
\}.$$
The experiment closely replicated the design of the 2022 study. As in the earlier experiment, job seekers were invited by email and enrollment was conditional on clicking a consent button. The survey templates were kept very similar to those previously used. Conducted in June 2023, the experiment invited 150,000 job seekers, of whom 30,973 were enrolled.

The reduced-form analysis follows the same methodology as in the first experiment, as specified in Equation~\eqref{eq:reduced}.  
The results for \V.0\ and \Pbs, presented in Figure~\ref{fig:coef_levels_all_D} (see also Table~\ref{tab:reduced2}, taking the group receiving recommendations from the \Pbs \ as  reference) are consistent with those observed in the previous experiment.\footnote{Preregistered outcomes for this beta-test are: ratings, clicks, applications and hirings on recommended vacancies.} \V.0\ outperforms \Pbs\ in terms of hiring probability (\( \mathcal{P} \)), but, as expected by construction, underperforms in terms of the adequacy score (\( \mathcal{U} \)). It performs marginally better than \Pbs\ in terms of clicks, but not in terms of applications.

The most notable changes are observed in click-through rates and application rates. Relative to the benchmark algorithms \Pbs \ and \V.0, the gains are considerable. Application rates for the \V.2 and \textsc{Application} recommendations are approximately twice as high as those for \Pbs. The improvements are less dramatic for \MixHalf\ and \textsc{XGBoost}, but still meaningful.

Figure~\ref{fig:coef_levels_all_D}-(f) (see also column (6) of  Table \ref{tab:reduced2}) reports hiring outcomes on recommended vacancies. The baseline hiring rate in the reference group is very low: 0.42 \textpertenthousand. We detect a modest signal for the new \MixHalf(\V.2) algorithm, significant at the 10\% level: although the absolute hiring rate remains small, it is roughly three times higher than in the control group \Pbs. Except this group, no statistically significant differences are observed. This is not particularly surprising: despite their differences, all algorithms yield low application rates—below 1\%—and, although the success rates of applications vary substantially across algorithms, they are capped around 7\%.

To inform expectations about the effects of scaling these interventions, the final column of Table \ref{tab:reduced2} reports hiring rates conditional on application, that is, the ratio of hires to applications among recommended vacancies.
\MixHalf(\V.2) and \textsc{XGBoost} nearly double this efficiency relative to the other algorithms, including \textsc{Application}.

These low hiring rates on the recommendations are not unexpected and are consistent with the model. Since application probabilities are small (below 1\% per recommendation), the joint hiring probability $p_h = p \times p_a$ is extremely low. Detecting welfare differences through hiring rates alone would require samples several orders of magnitude larger than a beta-test; the intermediate outcomes such as click-through rates, application rates, and the hiring rate conditional on application, are therefore the relevant margins for comparing recommendation principles at this scale. Accordingly, the primary objective of these experiments is not to estimate employment effects at scale, but to compare recommendation principles and generate the variation needed to estimate the structural model in Section~\ref{sec:sec6}.

\begin{figure}[!htbp]
    \centering

    \begin{minipage}{0.43\textwidth}
        \centering
        \includegraphics[width=\linewidth]{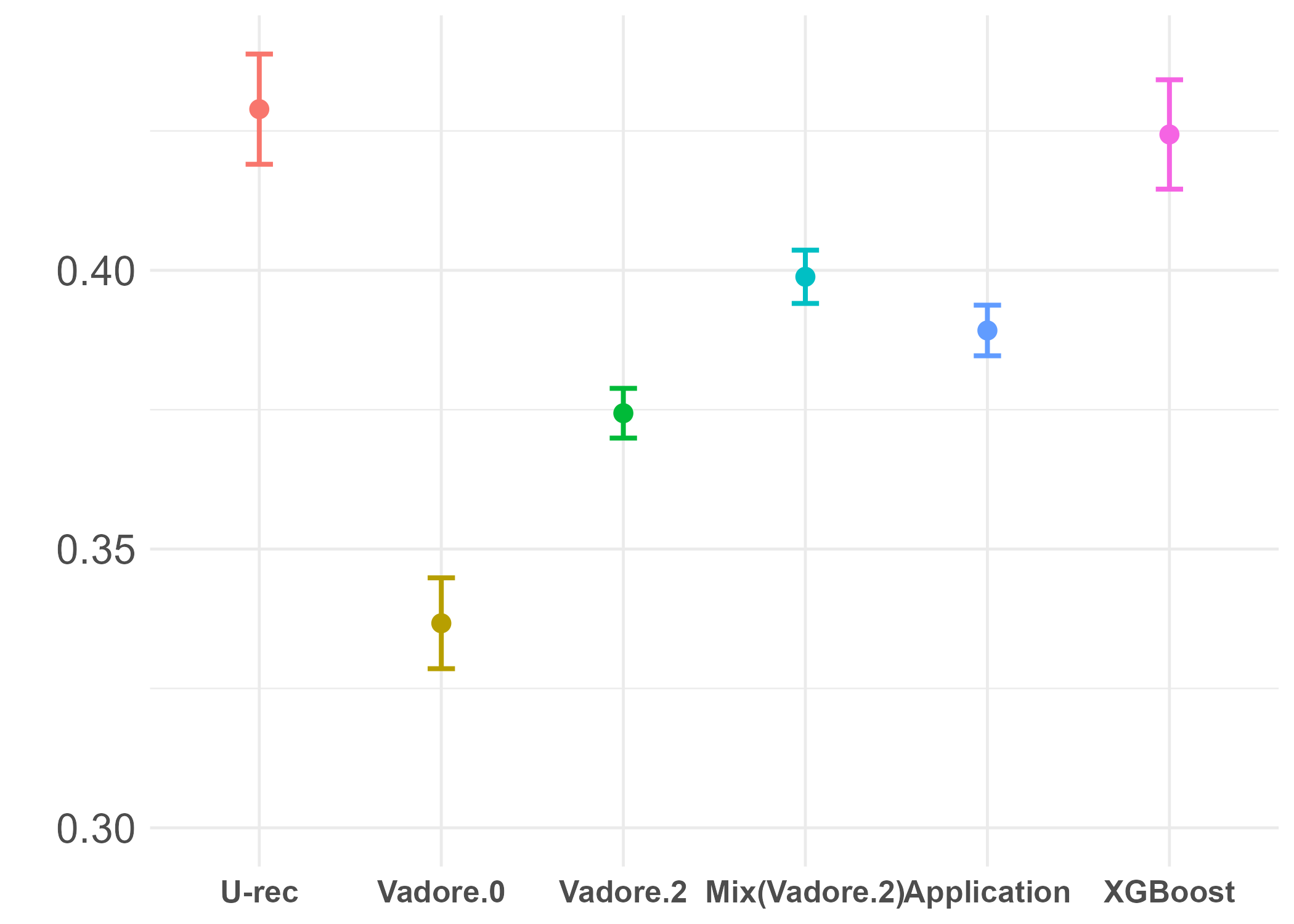}\\
         {\footnotesize F-test: F=73.1, p=0.000} \\
        {\small (a) $\mathcal{U}$}
    \end{minipage}\hfill
    \begin{minipage}{0.43\textwidth}
        \centering
        \includegraphics[width=\linewidth]{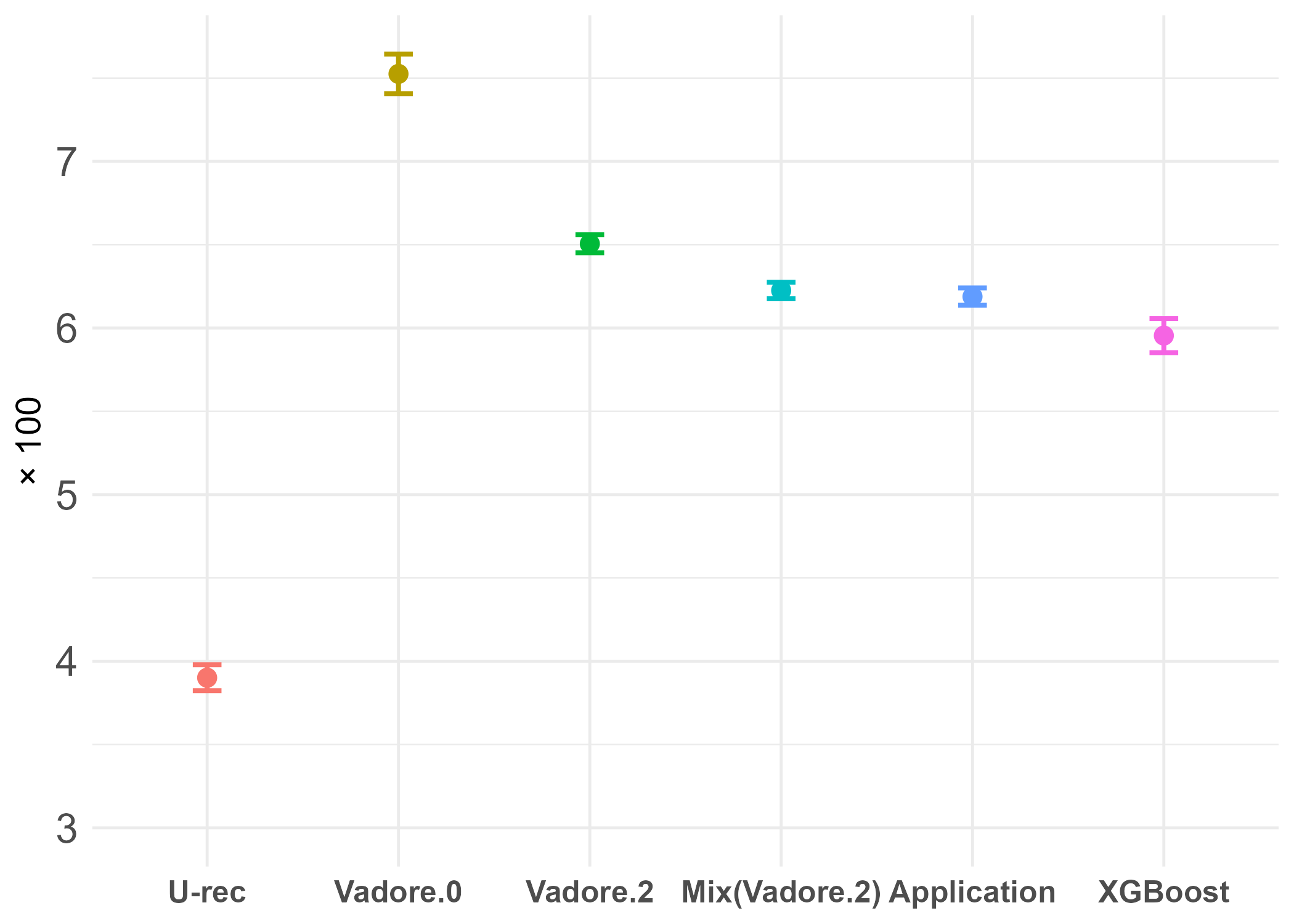}\\
         {\footnotesize F-test: F=805.8, p=0.000} \\
        {\small (b) $\mathcal{P}$}
    \end{minipage}

    \vspace{0.35cm}

    \begin{minipage}{0.43\textwidth}
        \centering
        \includegraphics[width=\linewidth]{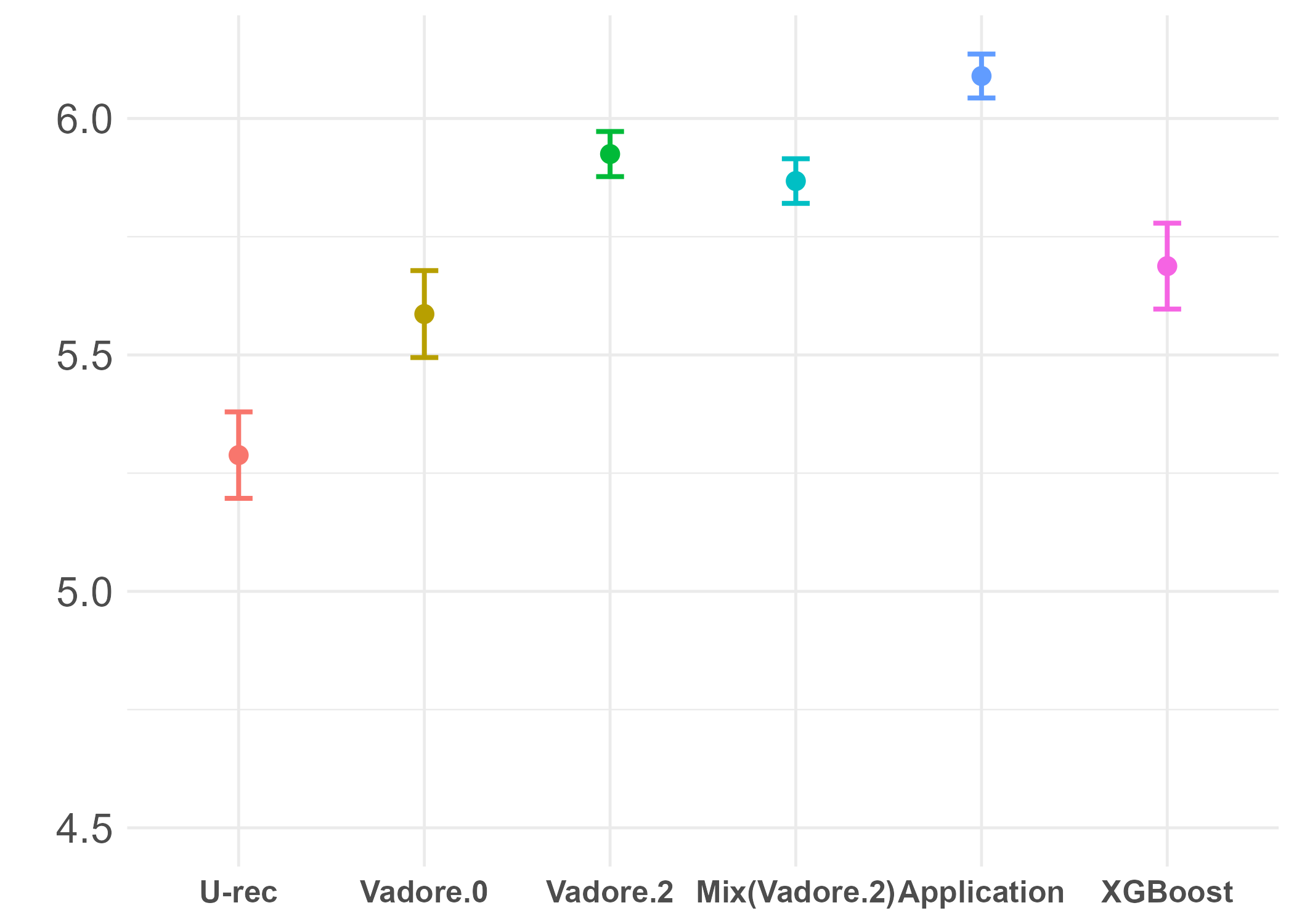}\\
         {\footnotesize F-test: F=65.7, p=0.000} \\
        {\small (c) Rating}
    \end{minipage}\hfill
    \begin{minipage}{0.43\textwidth}
        \centering
        \includegraphics[width=\linewidth]{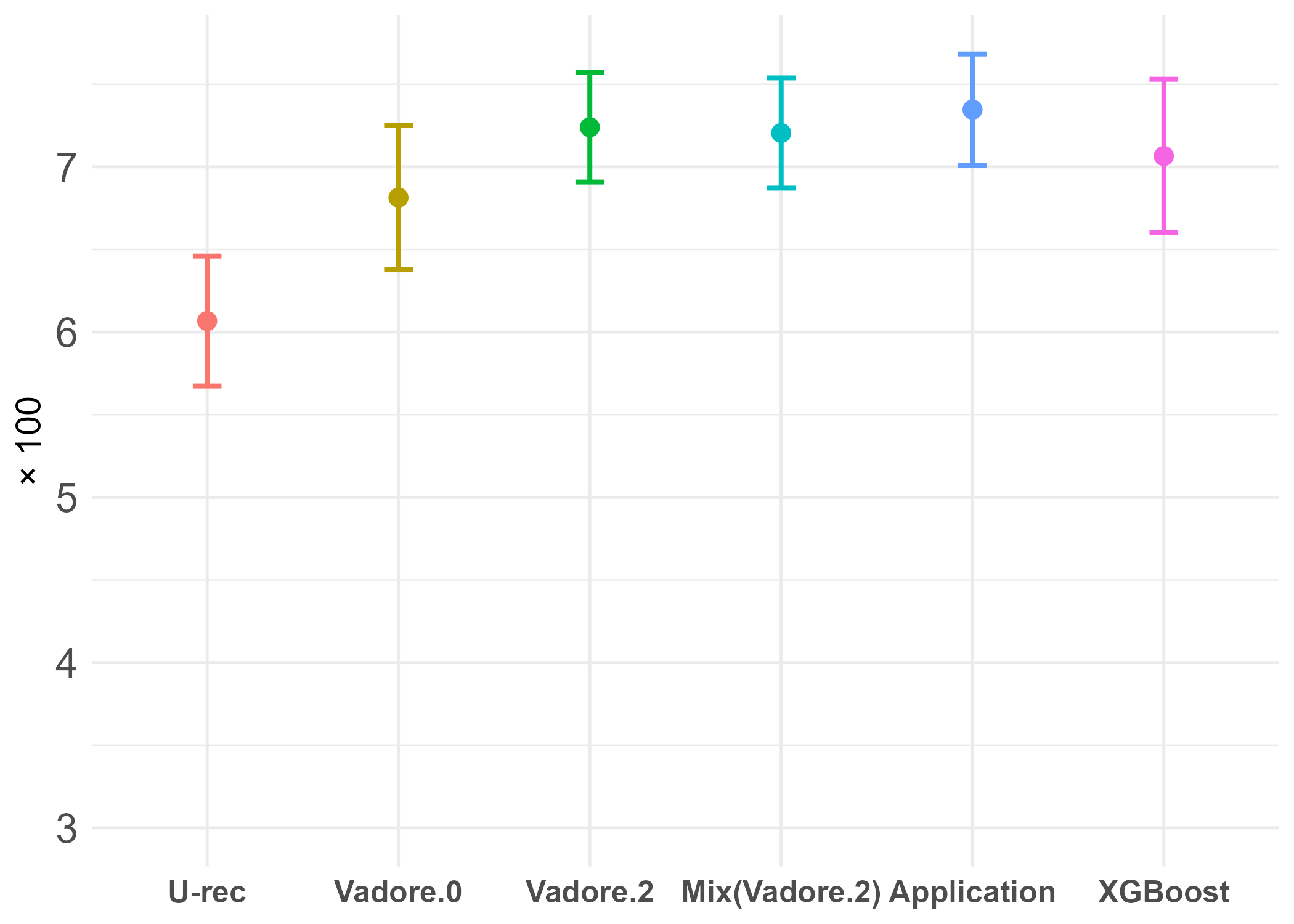}\\
         {\footnotesize F-test: F=11.9, p=0.000} \\
        {\small (d) Clicks}
    \end{minipage}

    \vspace{0.35cm}

    \begin{minipage}{0.43\textwidth}
        \centering
        \includegraphics[width=\linewidth]{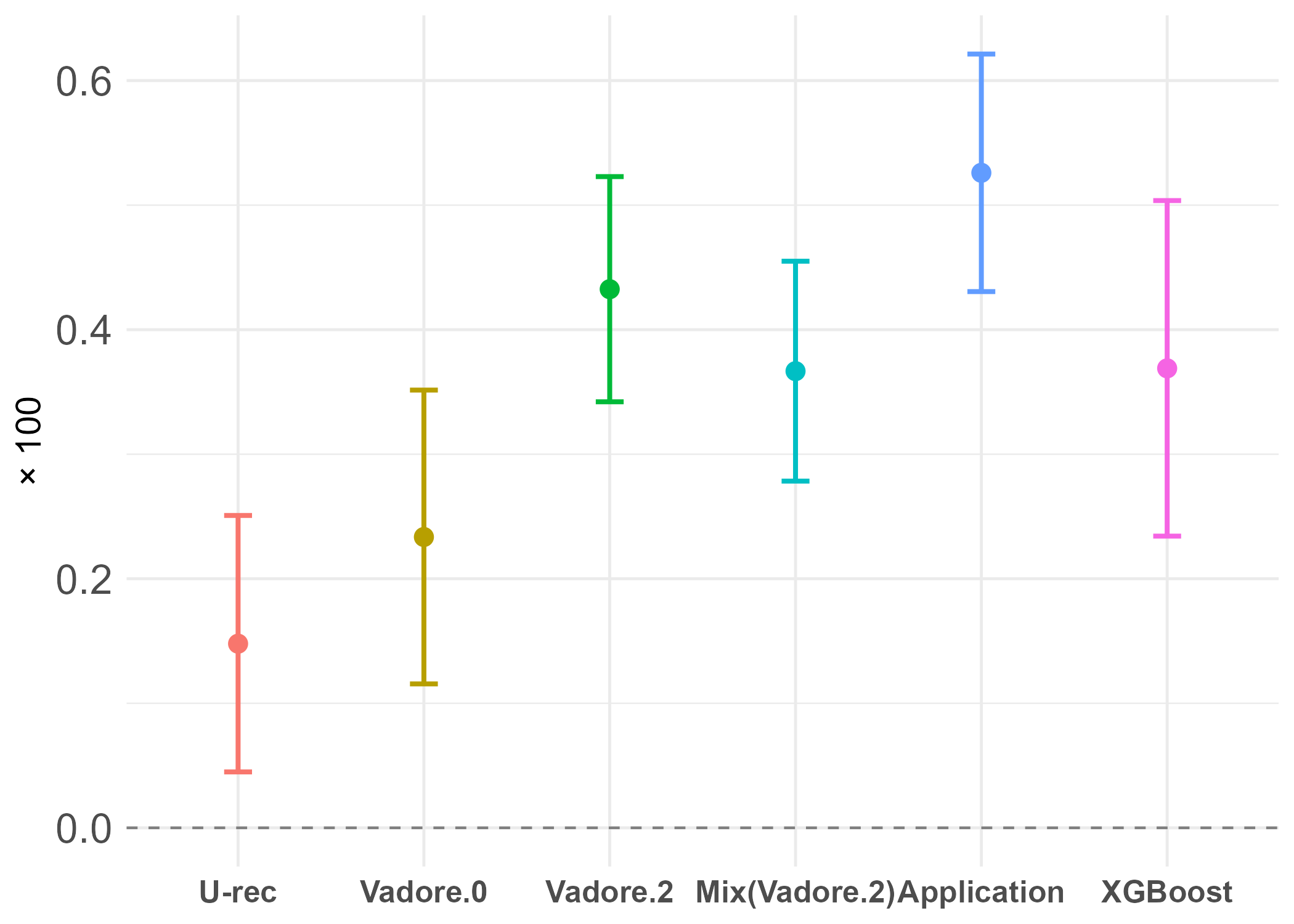}\\
         {\footnotesize F-test: F=11.3, p=0.000} \\
        {\small (e) Applications}
    \end{minipage}\hfill
    \begin{minipage}{0.43\textwidth}
        \centering
        \includegraphics[width=\linewidth]{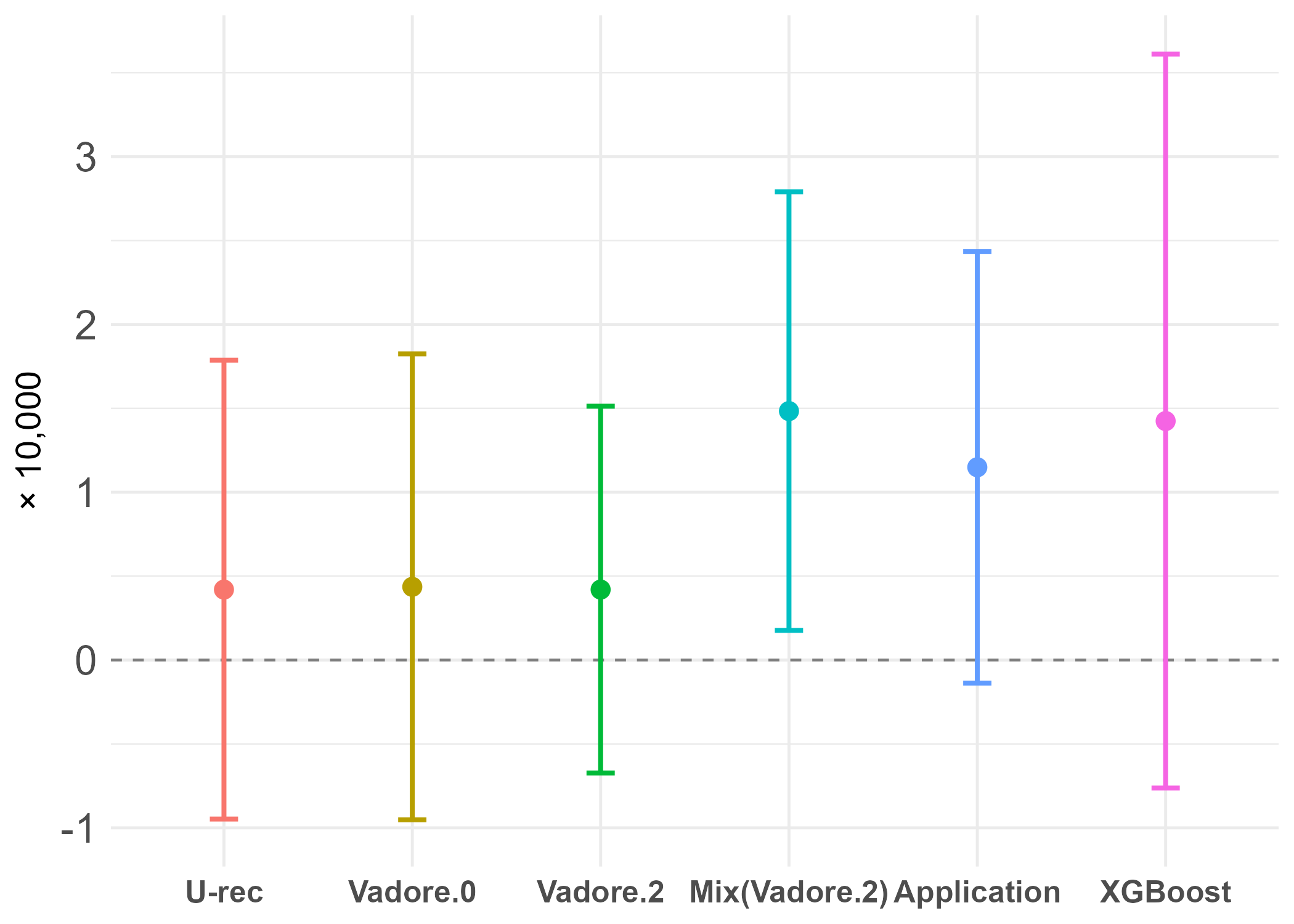}\\
         {\footnotesize F-test: F=1.465, p=0.197} \\
        {\small (f) Hirings}
    \end{minipage}
    \caption{Beta-test~2: Effects of algorithm assignment on vacancy characteristics and job-seeker responses}
    \label{fig:coef_levels_all_D}

    \vspace{0.2cm}
    \begin{minipage}{0.95\textwidth}
        \footnotesize
        \textit{Notes:} Each panel reports coefficient estimates from separate regressions of the indicated outcome on recommendation-treatment indicators, controlling for ad fixed effects. Points represent estimated coefficients and bars denote 95\% confidence intervals. Standard errors are clustered at the job-seeker level. Outcomes $\mathcal{P}$, Clicks, Applications, and Hirings are rescaled as indicated in the figure labels. F-tests correspond to the F-statistic and p-value from a joint test of significance of all estimated coefficients, with \Pbs\ taken as the reference group.
    \end{minipage}
\end{figure}

\section{Estimation of the search model}\label{sec:sec6}

This section uses the experimental variation generated by the random assignment of recommendation algorithms to serve two purposes. First, we estimate a structural model of application behavior and test whether the core behavioral predictions of the model, in particular that both utility and hiring probabilities shape application decisions, are consistent with the data. Second, we use the estimated model to construct the welfare metric $\Gamma(p,U)$ from Proposition~\ref{prop:myopic} and compare all tested recommendation rules against this common benchmark. 

\subsection{Preferences estimation using application behavior}\label{sec:empiric}

The model in Section \ref{sec:search} characterizes how job seekers’ application decisions depend on both the perceived utility of vacancies and their probability of success. In this section, we use observations on job seekers' applications to estimate these preferences.  As in \cite{hitsch2010matching}\footnote{See for example their Equation (9).} or \cite{le2021gender}, given the threshold based decision rule in Equation \eqref{eq:appli}, these preferences can be estimated using a discrete choice model.

We primarily rely on the experimental data introduced in Section \ref{sec:betatestso}, but also provide estimates based on observational data in the Appendix as a robustness check. In the experiments, for each job seeker in the sample, we observe:
\begin{itemize}
\item[-] The group corresponding to the algorithm used to generate the 10 recommendations: either one of the mixtures used in Experiment 1, or one of the six algorithms used in  Experiment 2. We denote the associated group dummy variables by $T_i$;
\item[-] The list of 10 vacancies selected by the corresponding algorithm;
\item[-] The scores $\mathcal{U}_{i,j}$ and success probabilities $\mathcal{P}_{i,j}$; 
\item[-] The clicks and applications for each of the 10 vacancies.
\end{itemize}

We closely follow Equation \eqref{eq:identification} in Proposition \ref{prop:identification}. We estimate this structural model of  application behavior by instrumenting $\mathcal{U}_{i,j}$ and $1/\mathcal{P}_{i,j}$ using the assignment variables $T_i$. Our binary choice model takes the following general form:
\begin{align}
\mathbb{P}\left( C_{i,j} =1 \mid W_{i,j},T_i,c_i \right) & = \Lambda\left(\alpha \mathcal{U}_{i,j}- \frac{\beta}{\mathcal{P}_{i,j}}+\delta B_{2,i} +\gamma+c_i \right), \label{eq:app2}
\end{align}
where $B_{2,i}\in\{0,1\}$ is an indicator for participation in Experiment~2, and the vector $W_{i,j}=(\mathcal{U}_{i,j},\, 1/\mathcal{P}_{i,j},\, B_{2,i})$ collects the utility score, the inverse hiring probability, and this experiment indicator.
This specification is in the spirit of \cite{hitsch2010matching,chen2023reducing}, who estimate similar equations in the context of marriage markets. However, we leverage here the randomization performed in our experiments. Note that the term $-\beta/\mathcal{P}_{i,j}$ enters with a negative sign in \eqref{eq:app2}, so $\beta>0$ implies a positive effect of $\mathcal{P}$ on the application probability, as expected. 
The individual effect $c_i$ captures heterogeneity across job seekers. As discussed  in Section~\ref{subsubsec:identify_prop}, this  effect captures systematic individual specific differences between the observed index $\mathcal{U}$ and the actual utility gain relative to a reservation value.

The functional form used in \eqref{eq:app2} corresponds to a logit model. However, logit models with fixed effects are not easily compatible with the transparent identification structure provided by random assignment. We would also like to account for potential measurement error in $W_{i,j}$: following the standard model $W_{i,j} = W^*_{i,j} + e_{i,j}$, where $W^*_{i,j}$ is the true latent regressor and $e_{i,j}$ is error, the instruments $T_i$ allow consistent estimation under classical measurement error assumptions.

Given the low application probability (3\textperthousand), we adopt the approximation $\Lambda(x) \approx \exp(x)$ and estimate a Poisson IV model. Under this specification, the conditional expectation becomes:
\begin{equation}\label{eq:struct}
\mathbb{P}\left(C_{i,j} \mid W_{i,j}, T_i, c_i, e_{i,j}\right) = \exp(W_{i,j}^{\top}\theta + \gamma + \underbrace{c_i - e_{i,j}\theta}_{\mu_{i,j}}),
\end{equation}
where, from Equation \eqref{eq:app2}, $\theta = (\alpha, -\beta, \delta)$. The control function approach (see \citet{wooldridge2010econometric}) can be used to address endogeneity.\footnote{\label{foot:CF}In a nutshell, the control function method in this context works as follows. The potentially endogenous regressors $W_{i,j}$ are linked to instruments $T_i$ through a first-stage equation:
\begin{equation}\label{eq:1stage}
W_{i,j} = T_i\Pi + v_{i,j}
\end{equation}
and the structural error term $\mu_{i,j}$ is modeled as depending on $v_{i,j}$ but not on $T_i$ (exclusion restriction):
\begin{equation}\label{eq:excl}
\mu_{i,j} = v_{i,j}^{\top}\rho + \nu_{i,j}
\end{equation}
with $\nu_{i,j}$ independent of $v_{i,j}$. Substituting \eqref{eq:excl} into \eqref{eq:struct} and integrating over the distribution of $\nu_{i,j}$ yields the control function moment condition:
\begin{equation}\label{eq:CF}
\mathbb{E}(C_{i,j} - \exp(W_{i,j}^{\top}\theta + \tilde{\gamma} + v_{i,j}^{\top}\rho) \mid W_{i,j}, T_i ) = 0.
\end{equation}

In practice,  the first-stage model \eqref{eq:1stage} is estimated, from which the residuals $\hat{v}_{i,j}$ are computed, and substituted  into Equation \eqref{eq:CF}. Notably, this approach also provides an estimate of $\rho$, which captures the correlation between the structural error term $\mu_{i,j}$ and the endogenous variables $W_{i,j}$.}

An alternative approach is to use a first-order Taylor expansion of $\Lambda(x)$ around the sample mean $\overline{x}$:  
$\Lambda(x) \approx \Lambda(\overline{x}) + \Lambda'(\overline{x})(x - \overline{x}) = \tilde{\Lambda} + \tilde{\Lambda}' x$. This leads to a linearized model of the form:
\begin{equation}\label{eq:structapp}
\mathbb{E}(C_{i,j} \mid W_{i,j}, T_i, \mu_{i,j}) \approx \tilde{\Lambda} + \tilde{\Lambda}'(W_{i,j}^{\top}\theta + \gamma + \mu_{i,j}) = W_{i,j}^{\top}\tilde{\theta} + \tilde{\gamma} + \tilde{\mu}_{i,j}.
\end{equation}
This equation can be estimated either by instrumental variables or using the control function approach described above. In the linear context both methods yield exactly the same results. Note that, due to the linear approximation, coefficients are identified up to a scaling factor, which is not problematic for our purposes. We are primarily interested in the sign and statistical significance of the coefficients on $\mathcal{U}$ and $1/\mathcal{P}$, and the ratio of the two coefficients (see Proposition \ref{prop:identification}). 


The results are presented in Table~\ref{tab:IVPoisson}. Columns (1) and (2) report estimates for the linear approximation in Equation \eqref{eq:structapp}. Column (1) reports results using OLS, ignoring potential endogeneity and column (2) the estimates when instead using random RS-assignment as instruments. Column (3) reports the estimates for the exponential model of Equation \eqref{eq:struct} using the control function method. Column (4) provides estimates of the Average Marginal Effects and can thus be compared more easily to columns (1) and (2). 

All columns provide evidence consistent with the theoretical model. The two key variables of the model, $\mathcal{U}$ and $1/\mathcal{P}$, are both highly significant, with the coefficient on $1/\mathcal{P}$ being, as expected, negative, indicating a positive relationship between the probability of success and the decision to click or apply. Comparing columns (1) and (2) shows that the coefficient of $\mathcal{U}$ is not changed when using instrumental variables rather than OLS but that the coefficient of $1/\mathcal{P}$ is reduced by almost 40\% in column (1) ignoring endogeneity compared to column (2). Indeed, when looking at the bottom panel of the table, reporting the $\rho$'s of the Control Function method (see footnote \ref{foot:CF}), we observe that the covariance between the residuals and the dependent variables is non-significant for $\mathcal{U}$, but significant and positive for $1/\mathcal{P}$.\footnote{If we use the benchmark error in variable model $y=a+bx^*+u$, a measurement model $x=x^*+e$ and an instrumental first stage regression $x=\alpha z+\varepsilon+e$ with $u$, $e$ and $\varepsilon$ uncorrelated and $u$ and $\varepsilon+e$ uncorrelated with $z$; if  we assume errors in variable is the only source of endogeneity; then, in equation \ref{eq:excl}, $\mu$ is equal to $u-be$ and $v$ to $\varepsilon+e$. Thus $\rho=-b\sigma^2_e/(\sigma^2_\varepsilon+\sigma^2_e)$. This would lead to a share of variance of the error $\sigma^2_e/(\sigma^2_\varepsilon+\sigma^2_e)$ of $0.008/0.014= 0.57$ for column (2) and $0.029/0.068=0.43 $ for column (3).}  It is also worth highlighting the consistency of the results. As is standard in discrete choice models, the key quantity is the ratio of the coefficients. When considering the ratio between the coefficient on $1/\mathcal{P}$ and that on $\mathcal{U}$, we obtain very similar values across columns: 0.016 for column (2)  and 0.024 for columns (3) and (4).

As stressed above, these results are especially important because they support the interpretation outlined at the beginning of Section~\ref{sec:search}. They are consistent with viewing the score $\mathcal{U}_{i,j}$ as a signal of the utility gap $U-U^*$ and the probability $\mathcal{P}_{i,j}$ as a signal of the likelihood of success of an application. Taken together, these findings reinforce the idea that both dimensions are relevant inputs for the design of welfare-improving RSs.

\begin{table}[ht]
	\centering
	\caption{Estimation of the application model on job postings, pooling both experiments}
	\label{tab:IVPoisson}
    \scalebox{0.95}{
\begin{threeparttable}
    \begin{tabular}{lcccc}
		\toprule
	&  LPM    & LPM, CF 	& \multicolumn{2}{c}{Poisson CF} \\
  & Coef (x100)  & 	Coef (x100)	&  Coef. & AME	(x100) \\
   &  (1)    &  (2) 	&  (3) &  (4)  \\
		\midrule
        $\mathcal{U}$ & 0.832***	&	0.857***    & 2.824**	 & 1.318* \\
		&  (0.048) &(0.300) 	 &  (1.355) & (0.641) \\
$1/\mathcal{P}$  & -0.0081*** &	-0.014***  &	-0.068*** & -0.032***\\
		& (6.06e-04) &(0.003) 	& (0.015)	& (0.007)\\
        Constant 	&  -0.100 &	-0.020	 &  -6.615*** &	\\
		 & (0.042) & (0.142) 	   & (0.596) & 	\\
        \midrule
        Controls betatest& X &  X & X & \\
		Controls Ad & X &  X & X& \\
   \midrule
     &  & \multicolumn{2}{c}{ $\rho$ for CF}    & \\
     &  & Coef (x100)& Coef.&\\
    \cmidrule{3-4}
   $\mathcal{U}$ &  & 0.003  & -1.173 & \\
                 &  & (0.303)   & (1.351) & \\
  $1/\mathcal{P}$ &  & 0.008**   & 0.029** & \\
                  &  & (0.003)  & (0.013) & \\
		\bottomrule
	\end{tabular}
    \begin{tablenotes}
 \footnotesize
 \item \textit{Notes}: 499,200 observations, with 49,920 individuals each receiving 10 recommendations, across the two experiments. Equation \eqref{eq:app2} is estimated modeling applications using a fixed-effect logit model approximated either as a linear model (columns (1) and (2)) or a Poisson model (columns (3) and (4)). Column (1): OLS; column (2): IV implemented using the control function approach.
  Column (3) ``Poisson CF'' estimation with multiplicative errors is performed using a control function approach. The sample pools both experiments. Instruments are the dummy variables for the groups job seekers have been randomly assigned to.  Column (4) ``Average Marginal'' Effect associated with estimates from column (3). Standard errors are clustered at the job seeker level. The lower panel provides estimates of the covariance between the residuals and the dependent variable, as provided by the Control Function method (see footnote \ref{foot:CF}). Significance levels: $<1\%: {}^{***}$, $<5\%: {}^{**}$, $<10\%: {}^{*}$.
 \end{tablenotes}
 \end{threeparttable}
}
\end{table}

As a robustness check, we also examine an alternative specification using observational data. We rely on data from the monitoring of job seekers’ search activity. All job postings on which a job seeker has clicked are identified and stored, along with subsequent actions—particularly whether an application was submitted. For each of these postings, we compute the indicators $\mathcal{U}$ and $\mathcal{P}$. We then estimate the model directly using this observational dataset. Table~\ref{tab:logit} in Appendix~\ref{app:modelobservational} presents these results, which are remarkably consistent with those in Table~\ref{tab:IVPoisson}. The main variables are all significant with the expected sign and, in addition, the coefficients from these estimations provide similar orders of magnitude. More precisely, the ratio of the two coefficients ranges from 0.018 to 0.025, very close to the previous values.

\subsection{Comparison of different RSs}\label{sec:comparisonRS}

We use the data from our experiments, together with the model developed in Section~\ref{sec:search}, to compare the performance of different recommendation systems (RSs).

Equation~\eqref{eq:gamma_plus} shows that the welfare-relevant score for a vacancy is given by the product of the probability of a successful application, $p$, and the function $\sigma \log\!\left(1+e^{\Delta(p,U)/\sigma}\right)$. The estimates from the previous section allow us to identify the surplus function $\Delta(p,U)$. In principle, one could therefore reconstruct the model-implied optimal score for each vacancy directly from the baseline utility and hiring scores.

We do not pursue this approach. Doing so would mechanically anchor the analysis on the two initial scores $\mathcal{U}$ and $\mathcal{P}$, and would prevent us from exploiting the additional information revealed by alternative RSs. Instead, we adopt the perspective that each RS provides a noisy signal about the two underlying components that matter for job seekers' welfare: the probability of being hired conditional on applying and the surplus from applying. Our objective is therefore to enrich the proxies for both dimensions by exploiting the full set of algorithms tested in the experiments.

To implement this strategy, we rely on data from Experiment~2. We briefly describe how we proceed and provide further details in Appendix \ref{app:estimation_scores}. For each enrolled job seeker, we observe the list of ten recommended vacancies and, for each of these vacancies, the scores produced by all algorithms considered in the second beta-test (with the exception of XGBoost). By combining these scores with observed application decisions and subsequent hiring outcomes on recommended vacancies, we estimate the components of the welfare-relevant score.

We proceed in four steps, detailed in Appendix~\ref{app:estimation_scores}.

\textit{Step~1. Hiring conditional on applying.} We identify the probability of a successful application, $p$. Restricting attention to vacancies to which job seekers actually applied, we estimate a logistic regression of the hire outcome on the scores produced by the different algorithms (\V.0, \V.2, \MixHalf(\V.2), \textsc{Application}, \Pbs). This yields predicted hiring probabilities $p_{i,j}$ for each job seeker--vacancy pair in the recommendation lists.

\textit{Step~2. Applications on recommended vacancies.} We identify the surplus component $\Delta$, or equivalently the application probability $p_a$. We estimate a logistic model for the probability that a recommended vacancy receives an application, again as a function of the same set of algorithmic scores. This yields predicted application probabilities $p_{a,i,j}$.\footnote{The results of these two first steps are reported in columns (1) and (2) of Table \ref{tab:calibration2}. For the hiring probability (column (1)), only the coefficient for the  \V.2 score is significant. This supports the idea that it effectively incorporates the  \V.0 score in predicting hires. Moreover, the fact that the  \textsc{Application} and  \Pbs\ scores do not predict hires once the  \V.2 score is accounted for also strengthens the point that these scores contain information distinct from the hiring probability. For applications (column (2)), the  \textsc{Application} score is as expected the most important predictor, even though the  \V.2 and \Pbs\ scores are also significant.}

\textit{Step~3. Reconstructing $\Gamma$ and related objects.} Combining these two sets of predictions, we reconstruct for each job seeker $i$ and each recommended vacancy $j\in\{1,\ldots,10\}$ the composite score $\Gamma_{i,j}$ as well as its components $p_{i,j}$, $p_{a,i,j}$, and $p_{h,i,j}=p_{i,j}\times p_{a,i,j}$.

\textit{Step~4. Counterfactual optimal recommendations.} We use the estimated $\Gamma$ function to construct a counterfactual benchmark. For each job seeker, we identify the set of vacancies that would have been recommended had it been possible to rank all available vacancies at the time of the experiment using the welfare-relevant score. This yields, for each job seeker, counterfactual scores $\Gamma_{i,j}^*$ for the (counterfactual) top ten vacancies.

We then evaluate each RS using the performance measures $\mu_{i,j}\in\{p_{i,j},p_{a,i,j},p_{h,i,j},\Gamma_{i,j},\Gamma_{i,j}^*-\Gamma_{i,j}\}$ and compute their averages within each experimental group assigned to algorithm $a\in \mathcal{A}_2$.

Since these estimates are subsequently used to construct predicted scores, we adopt a split-sample approach. One randomly selected half of the data, $S_1$, is used to estimate $p_a$, $p$, $p_h$, and $\Gamma$. The remaining half, $S_2$, is then used to compute average performance measures by experimental group:
\[
\overline{\mu_{i,j}}^{\,i\in S_2,\; a_i=a}.
\]
To account for the uncertainty introduced by sample splitting, valid confidence intervals are constructed by taking the medians of the upper and lower bounds of the confidence intervals across multiple splits \citep[see][]{chernozhukov2018generic}.

Figure \ref{fig:coef_levels_all_gamma} summarizes the performance of the different RSs along four dimensions: the probability of being hired conditional on applying, $p$ (panel (a)); the probability of applying, $p_a$ (panel (b)); expected utility prior to application, $\Gamma$ (panel (c)); the joint probability of applying and being hired, $p_h=p \times p_a$ (panel (d)).\footnote{Table \ref{tab:evaluation1} provides the associated estimates.} The figures also display the average value of the vacancies that would have been recommended by the optimal RS, as measured by the metric used in the figure (depending on the figure, either $p$, $p_a$, $p_h$ or $\Gamma$). 
This benchmark represents the performance that would be attained by the $\Gamma$-optimal recommendation set. It allows us to directly compare the performance of each recommendation rule to that of the optimal RS, using the metric implicit to each panel. Accordingly, the gap between the performance of a given RS and this reference captures how far the RS is from the optimal benchmark along each dimension ($p$ in panel (a), $p_a$ in panel (b), $\Gamma$ in panel (c), and $p_h$ in panel (d)).
While this comparison is informative across all panels, it is particularly meaningful in panel (c), which relies on the $\Gamma$ metric: only for $\Gamma$ does the gap admit a direct welfare interpretation as a loss relative to the optimum. This panel therefore provides a direct assessment not only of how well each algorithm performs under the appropriate objective, but also of how close it comes to the optimal benchmark.

Panel (a) shows that all algorithms significantly increase the probability of  a hire conditional on application $p$, relative to the baseline \Pbs\ system. Even the algorithm with the smallest gain, \textsc{XGBoost}, more than doubles this probability. The best-performing algorithm on this dimension is \V.2, for which the probability of hire is 3.2 times higher than under \Pbs. Despite these relative improvements, it is important to note that absolute success rates remain very low. Even for the best algorithm (\V.2), the conditional probability of hire on a recommended job posting remains below 1.5\%. 

Panel (b) presents similarly strong gains in the probability of application $p_a$. Again, relative to the \Pbs\ benchmark, improvements are substantial. The smallest gain is observed for \V.0, which still increases the probability of application by 75\%. The \textsc{Application} algorithm achieves the highest impact, increasing the probability by a factor of 1.7, with \V.2 not far behind at 1.4. This can be seen as a consistency check given as the main predictor in our estimation of $p_a$ is the \textsc{Application}  score (see Table \ref{tab:calibration2}).   Yet, absolute levels remain modest: even the best-performing algorithm yields an application probability of only 0.4\%.

Panel (d) shows the unconditional probability of hire $p_h$. The pattern mirrors that of the previous panels: the new RSs all substantially outperform \Pbs, with \V.2 again providing the largest gain—tripling the likelihood of successful matches. The \textsc{Application} algorithm performs nearly as well. However, the absolute probability of a match remains extremely low, around 1 in 10,000.

Finally, panel (c) reports the optimal score $\Gamma$ which ranks the algorithms according to their expected value from the job seeker’s perspective. Two performance tiers emerge clearly: \V.2 and \textsc{Application} form a top tier, generating large gains relative to the benchmark \Pbs, while \V.0, \MixHalf(\V.2), and \textsc{XGBoost} constitute a second tier. As shown in Section~\ref{subsubsec:why_ph}, $\Gamma\approx p_h\times(1+p_a/2)$ for small $p_a$, so the two metrics are nearly identical in our empirically low-application setting. 
Overall, the strong performance of \V.2 and \textsc{Application} highlights the importance of explicitly modeling application behavior when deriving optimal recommendations using the objective $\Gamma$. At the same time, these gains should be interpreted cautiously: in absolute terms, application and hiring rates for recommended vacancies remain low across all systems.

Across all panels, recommendations generated by the optimal algorithm strictly dominate those produced by any alternative algorithm. This ranking holds regardless of the metric considered. Among the feasible algorithms, \V.2 consistently performs closest to the optimal benchmark, with one exception: for the application probability metric, the algorithm based solely on $p_a$ is the closest and in fact delivers nearly identical values. The most informative comparison is that based on the $\Gamma$ metric. Along this dimension, the performance achieved by \V.2 is very close to that of the optimal algorithm, implying that the effective loss from using \V.2 instead of the optimal recommendation rule is  small.

\begin{figure}[!htbp]
    \centering

    \begin{minipage}{0.45\textwidth}
        \centering
        \includegraphics[width=\linewidth]{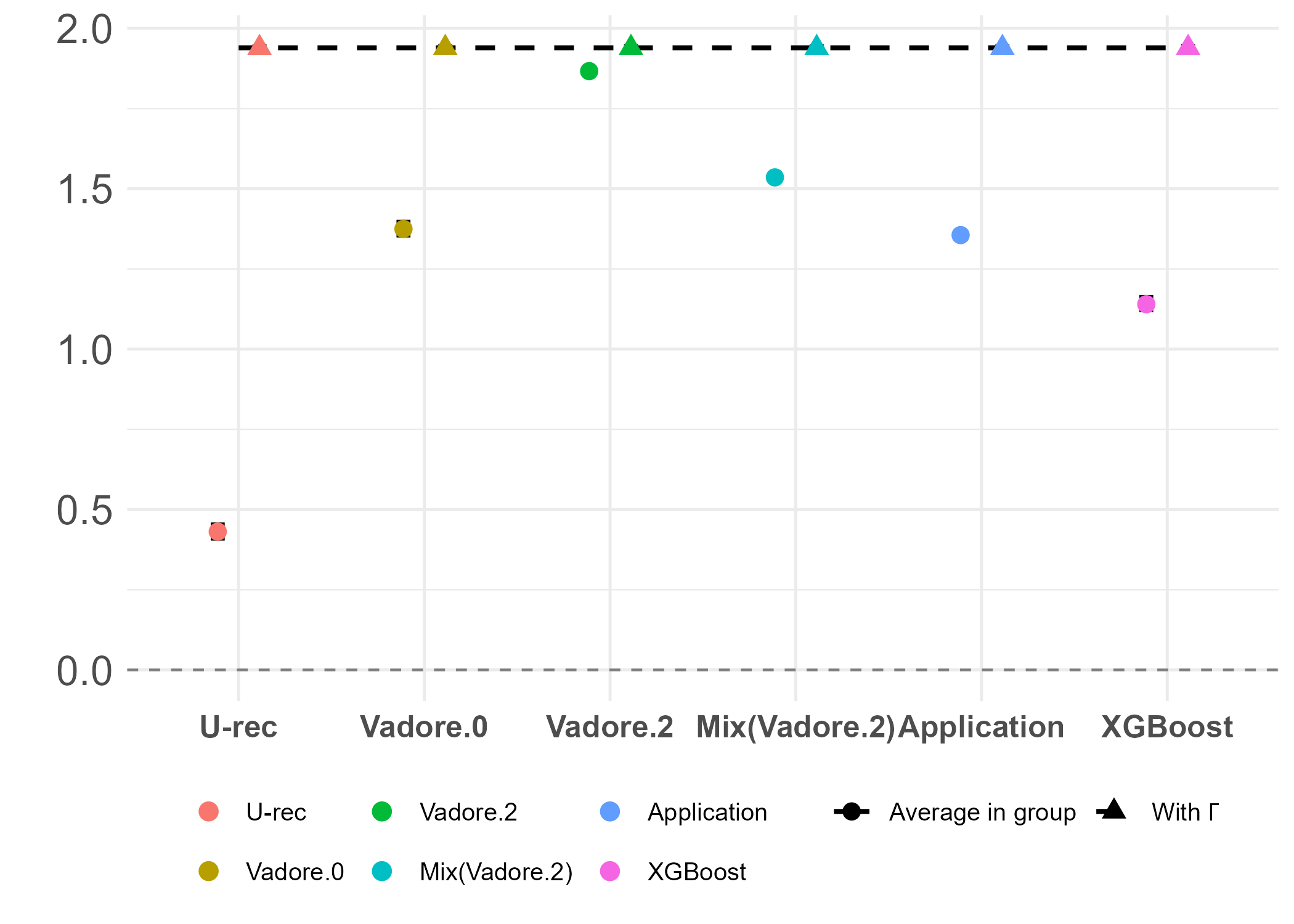}\\
        {\small (a) $p$ ($\times 100$)}
    \end{minipage}\hfill
    \begin{minipage}{0.45\textwidth}
        \centering
        \includegraphics[width=\linewidth]{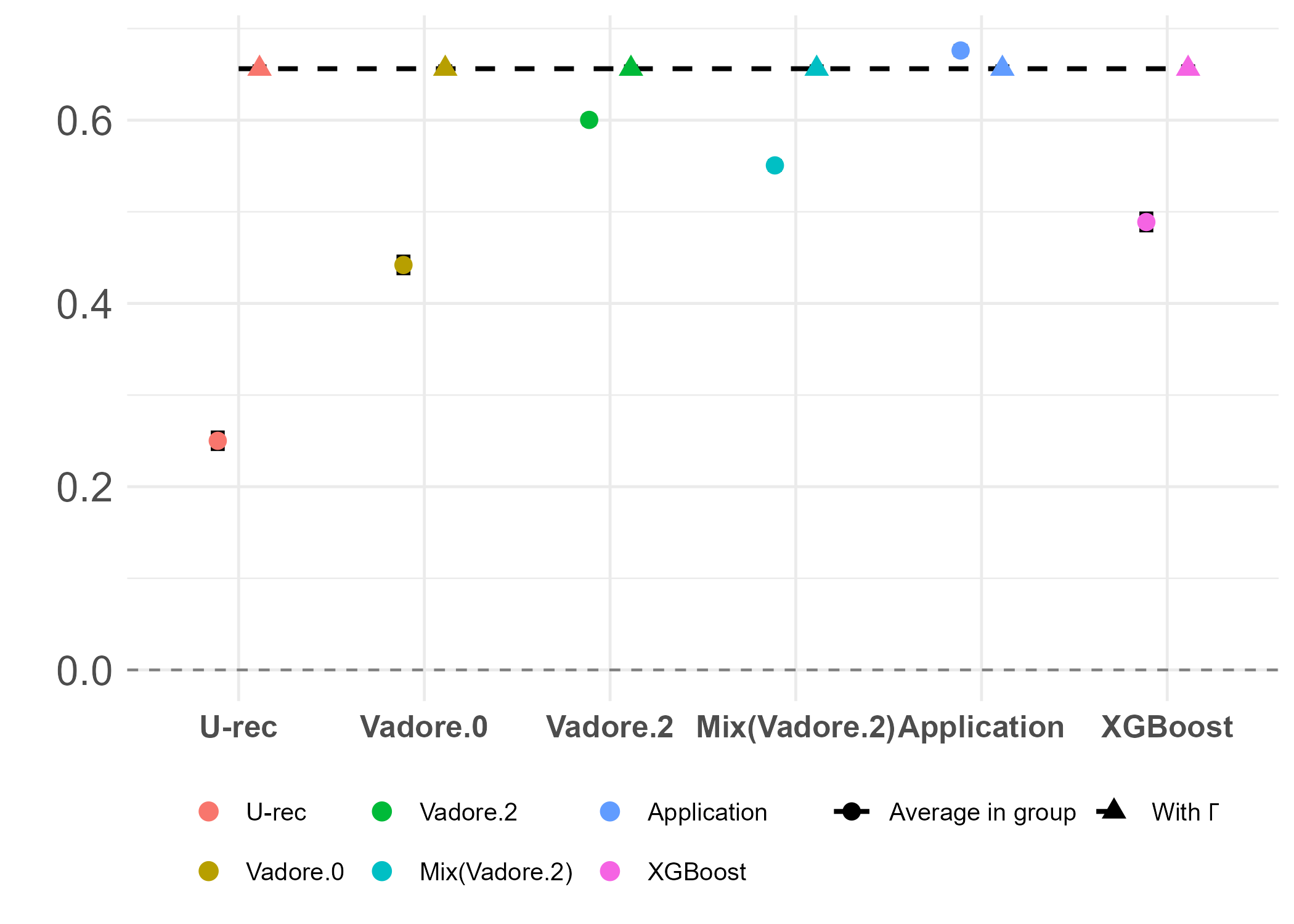}\\
        {\small (b) $p_a$ ($\times 100$)}
    \end{minipage}

    \vspace{0.35cm}

    \begin{minipage}{0.45\textwidth}
        \centering
        \includegraphics[width=\linewidth]{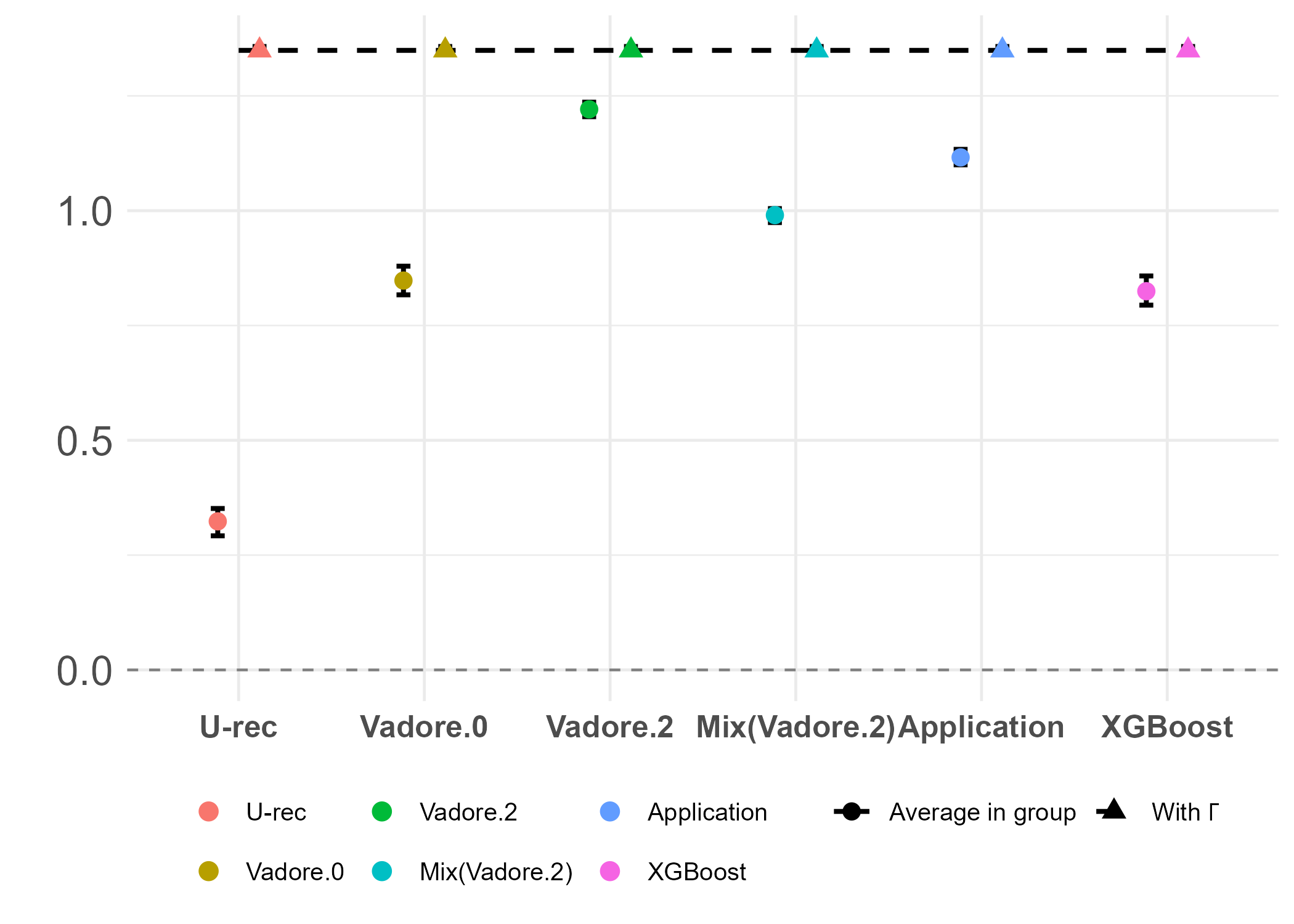}\\
        {\small (c) $\Gamma$ ($\times 10,000$)}
    \end{minipage}\hfill
    \begin{minipage}{0.45\textwidth}
        \centering
        \includegraphics[width=\linewidth]{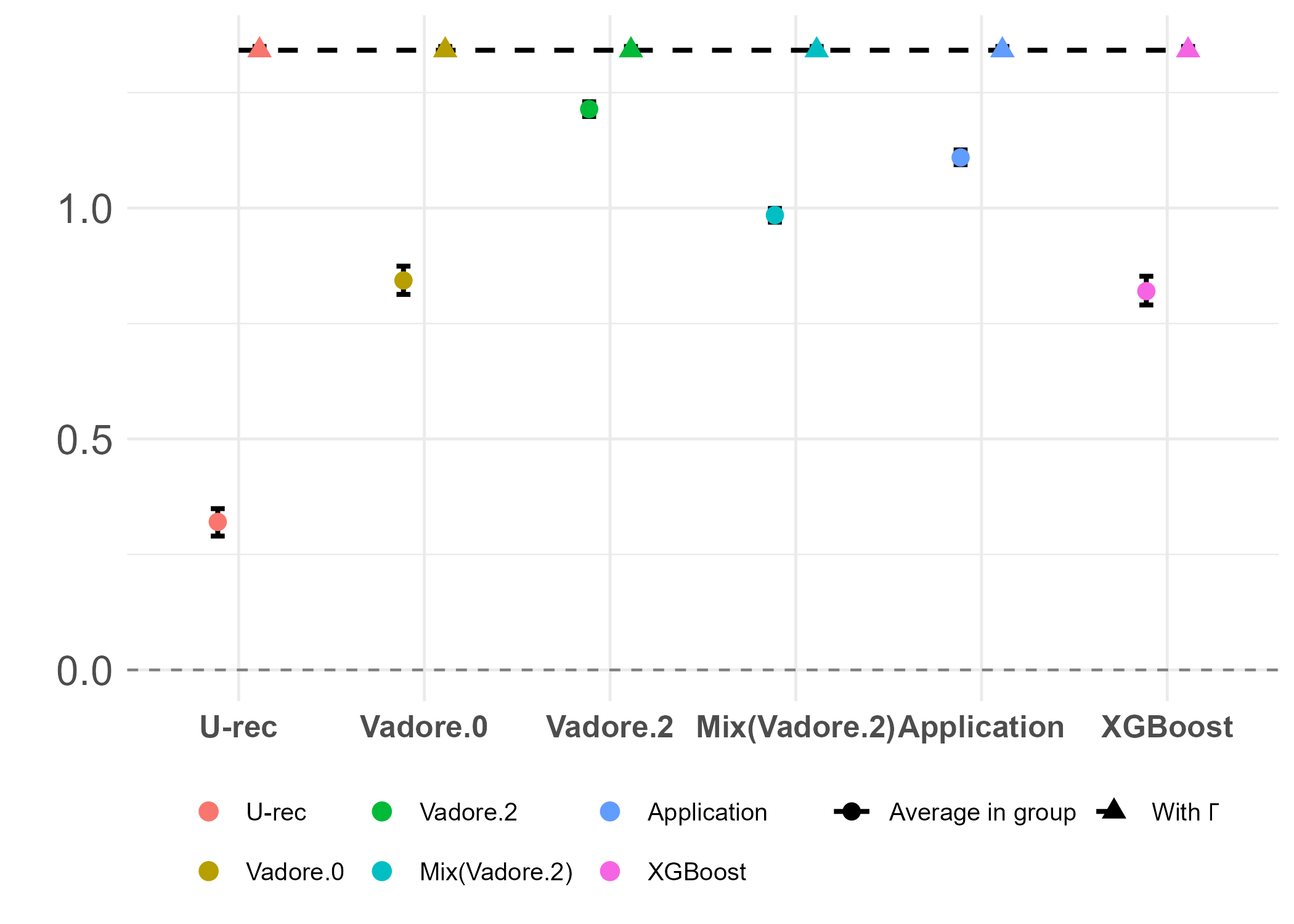}\\
        {\small (d) $p_h$ ($\times 10,000$)}
    \end{minipage}

    \vspace{0.35cm}

    \caption{Second experiment - Reduced-Form estimation - impact of the assignment on $\Gamma$ for the 10 first ads}
    
    \label{fig:coef_levels_all_gamma}

    \vspace{0.2cm}
    \begin{minipage}{0.95\textwidth}
        \footnotesize
        \textit{Notes:} Each panel reports coefficient estimates from separate regressions of the indicated outcome on recommendation-treatment indicators, controlling for ad fixed effects. Points represent estimated coefficients and bars denote 90\% confidence intervals. Standard errors are clustered at the job-seeker level. Outcomes $\mathcal{P}$, Clicks, Applications, and Hirings are rescaled as indicated in the figure labels. 
    \end{minipage}
\end{figure}

\section{Conclusion}\label{sec:conclusion}

In this paper, we study the design of job recommendation systems (RSs) by combining economic modeling, machine learning, and field experimentation. We develop a job-search framework in which vacancies are lotteries characterized by a hiring probability $p$ and payoff $U$. The model shows why recommendation rules based solely on proxies for $p$, proxies for $U$, or observed application behavior are incomplete, and that welfare-relevant rankings must combine both dimensions into an expected-surplus index. It also highlights an inversion problem: observed application choices reveal whether applying is privately profitable, but not the magnitude of the expected gains relevant for welfare.

We bring this framework to the data through collaboration with the French PES. Starting from two operational RSs, one reflecting stated preferences and one optimized to predict hiring outcomes, we conduct two randomized field experiments conceived as beta tests within a learning cycle. Guided by the model and experimental feedback, this process leads to an approximation of the welfare-optimal RS (\V.2), whose performance in terms of clicks and applications substantially exceeds that of the initial systems.

Beyond reduced-form performance, the experiments generate exogenous variation in the characteristics of recommended vacancies, which we use to estimate a structural model of application behavior. The estimates support the key behavioral mechanisms emphasized in the theory and quantify the relative importance of hiring probabilities and utility in job seekers’ decisions. Combined with the model’s structure, the experimental data allow us to construct an empirically grounded benchmark for welfare-relevant rankings and compare all tested recommendation rules to it. We find that the approximation \V.2 of the welfare-optimal algorithm delivers large gains relative to the initial systems and performs close to the model-implied optimum.

A broader lesson from our analysis concerns the performance of simple heuristic rankings. While the joint application-and-hiring probability $p_h$ is not welfare-optimal in theory, it emerges as a strong empirical benchmark in our setting. This result is structural rather than algorithmic: application probabilities are empirically small and remain so even under recommendation rules designed to stimulate applications. In this regime, the welfare-relevant index is well approximated by $p \times p_a$, explaining why hiring-based rankings dominate alternative heuristics. By contrast, rankings based solely on application behavior are theoretically fragile. Their reasonable performance in our setting may not generalize to environments where application behavior responds more strongly to recommendations. When recommendations substantially affect application decisions, the gap between behavior-based and welfare-based rankings may be much larger.

More broadly, our results suggest a general lesson for the design of algorithmic intermediation in labor markets. Machine-learning tools can substantially improve matching outcomes, but only when embedded in a framework that defines the economic objective and disciplines behavioral assumptions with experimental evidence. Without such a framework, RSs optimized for observable behaviors may perform well on predictive metrics yet remain misaligned with welfare-relevant outcomes.

Our findings suggest several directions for future research. First, improving RS performance requires better prediction of the primitives $p$ and $U$, especially job seekers’ utility for different jobs. Second, scaling recommendations may generate congestion and general-equilibrium effects \citep[see, e.g.,][]{gee2019more,altmann2022direct,bied2021congestion,SuBayoumiJoachims2022,roland2022,behaghel2024potential,lehmann2023}, calling for market-level recommendation rules, for example based on optimal transport \citep{bied2021congestion}. Third, fairness and inequality concerns remain central in labor-market RSs \citep{zhang2022understanding,bied2022fairness}, and post-processing approaches with fairness constraints appear promising. Fourth, incorporating behavioral frictions may further improve recommendation design, for instance by combining RSs with elicited beliefs to study how recommendations shape expectations and search strategies \citep[see, e.g.,][]{almaas2023economics,anticipations}. Finally, an important extension is to develop RSs that also benefit firms \citep{Horton2017,algan2020active},  moving toward two-sided systems that jointly model applications and hiring decisions.

\FloatBarrier

\newpage
\clearpage

\newpage
\clearpage

\FloatBarrier

\bibliographystyle{chicago}
\bibliography{references}

\newpage

\renewcommand{\thesection}{\Alph{section}}

\renewcommand\thefigure{\Alph{section}\arabic{figure}}
\renewcommand{\thetable}{\Alph{section}\arabic{table}}
\renewcommand{\thefigure}{\Alph{section}\arabic{figure}}
\renewcommand{\thesubfigure}{\Alph{section}\arabic{figure}}
\renewcommand{\theequation}{\Alph{section}\arabic{equation}}
\setcounter{section}{0}
\setcounter{table}{0}
\setcounter{figure}{0}
\setcounter{equation}{0}

\appendix

\newpage

\pagenumbering{arabic}

\hspace{3.8cm} \textbf{ \Large ONLINE APPENDIX}

\section{Non-myopic job seekers}\label{sec:nonmyopic}
\subsection{Derivation of RSs for non-myopic job seekers} \label{sec:nonmyopicth}

The analysis in Section \ref{subsec:optimal} focuses on the discounted value when using an RS for myopic job seekers. We analyze here the case of non-myopic job seekers. We first address the question of the discounted value when the optimal RS assuming a myopic job seeker is introduced. We simply correct the discounted value. We then show how to adapt the set of recommendations to the case of non-myopic job seekers.

The main change compared to the myopic case is that non-myopic job seekers adjust their reservation utility and their decision rule to apply. In such a case, the discounted value, that we denote by $rV^{adj}_{1}$ for an unemployed job seeker exposed to recommendations, is the solution of an equation that takes the following form:

\begin{equation}\label{eq:adj}
	rV^{adj}_{1}(S,s,\alpha_1)= u(b) +  \frac{\alpha_1}{r+q}\frac{\mathbb{E}\!\left( \Gamma^{adj}(p,U)\indic\{S>\overline{q}_S(s)\}\right)}{s},
\end{equation}

where $\Gamma^{adj}(p,U)$ is based on $rV^{adj}_{1}$ rather than on $rV_{0}$ as in the myopic case. More precisely:
\begin{equation}\label{eq:gammaadj}
\Gamma^{adj}(p,U)=p\sigma \log\left(1+e^{\Delta^{adj}(p,U)/\sigma}\right),
\end{equation}
using $\Delta^{adj}(p,U)=U-U^*_1(p)$ and 
$U_1^*(p) =rV^{adj}_{1} - \overline{R} + (\overline{k} + \overline{R})/p$.

The changes in the discounted value following the introduction of an RS obtained when assuming a myopic job seeker are also informative about the changes in the discounted value for non-myopic job seekers. Consider an RS (in the case $\alpha_1=\alpha_0$) that selects the top $s\%$ jobs according to an index $S$, and call $\delta^m(S,s)=rV^m_{1}(S,s,\alpha_0)-rV_{u,0}$ the change in the discounted value of an unemployed job seeker compared to the discounted values given by the equations \eqref{eq:myopic} and \eqref{eq:Bellman1_pi}. Let $\delta^{adj}(S,s)=rV^{adj}_{1}(S,s,\alpha_0)-rV_{u,0}$ be the change in the discounted value of a non-myopic unemployed job seeker.
We define the expected hiring rate on recommended vacancies for a myopic job seeker $\theta^m(S,s)=\mathbb{E}\!\left(p F\left(\Delta(p,U)/\sigma\right)\indic\{S>\overline{q}_S(s)\}\right)\big/s$ and for a non-myopic job seeker (see Equation \eqref{eq:gammaadj}) $\theta^{adj}(S,s)= \mathbb{E}\!\left(p F\left(\Delta^{adj} (p,U)/\sigma\right)\indic\{S>\overline{q}_S(s)\}\right)\big/s$ .

\begin{prop}\label{prop_nonmyopic}
    The changes in the discounted value for myopic and non-myopic job seekers have the same sign. Focusing on RSs for which $\delta^m(S,s)$ is positive: 
   \begin{itemize}
    \item The hiring rate on recommended vacancies for myopic job seekers is larger than for non-myopic job seekers:
   $$\theta^m(S,s)\geq\theta^{adj}(S,s)$$
   \item The change in the discounted value for a non myopic unemployed job seeker $\delta^{adj}(S,s)$ can be bracketed by
    \begin{equation}\label{eq:myop_non_myop1}
  \frac{r+q}{r+q+\alpha_0\theta^{m}(S,s)}\delta^m(S,s) \leq   \delta^{adj}(S,s)\leq\frac{r+q}{r+q+\alpha_0\theta^{adj}(S,s)}\delta^m(S,s),
     \end{equation}
\item An order of magnitude of the change is given by
     \begin{equation}\label{eq:myop_non_myop2}
     \delta^{adj}(S,s)\approx\frac{r+q}{r+q+\alpha_0\theta^m(S,s)}\delta^m(S,s).
     \end{equation}
 \item Non-myopic job seekers  are more selective: their reservation utility for a job they are sure to get increases:
$$U_1^*(1) =U_0^*(1)+\delta^{adj}(S,s).$$
\end{itemize}
\end{prop}

\begin{proof}
See model Appendix \ref{app:proof}.
\end{proof}

The result in Proposition \ref{prop:myopic} states that the optimality of the selection rule $\Gamma^m(p,U)$  holds in the case of myopic job seekers.  

When job seekers are no longer myopic, the problem of identifying an optimal RS is more complex, as shown by the equations \eqref{eq:adj} and \eqref{eq:gammaadj}, since the present value once the RS has been trivially implemented depends on the optimal RS, but conversely the optimal RS also depends on the present value.\footnote{The optimal RSs consists of ranking the vacancies according to
$p\sigma \log\left(1+e^{\left(\Delta^m(p,U)-x\right)/\sigma}\right)$
but the $x$ to consider is difficult to identify. Defining $S(x)$ as $S(x)=p\sigma \log\left(1+e^{\left(\Delta^m(p,U)-x\right)/\sigma}\right)$, it is the solution of 
$ rV+x= u(b) +  \frac{\alpha_0}{r+q}\mathbb{E}\!\left( S(x)\left|F_{S(x)}(S(x))>1-s\right.\right)$ with $rV= u(b) +  \frac{\alpha_0}{r+q}\mathbb{E}\!\left( S(0)\right)$
}

\section{Model Appendix}\label{app:model}
\subsection{Proof of propositions}\label{app:proof}

\begin{proof}[\textbf{Proof of Proposition \ref{prop_V_u0}}]
Fix a worker type $x$ and suppress $x$ when there is no ambiguity. Over an interval of length $dt$, the match survives with probability $1-qdt$ and breaks
with probability $qdt$. Thus,
\begin{equation}\label{eq:v_e_clean}
(1+rdt)V_e(U+\varepsilon)
=
(U+\varepsilon)dt + (1-qdt)V_e(U+\varepsilon) + qdt\,V_0.
\end{equation}
Rearranging and taking the limit yields
\[
(r+q)V_e(U+\varepsilon)=U+\varepsilon+qV_0,
\qquad\text{so that}\qquad
V_e(U+\varepsilon)-V_0=\frac{U+\varepsilon-rV_0}{r+q}.
\]
While unemployed, vacancies arrive at rate $\alpha_0$ and each draw yields $(p,U,\varepsilon)$.
Upon observing $(p,U,\varepsilon)$, the worker either applies ($C=1$) or does not apply
($C=0$). The unemployment Bellman equation is
\begin{align*}
&(1+rdt)V_0 \\
&= u(b)dt
+ (1-\alpha_0dt)V_0 \\
&\quad + \alpha_0dt \int C(p,U,\varepsilon) \big(pV_e(U+\varepsilon) + (1-p)(V_0-R)-k\big) \, dF_0(p,U)\, dF_\varepsilon(\varepsilon) \\
&\quad + \alpha_0dt \int  (1-C(p,U,\varepsilon))V_0 \ dF_0(p,U)\, dF_\varepsilon(\varepsilon).
\end{align*}
Subtracting $V_0$ from both sides, dividing by $dt$, and taking the limit in $dt$ gives
\[
rV_0
=
u(b)
+\alpha_0
\int
C(p,U,\varepsilon)\Big(
p\big(V_e(U+\varepsilon)-V_0\big) - (1-p)R - k
\Big)\, dF_0(p,U)\, dF_\varepsilon(\varepsilon).
\]
Using the expression for $V_e(U+\varepsilon)-V_0$ derived above,
\[
rV_0
=
u(b)
+\frac{\alpha_0}{r+q}
\int
C(p,U,\varepsilon)\Big(
p(U+\varepsilon-rV_0) - (r+q)\big((1-p)R+k\big)
\Big)\, dF_0(p,U)\, dF_\varepsilon(\varepsilon).
\]

Define $\overline R:=(r+q)R$ and $\overline k:=(r+q)k$, and define the reservation utility
\[
U_0^*(p):= rV_0-\overline R + \frac{\overline k+\overline R}{p},
\]
which is equivalent to \eqref{eq:res_u}. Then the term inside the integral can be written as
\[
p\big(U-U_0^*(p)+\varepsilon\big).
\]
Therefore the optimal application rule is
\[
C(p,U,\varepsilon)=\indic\{U-U_0^*(p)+\varepsilon>0\}.
\]
Defining the surplus $\Delta(p,U):=U-U_0^*(p)$ yields the application rule \eqref{eq:appli}.
Moreover, substituting this rule back into the Bellman equation gives
\[
rV_0
=
u(b)
+\frac{\alpha_0}{r+q}\,
\mathbb{E}\!\left[
p\int (\Delta(p,U)+\varepsilon)\indic\{\Delta(p,U)+\varepsilon>0\}\, dF_\varepsilon(\varepsilon)
\right],
\]
which corresponds to \eqref{eq:Bellman1_pi}--\eqref{eq:gamma_mplus} with
\[
\Gamma(p,U)
:=
p\int (\Delta(p,U)+\varepsilon)\indic\{\Delta(p,U)+\varepsilon>0\}\, dF_\varepsilon(\varepsilon).
\]
This establishes that application behavior is governed by $\Delta(p,U)$, and the value of
encountering a vacancy is summarized by $\Gamma(p,U)$.
\end{proof}

\begin{proof}[\textbf{Proof of Proposition \ref{prop:myopic}}]
Under an RS characterized by $(S,s,\alpha_1)$, the value of unemployment for a myopic job
seeker is given by \eqref{eq:myopic}:
\[
rV^m_{1}(s,\alpha_1)
=
u(b)
+
\frac{\alpha_1}{r+q}\,
\mathbb{E}\!\left(
\Gamma^m(p,U)
\mid
S>\overline{q}_S(s)
\right),
\]
where $\overline{q}_S(s)$ is the $(1-s)$-quantile of $S$ under $F_0$ (with arbitrary
tie-breaking if needed).
For fixed $s$, maximizing $V_1^m$ is therefore equivalent to maximizing
\[
\mathbb{E}\!\left(
\Gamma^m(p,U)\,\mathbf 1\{S>\overline{q}_S(s)\}
\right)
\]
over all measurable scores $S$.

\medskip

\noindent\textbf{Proof of point \eqref{2.2.1}.}
The previous expression depends on $S$ only through the induced selection set
$A=\{S>\overline{q}_S(s)\}$, which must satisfy $\mathbb{P}(A)=s$.
An application of the rearrangement inequality \citep[see, \emph{e.g.}, Theorem~368,][]{HLP52}
implies that, among all measurable subsets $A$ of probability $s$, the expectation
$\mathbb{E}\!\left(\Gamma^m \indic\{A\}\right)$ is maximized when the indicator
$\indic\{A\}$ is comonotonic with $\Gamma^m$, that is
\[
A=\{\Gamma^m>\overline{q}_{\Gamma^m}(s)\}.
\]
Equivalently, the optimal rule selects the top $s$ fraction of vacancies according to
$\Gamma^m$.

Indeed, if $A$ contains states with low $\Gamma^m$ while excluding states with higher
$\Gamma^m$, swapping them weakly increases
$\mathbb{E}\!\left(\Gamma^m \indic\{A\}\right)$.
Iterating this argument implies that the optimizer must include all higher-$\Gamma^m$
states before lower-$\Gamma^m$ ones, hence be a threshold rule based on $\Gamma^m$.
Such a selection rule is implemented by the score $S(p,U)=\Gamma^m(p,U)$.

\medskip

\noindent\textbf{Proof of point \eqref{2.2.2}.}
In the absence of recommendations, the value of unemployment satisfies
\[
rV_0
=
u(b)
+
\frac{\alpha_0}{r+q}\,
\mathbb{E}\!\left(\Gamma^m(p,U)\right).
\]
Thus, when $\alpha_1\geq\alpha_0$, the RS improves welfare relative to baseline search
whenever
\[
\mathbb{E}\!\left(
\Gamma^m(p,U)
\mid
S>\overline{q}_S(s)
\right)
\ge
\mathbb{E}\!\left(\Gamma^m(p,U)\right).
\]

Let $g(z)=\mathbb{E}(\Gamma^m\mid S=z)$. By the law of iterated expectations,
\begin{align*}
\mathbb{E}(\Gamma^m\mid S>\overline{q}_S(s))
 & =
\mathbb{E}\!\left(g(S)\mid S>\overline{q}_S(s)\right),   \\
\mathbb{E}( \Gamma^m)&=\mathbb{E}(g(S))=\mathbb{E}(g(S)|S>\overline{q}_S(s))s+\mathbb{E}(g(S)|S<\overline{q}_S(s))(1-s).
\end{align*}
Thus, we obtain
\begin{align*}
&\mathbb{E}( \Gamma^m)-E( \Gamma^m|S>\overline{q}_S(s))\\
&=\mathbb{E}(g(S)|S>\overline{q}_S(s))s+\mathbb{E}(g(S)|S<\overline{q}_S(s))(1-s)-\mathbb{E}(g(S)|S>\overline{q}_S(s))
\end{align*}
and thus 
$$\mathbb{E}( \Gamma^m)-\mathbb{E}( \Gamma^m|S>\overline{q}_S(s))=(1-s)(\mathbb{E}(g(S)|S<\overline{q}_S(s))-\mathbb{E}(g(S)|S>\overline{q}_S(s))),$$ 
which is negative given $g$ is increasing.

\medskip

\noindent\textbf{Strict improvement under the optimal RS.}
The second sentence of point~\eqref{2.2.2} follows immediately: when $S=\Gamma^m$, we have $g(z)=z$, which is strictly increasing.
Therefore,
\[
\mathbb{E}(\Gamma^m\mid \Gamma^m>\overline{q}_{\Gamma^m}(s))
>
\mathbb{E}(\Gamma^m)
\]
for any $s\in(0,1)$ whenever $\Gamma^m$ is non-degenerate.
If $\alpha_1\geq\alpha_0$, this implies a strict increase in the value of unemployment
relative to search without recommendations.
\end{proof}

\begin{proof}[\textbf{Proof of Proposition \ref{prop:identification}}] Identification of $\alpha$, $\beta$, and $\gamma$ in the binary choice model \eqref{eq:identification} is direct given the normalizations. Then, using \eqref{eq:appli} and that  $\varepsilon$  is distributed
as a logistic distribution with scale parameter $\sigma$, with $F (z/\sigma)$ its cumulative
distribution, we obtain that 
\begin{align*}
\mathbb{P}(C_{i,j}=1 \mid p_{i,j},U_{i,j}-U_{0,i}^*(1)) & =F( \Delta(p_{i,j},U_{i,j})/\sigma) \\
& =F\left( \frac{1}{\sigma}(U_{i,j}-U_{0,i}^*(1))-\frac{\overline{k}+\overline{R}}{\sigma p_{i,j}}+\frac{\overline{k}+\overline{R}}{\sigma}\right),
\end{align*}
using \eqref{eq:res_u} and \eqref{eq:def_surplus}. Using the assumption that parameters in \eqref{eq:identification} are identified, this yields that $1/\alpha$ identifies $\sigma$ and $\beta/\alpha$ identifies $\overline{k}+\overline{R}$. Identification of $\Delta(p,U)$ and $\Gamma(p,U)$ is a direct consequence.
\end{proof}

\begin{proof}[\textbf{Proof of Proposition \ref{prop_nonmyopic}}]
When job seekers are not myopic, the discounted value of job seekers is given by
\begin{align*}
	rV_{1}(S,s,\alpha_0)&= u(b) +  \frac{\alpha_0}{r+q}\mathbb{E}\!\left(\Gamma^{1}(p,U,rV_{1}(x,S,s,\alpha_0))\middle|F_S(S(p,U)>1-s\right) \notag\\
 &=u(b)+\frac{\alpha_0}{r+q}\mathbb{E}\!\left(p\sigma\log\left(1+e^{\Delta^1(p,U,rV_{1}(S,s,\alpha_0))/\sigma}\right)\middle|F_S(S(p,U)>1-s\right)
\end{align*}
where 
\begin{equation*}
\Delta^1(p,U,z) := U - z + \overline{R} - \frac{\overline{k}+\overline{R}}{p}.
\end{equation*}
The discounted value $rV_{1}(S,s,\alpha_0)$ is thus the solution of an equation of the form:
$$z(s)=f_S(z(s),s).$$
$f_S(z,s)$ is a decreasing function of $z$. The value of the derivative is given by
 $$\frac{df_S}{dz}(z,s)=-\frac{\alpha_0}{r+q}\mathbb{E}\!\left(pF\left(\Delta^1(p,U,z)/\sigma\right)\left|F_S(S(p,U)>1-s\right.\right)$$
 and the second derivative is given by 
\begin{align*}&\frac{d^2f_S}{dz^2}(z,s)\\
&=\frac{\alpha_0}{(r+q)\sigma}\mathbb{E}\!\left(pF\left(\Delta^1(p,U,z)/\sigma\right)\left(1-F\left(\Delta^1(p,U,z)/\sigma\right)\right)\left|F_S(S(p,U)>1-s\right.\right)>0.
\end{align*}
 
The discounted value in the absence of the RSs is the solution $z_0$ of  $z_0=f_S(z_0,1)$. Thus $\delta^{adj}(S,s)=z(s)-z_0$, and we have
\begin{align*}
z(s)-z_0&=f_S(z(s),s)-f_S(z_0,1)=f_S(z(s),s)-f_S(z_0,s)+[f_S(z_0,s)-f_S(z_0,1)]\\
&=(z(s)-z_0)\frac{df_S}{dz}(z'(s),s)+[f_S(z_0,s)-f_S(z_0,1)]
\end{align*}
with $\frac{df_S}{dz}(z'(s),s)$ the derivative of $f_S$ for a value $z'(s)$ in between $z_0$ and $z(s)$. Thus,

$$z(s)-z_0= \frac{f_S(z_0,s)-f_S(z_0,1)}{1-\frac{df_S}{dz}(z'(s),s)}$$
The quantity $f_S(z_0,s)-f_S(z_0,1)$ is the myopic change $\delta^m(S,s)$. Given $\frac{df_S}{dz}(z'(s),s)<0$, this shows that $\delta^m(S,s)$ and $\delta^{adj}(S,s)$ have the same sign.

Focusing on cases for which $\delta^m(S,s)$ is positive, this implies that $\delta^{adj}(S,s)$ is positive and thus that $z(s)\geq z_0$. 
In addition, the second derivative of $f$ with respect to $z$ is positive, thus $\frac{df_S}{dz}(z_0,s)<\frac{df_S}{dz}(z'(s),s)<\frac{df_S}{dz}(z(s),s)$, and we have
\begin{align*}
\frac{df_S}{dz}(z_0,s)&=-\frac{\alpha_0}{r+q}\mathbb{E}\!\left(pF\left(\Delta^1(p,U,z_0)/\sigma\right)\left|F_S(S(p,U)>1-s\right.\right)\\
&=-\frac{\alpha_0}{r+q}\theta^m(S,s)\\
\frac{df_S}{dz}(z(s),s)&=-\frac{\alpha_0}{r+q}\mathbb{E}\!\left(pF\left(\Delta^1(p,U,z(s))/\sigma\right)\left|F_S(S(p,U)>1-s\right.\right)\\
&=-\frac{\alpha_0}{r+q}\theta^{adj}(S,s).
\end{align*}

As a result
$$\theta^m(S,s)\geq\theta^{adj}(S,s)$$
and 
$$\frac{\delta^m(S,s)}{1+\frac{\alpha_0}{r+q}\theta^m(S,s)}\leq \delta^{adj}(S,s)\leq\frac{\delta^m(S,s)}{1+\frac{\alpha_0}{r+q}\theta^{adj}(S,s)}$$

which gives the result.
\end{proof}

\section{Usual assessment of an RS's performance: recall@k}\label{app:recall}

Although algorithms can be obtained according to different principles and with different data, there is a common way of measuring their performance, which is the ``Recall@k''. Consider a target variable $M$, such as $M_{i,j}=1$ if $i$ has been hired by $j$. As usual, the algorithm is trained on a ``train sample'' and tested on a ``test sample''. For each individual $i$ and the $k$ best vacancies according to $S$ in the test sample $\mathcal{J}_k^*(S,i)$, we can build a variable $M_i^k(S)$ which takes value 1 if the target variable $M_{i,j}$ takes value 1 for one of these $k$ best $S$-based vacancies:

\begin{equation}
\label{eq:recall(i)}
    M_i^{k}(S)=\mathds{1}\left(\sum_{j\in{\mathcal{J}_k^*(S,i)}}M_{i,j}=1\right),
\end{equation}
where $\mathds{1}(\cdot) $ denotes the indicator function. If the target variable is hiring, the recall@$k$ is the proportion of job seekers who were hired on one of the top-$k$ recommendations: 
\begin{equation}
\label{eq:recall}
\textrm{recall}@k(S)=\frac{1}{N}\sum_{i=1}^NM_i^{k}(S).
\end{equation}

This is the usual measure in the machine learning literature of the global performance of the RSs $S$, which can be used for example to compare two RSs.

\medskip

Figure \ref{fig:perfs} shows the performance of different  RSs that we considered when building our final ML-based RS. The figure on the right panel compares the performances in terms of recall@100 on the test set of different RSs.  Progressively including additional variables (such as previously considered vacancies) yields huge improvements on the recall.\footnote{See the definition of the recall@k in Appendix~\ref{app:recall}.}
The first RSs (``fixed weights'') corresponds to the $\mathcal{U}$-ranking, the preference-based RSs inspired from the PES's current one. 
As the graph shows, the recall@100 is very low, around 5\%. The second RS considered uses the same variables as those used to build the matching score, but instead of giving them fixed weights, it optimizes them to best predict the return to employment. This leads to improvements, but the recall@100 is still modest, remaining below 20\%. The last two RSs consider a broader set of variables. The first of the last two, based on neural networks, follows the method described in section \ref{sec:RSP} and is our $\mathcal{P}$-ranking. The second uses a different machine learning method based on ensembling and uses variables that explicitly describe the interactions between the variables characterizing the job supply and the job seekers (e.g. the distance between a job seeker and an establishment). Both RSs perform significantly better than the first two. The neural network achieves a recall@100 of about 57.5\% and the last one an even higher recall@100. The disadvantage of the last system is its speed, especially when making recommendations. The neural network model takes about one hour to train and about 0.07 seconds to generate a set of recommendations for a given job seeker; these figures are 2 hours and 10 seconds respectively for the last model.

\section{Estimation of the matching probability using the score as predictor}\label{sec:calibration}

Let $M_{i,j}^* \in \{0,1\}$ be the latent variable that takes the value 1 if there is a match for a pair of job seeker-firm $(i,j)$ after they meet. Importantly, after this process, the observed hiring dummy between $i$ and $j$ is $M_{i,j} = M_{i,j}^* C_{i,j}$.

\medskip

We want to characterize the true probability of $i$ being matched with $j$ conditional on the very rich information $X_{i,j} $ available to us, namely 
$\mathbb{P}(M_{i,j}^* = 1 \mid X_{i,j})$. There are three main difficulties in estimating this true conditional probability. First, there is a selection problem, since we only observe matches conditional on a past interview $C_{i,j}=1$, so the variable $M_{i,j} = M_{i,j}^* C_{i,j}$. Second, since we want to consider all potentially relevant covariates at our disposal, this is a high-dimensional setting.  Third, ML algorithms, and in particular those of section \ref{sec:RSP}, generally do not produce a consistent estimator of the matching probabilities, but provide excellent predictive performance of future matching leveraging the complex interactions between the components of $X_{i,j} $. We provide a framework that allows to estimate the best predictions of the matching probabilities with a logistic predictor, given the score $S_{i,j}$ produced by the RS.\footnote{Note that the objective function \eqref{eq:triplet} of this algorithm, whose purpose is to rank, is invariant to the addition of an individual specific effect $\alpha_i$. However, this does not change the interpretation of our object of interest here, which is the best predictor \textit{given the score and data used}. Of course, a different training could change the values of the estimated coefficients as the score would be different, but would marginally change the predicted probabilities which are our objects of interest. Alternatively, one could use a logit with fixed effects approach similar to the one of section  \ref{sec:empiric} to account for these potential shifts. However, this importantly limits the predictions to the set of ``movers" (here 427 individuals) and even on them we typically observe few applications in this period  (median of 3).}

\medskip

Denote by $\mathcal{F}_{j}$ the sigma-algebra generated by the vector of past applications up to $j$, i.e., $\{C_{i,k}=1, \ k=1,\dots,j\}$, and negative observed results $\{M_{i,k}=0, \ k=1,\dots,j\}$. The selection problem translates into the fact that our data only allows us to identify $P(M_{i,j(i)}=1|X_{i,j(i)},\mathcal{F}_{j(i)-1}, C_{i,j(i)}=1)$ instead of $P(M^*_{i,j(i)}=1|X_{i,j(i)})$. To deal with this selection problem, we make the following assumption of conditional independence of the matching $\{M^*_{i,j}, \ j\in \mathcal{J}\}$ and application processes $\{C_{i,j}, \ j\in \mathcal{J}\}$.

\begin{ass}[Selection on observables and Markovian property]\label{eq:model_lambda} 
$$\Theta(X_{i,j(i)} ):= \mathbb{P}(M_{i,j(i)}=1\mid X_{i,j(i)},\mathcal{F}_{j(i)-1}, C_{i,j(i)}=1)= \mathbb{P}(M^*_{i,j(i)}=1\mid X_{i,j(i)}).$$
\end{ass}

Given how large the set of covariates we are starting from is, this assumption makes sense. Then, we take advantage of observing the chronologically ordered sequence for an individual $i_0$, $1(i_0),2(i_0), \ldots,j^{max}(i_0)$ as a sequential search model and analyze it as a discrete duration model \citep[see, \emph{e.g.},][]{tutz2016modeling}, where the conditional hazard rate is $\Theta(X_{i,j(i)} )$.

Let us use the shortened notation for the score $ S_{i,j(i)}  := S_{i,j(i)}(X_{i,j(i)})$ and $r(i,j)$ for the rank of the vacancy $j$ in the application set. We define the best logistic predictor of this conditional probability given the score and this rank as $\Lambda (\alpha^*_{r(i,j(i))} +  \beta^*  S_{i,j(i)} )$, where $\Lambda$ is the usual logistic function, in the sense that $(\alpha^*_{r(i,j(i))}, \beta^*)$ minimizes the  Kullback Leibler Information Criterion (KLIC) with $\Theta(X_{i,j(i)} )$, see \cite{white1982maximum}.\footnote{Thus, \cite{white1982maximum} also suggest this is a ``minimum ignorance'' solution. When the model is correctly specified, this identifies the true parameters.} The parameters of this best logistic predictor $(\alpha^*_{r(i,j(i))}, \beta^*)$ can be consistently estimated using conditional maximum likelihood estimation (MLE).

\paragraph{Estimation.}

 For the estimation we use the sequence $\{M_{i,j(i)}\}_{i=1,\dots,N; j=1,\dots,n(i)}$, where $n(i)$ is the number of observed applications for job seeker $i$ and $N$ is the number of observed job seekers. Taking into account completed and censored spells \citep[see, \emph{e.g.},][page 52]{tutz2016modeling}, the estimation can be done using conditional MLE, considering the log-likelihood function, conditional on the scores produced by the RS, given by
\begin{align*}
    \mathcal{L}( \alpha,\beta | M , S) =
    & \sum_{i=1}^{N}  M_{i,n(i)}\ln(\Lambda( \alpha_{r(i,n(i))}  +  \beta  S_{i,n(i)}))\\
    & + \sum_{i=1}^{N} \sum_{j \in \mathcal{J}(i)\setminus\{n(i)\} } (1-M_{i,j})\ln(1- \Lambda( \alpha_{r(i,j)}  +  \beta  S_{i,j} )) .
\end{align*}

\medskip 

There is a simple generalization of the former expression to consider $r(i,j)$ the rank of vacancy $j$ in the application set of job seeker $i$, but also $q(i,j)$ the rank of $i$ in the applicant pool for job $j$. The likelihood expression in this case is written as
\begin{align*}
\mathcal{L}( \alpha,\beta | M , S) =
    & \sum_{(i,j): \ C_{i,j}=1} M_{i,j}\ln(\Lambda( \alpha^v_{q(i,j)}  +\alpha^{js}_{r(i,j)}  +\beta S_{i,j} ) \\
      & + \sum_{(i,j): \ C_{i,j}=1}(1-M_{i,j})\ln(1- \Lambda( \alpha^v_{q(i,j)} + \alpha^{js}_{r(i,j)}+  \beta  S_{i,j} )) ,
\end{align*}
where $\alpha^v$ and $\alpha^{js}$ are the sequences of ``weariness'' effects for vacancies and job seekers. 

\paragraph{Estimation results.} The estimation is performed on 34,255 randomly selected job seekers in the test set, representing 79,097 applications. As expected, the estimated coefficient of $\beta$ of 0.061 is significantly positive at the 1\% level. This result is robust to various specifications, including application and interview rank effects (the coefficient on $S_{i,j}$ drops to 0.038 and 0.047 , respectively, see table \ref{tab:calibration}). Overall, this validates the content of the ML score $S_{i,j}$ in terms of its potential to reflect hiring chances. From now on, instead of $S_{i,j}$, we will think of our estimated best logistic predictor given the score in column (1) of Table \ref{tab:calibration} as $\mathcal{P}_{i,j}:= \Lambda(0.061 \ S_{i,j} -4.113)$, which is our best prediction of the probability of hiring $p(i,j)$.

\begin{table}[!h]
\begin{center}
\caption{Estimates of the best logistic predictor of hirings given the ML score}
\label{tab:calibration}
\begin{threeparttable}
\begin{tabular}{lccc}
\toprule
\textbf{Method} & \textbf{(1)} & \textbf{(2)} & \textbf{(3)} \\ 
\midrule
\textbf{Score $S_{i,j}$ }& 0.061*** & 0.038*** & 0.047*** \\ 
& (0.0029) &  (0.0030) &  (0.0030)\\ 
\cmidrule{2-4}
With application rank &No&Yes&Yes\\
With interview rank &No&No&Yes\\
  \midrule
Intercept & -4.113*** & -2.994*** & -2.538*** \\ 
& (0.0559) &  (0.0570) &  (0.0575)\\ 
  \midrule 
AIC &  28,040  &  25,116 &  23,897\\ 
\bottomrule 
\end{tabular}
\begin{tablenotes}
 \item \footnotesize Notes: On a half of the job seekers present in the test sample (weeks 44-48 of 2019):    79,097 applications, 3,469 matches, 34,255 job seekers. Significance levels: $<1\%: {}^{***}$, $<5\%: {}^{**}$, $<10\%: {}^{*}$. ``With application rank" denotes dummies for the ranking of the application $j$ in the list of applications of job seeker $i$. ``With interview rank " denotes dummies for the ranking of the candidate $j$ in the list of recorded interviews for vacancy $j$.
 \end{tablenotes}
 \end{threeparttable}
 \end{center}
\end{table}

\section{Comparison of the two RSs obtained from $\mathcal{U}$ and $\mathcal{P}$ rankings}\label{app:comparison}

We apply the framework outlined in Section \ref{sec:sec1} to our two scores: the preference score $\mathcal{U}_{i,j}$, that we call hereafter the $\mathcal{U}$-ranking, and the hiring prediction score $\mathcal{P}$ derived from the ML-based estimate $S$. We refer to the latter as the $\mathcal{P}$-ranking hereafter. For each job seeker we define two sets of recommendations: a set based on the $\mathcal{U}$-ranking and one based on the $\mathcal{P}$-ranking. For each of these two scores, to make $k$ recommendations to job seeker $i_0$, the $k$ vacancies with the largest score $\mathcal{U}(i_0,j)$ (respectively $\mathcal{P}_{i_0,j}$) are selected.

\medskip

The right panel of Figure \ref{fig:perfs} shows how the recall of the last model varies with the number of recommendations. For 5 recommendations, the proportion is as large as almost 20\%. As shown in the figure, the proportion increases progressively when the number of recommendations increases. 

\begin{figure}[!ht]
    \centering
    \begin{tabular}{cc}
   \includegraphics[scale=0.45]{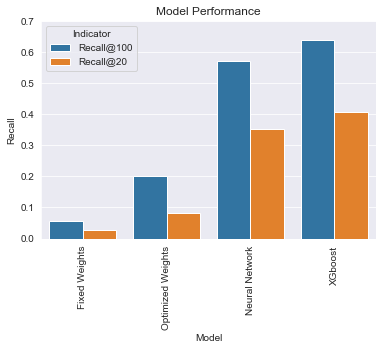}& \includegraphics[width=0.6\linewidth, height=0.3\textheight]{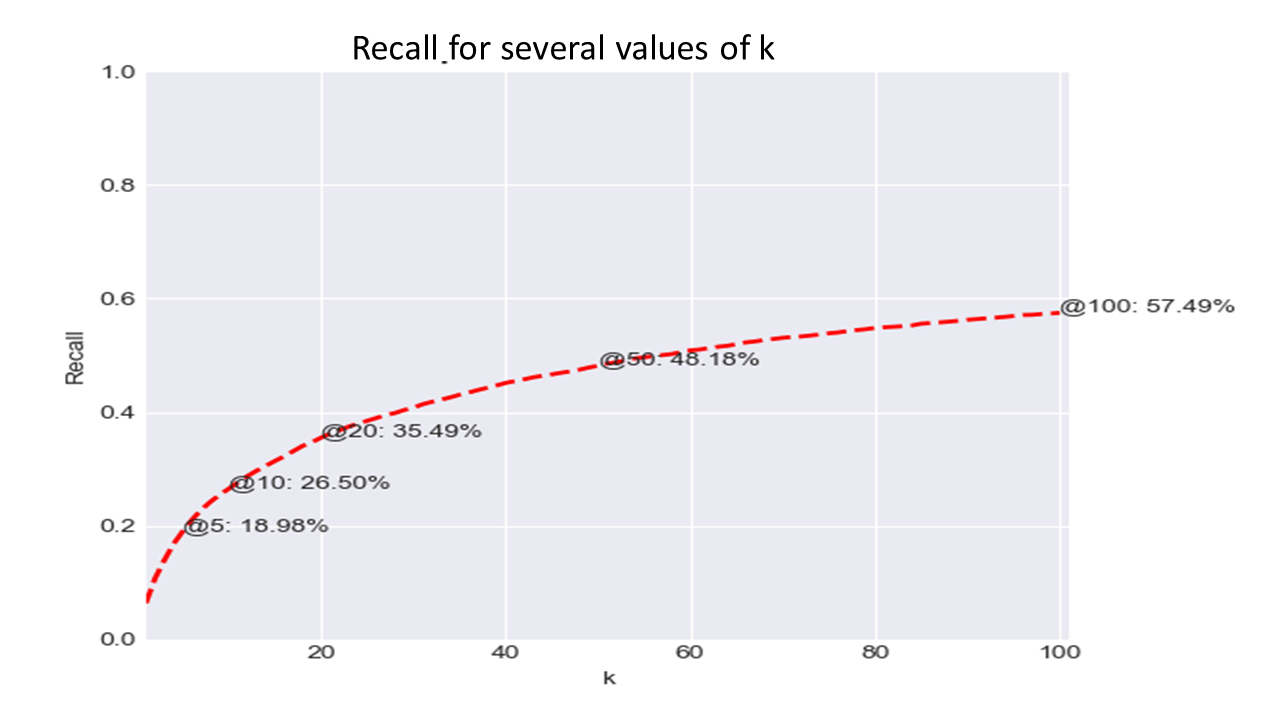}\\
   \end{tabular}
     \caption{Performance on the test set of different RSs.}
   \label{fig:perfs}
\end{figure}

The two rankings are positively correlated for a large part of the population (median at 0.14 and first and third quartiles at 0.10 and 0.19 respectively), but there are still significant differences in the ranking of vacancies between the two criteria. To explore this further, we compare for each type of job seeker $i$ the optimal vacancy based on the $\mathcal{P}$ ranking, denoted by $j^{\mathcal{P}}(i)$, and the optimal vacancy based on the $\mathcal{U}$ ranking, denoted by $j^{\mathcal{U}}(i)$. We first compare the respective ranks of these optimal vacancies: the rank of $j^{\mathcal{P}}(i)$ in the $\mathcal{U}$-ranking: $r^{\mathcal{U}}(i,j^{\mathcal{P}}(i))$, and symmetrically the rank of $j^{\mathcal{U}}(i)$ in the $\mathcal{P}$-ranking: $r^{\mathcal{P}}(i,j^{\mathcal{U}}(i))$.  The upper panel of Figure \ref{fig:dist_all} shows the distribution of these ranks. While some individuals have optimal recommendations according to the two ranks that match, this is a small minority. For most, the ranks considered are very large. The median of $r^{\mathcal{U}}(i,j^{\mathcal{P}}(i))$ is 381 (top 2\%) and that of $r^{\mathcal{P}}(i,j^{\mathcal{U}}(i))$ is 3,093 (top 16\%).

\begin{table}[!h]
\begin{center}
\caption{Estimates of the best logistic predictor of hirings and applications given the different scores, based on the second experiment}
\label{tab:calibration2}
\scalebox{0.95}{
\begin{threeparttable}
\begin{tabular}{lccc}
\toprule
\textbf{Outcome} & Hiring cond. application  & Application & Hiring\\ 
 & (1)  & (2)  & (3) \\ 
\midrule
Score \V.2  & 1.878*** &  0.252***  & 2.574*** \\ 
& (0.711) & (0.080)  & (0.850) \\ 
Score \V.0 & 0.196   &  0.044  & 0.288 \\ 
& (0.169)  & (0.033)  & (0.188) \\ 
Score APPLICATION & -0.015  &   0.680***  & 0.302* \\ 
& (0.296) & (0.052)  & (0.171) \\ 
Score U-Rec & -0.195 &   0.301*** & 0.134 \\ 
& (0.172) & (0.026) & (0.167) \\ 
  \midrule
Intercept & -4.728***  & -5.505*** &  -10.241*** \\ 
& (0.316)   &  (0.042)  & (0.362) \\ 
\bottomrule
\end{tabular}
\begin{tablenotes}
 \item \footnotesize Notes: Significance levels: $<1\%: {}^{***}$, $<5\%: {}^{**}$, $<10\%: {}^{*}$. The column ``Hiring cond. application'' considers hirings realizations conditional on applications while column ``Application'' considers applications to job postings conditional on clicks. Regressors are the scores corresponding to the different algorithms.  
 All scores are standardized to simplify the comparisons. Number of observations 309,730, of applications 1,752. Standard errors are clustered at the individual level. 
 \end{tablenotes}
 \end{threeparttable}
 }
 \end{center}
\end{table}

Thanks to the estimation result of the search model in section \ref{sec:calibration}, we give a further quantitative interpretation of these two systems. The lower panel of Figure \ref{fig:dist_all} shows the distribution of the hiring probabilities for the two vacancies: $\mathcal{P}(i,j^{\mathcal{P}}(i))$ and $\mathcal{P}(i,j^{\mathcal{U}}(i))$.
The median value of the maximum hiring probability for each individual $\mathcal{P}(i,j^{\mathcal{P}}(i))$ is 0.06, sharply contrasting with the hiring probability for the optimal vacancy according to the adequacy criterion (0.02). Although the probability of hiring from the best vacancy in the $\mathcal{P}$-ranking is higher than the probability of hiring from the $\mathcal{U}$-ranking, it is worth noting that this probability in absolute terms is not so large. Even more pronounced differences arise in the matching scores $\mathcal{U}(i,j^{\mathcal{U}}(i))$ and $\mathcal{U}(i,j^{\mathcal{P}}(i))$. As shown in Figure \ref{fig:dist_all}, the distribution $\mathcal{U}(i,j^{\mathcal{U}}(i))$ has a substantial mass at 1 (median 0.98), indicating that for many job seekers there are vacancies that meet all their criteria. Conversely, for the optimal vacancy according to the hiring probability, there is a significant mass at zero (median 0.46). 

\begin{figure}[!h]
 \centering
\includegraphics[width=0.95\linewidth, height=0.35\textheight]{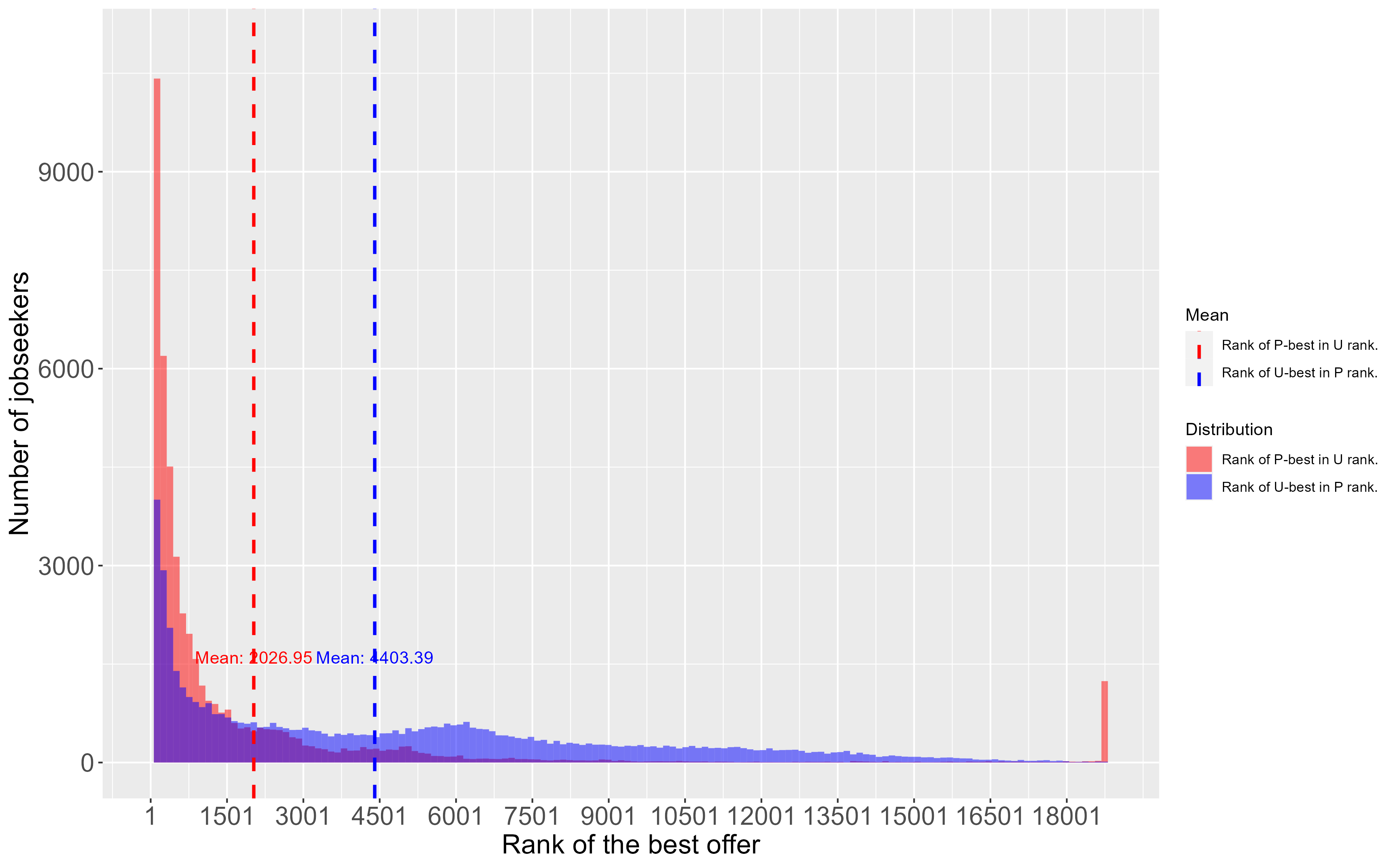}
\caption*{\footnotesize Notes: 60,299 job seekers whose main sector is transportation and logistics in the Rhône-Alpes region -- ISO weeks 44-48 of 2019 -- 18,873 vacancies available at that period in this sector.  Distributions of the ranks of the best recommendations based on past hirings ($\mathcal{P}$) and elicited preferences ($\mathcal{U}$) in each other rankings. The small bunch at the right gathers vacancies associated with some best recommendations according to $\mathcal{P}$ but ranked after 18,800 according to $\mathcal{U}$ as they have a preference score of $0$ and are ranked by distance to the job seekers. 
   	}
\caption{Comparison of the rankings of the best recommendations with respect to the other ranking}
\label{fig:ranks}
\end{figure}

\begin{figure}[!h]
 \centering
\subfigure[Distributions of hiring probabilities]{\label{fig:dist_a}\includegraphics[width=0.9\linewidth, height=0.3\textheight]{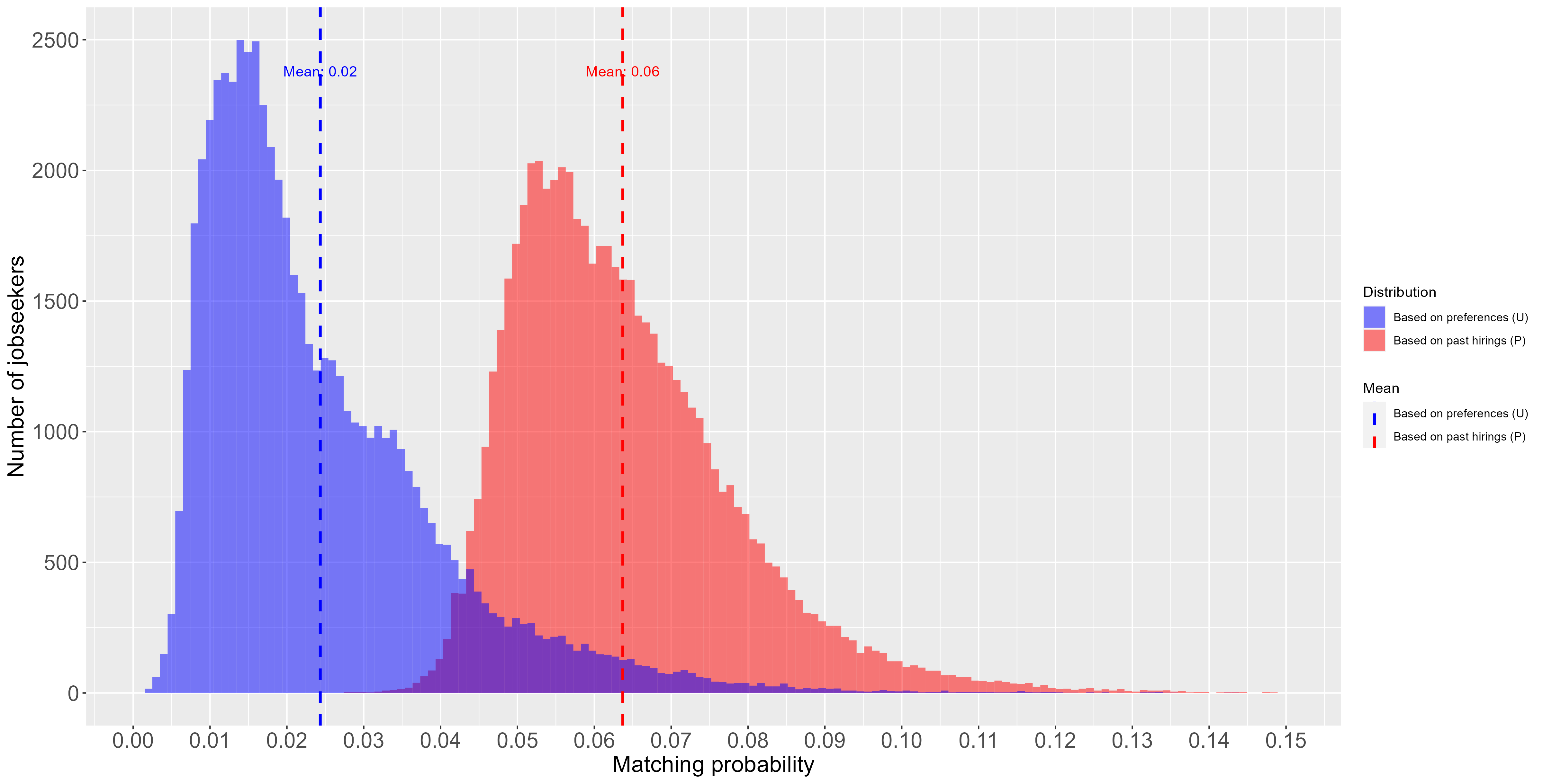}}
 \\
\subfigure[Distributions of preference score]{ \label{fig:dist_b}\includegraphics[width=0.9\linewidth, height=0.3\textheight]{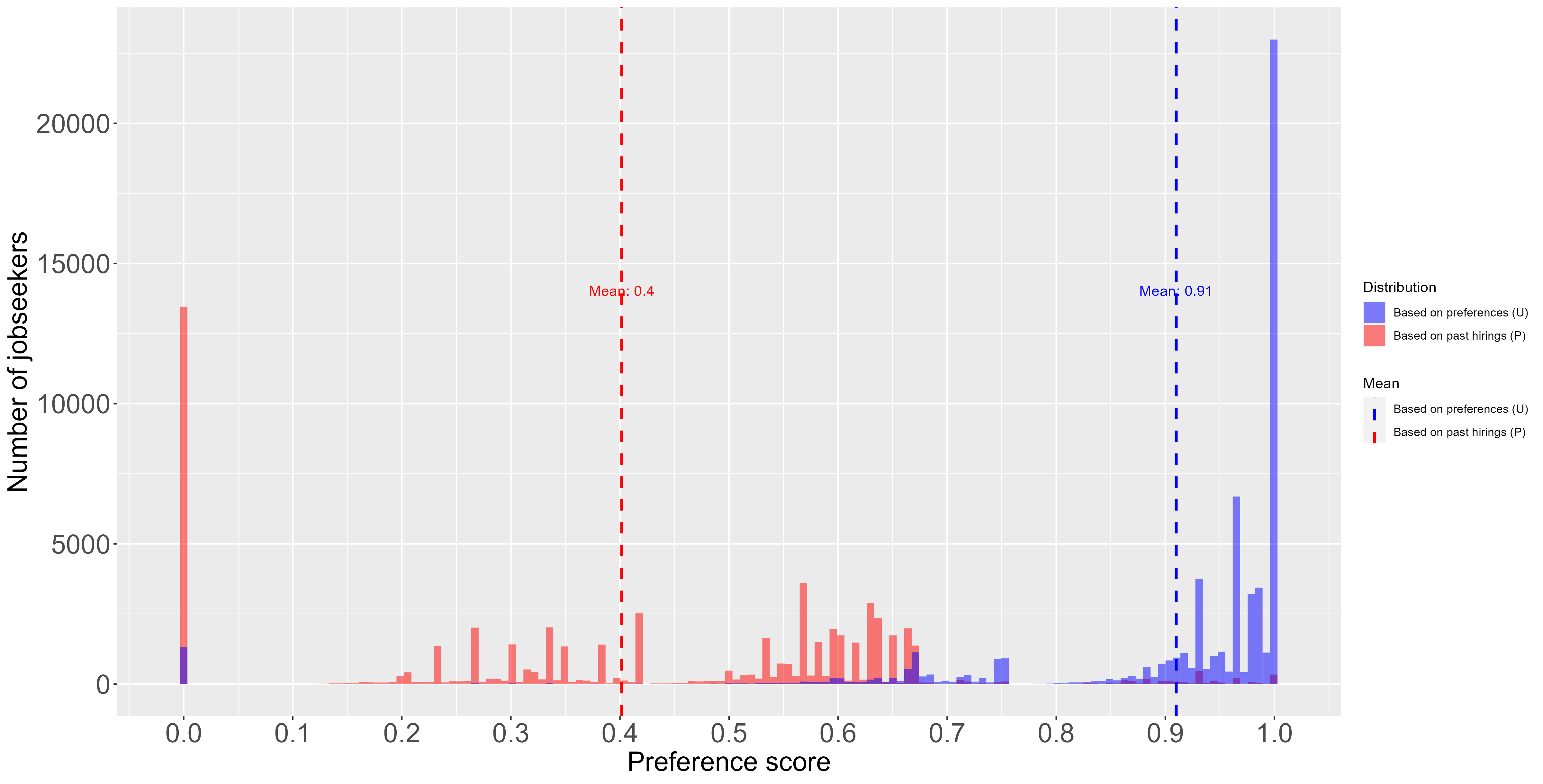}}
\caption*{\footnotesize Notes: 60,299 job seekers whose main sector is transportation and logistics in the Rhône-Alpes region -- ISO weeks 44-48 of 2019 -- 18,873 vacancies available at that period in this sector. 
\textit{Upper panel}: Histograms of the hiring probabilities for the best recommendations in both systems. \textit{Lower panel}: Histograms of the preference score for best recommendations in both systems.
   	}
\caption{Comparison of the best recommendations in the two rankings: hiring probabilities and preference score}
\label{fig:dist_all}
\end{figure}

\clearpage

\section{Randomized Experiment} \label{sec:betatests}

\begin{table}[h!]
\caption{Summary of the Two Field Experiments (2022 and 2023)}
\label{table:experiments2}
\centering
\renewcommand{\arraystretch}{1.25}
\scalebox{0.9}{
\begin{threeparttable}

\begin{tabular}{p{0.20\textwidth} p{0.38\textwidth} p{0.38\textwidth}}
\hline\hline
& \textbf{Experiment 1: March 2022} & \textbf{Experiment 2: June 2023}$^{(a)}$ \\
\hline

\textbf{Population} &
Job seekers registered at PES in Auvergne--Rhône--Alpes, actively seeking work. 
&
Same population and recruitment protocol as Experiment~1. \\

\textbf{Invited} &
102{,}314 job seekers 
&
150{,}000 job seekers \\

\textbf{Enrolled} &
18{,}947 (clicked consent link) 
&
30{,}973 (clicked consent link) \\

\textbf{Design} &
One email with recommendations; no control group; random assignment across treatment arms. 
&
Same core structure: email invitation, enrollment conditional on consent, random assignment. \\

\textbf{Treatment dimensions} &
(1) Algorithm generating recommendations;  
(2) Display of additional information (2 conditions). 
&
(1) Expanded set of algorithms;  
(2) Richer variation in displayed information (4 conditions). \\

\textbf{Treatment arms} &
10 total (5 algorithms $\times$ 2 display conditions).
&
15 total (6 algorithms + 3 information-display variants for selected algorithms). \\

\textbf{Algorithms tested} &
\V.0, \Pbs, and three hybrid Mix variants (Mix$^{1/4}$, Mix$^{1/2}$, Mix$^{3/4}$). 
&
\V.0(retrained), $\mathcal{U}$-rec, \V.2, Mix$^{1/2}$(\V.2), \textsc{Application}, and \textsc{XGBoost}. \\

\textbf{Information display} &
Baseline vs. additional information. 
&
No information; personalized ``scores''; ``explanation''; ``explanation+''. \\

\textbf{Outcomes} &
$\mathcal{U}$ and $\mathcal{P}$ scores; clicks; applications; hires.
&
Same outcome measures. \\

\hline\hline
\end{tabular}
\begin{tablenotes}
    \footnotesize \item (a) Actually, in addition to the 15 arms listed previously, the experiment includes a sixteenth arm where participants receive recommendations based on the score $\mathcal{U}^{PES}$, currently used by the PES. Its exact formula is unknown. It follows the same structure and weights as $\mathcal{U}$ but includes nonlinearities and strict exclusions (as already mentioned in footnote \ref{foot:U}). These features prevent some job seekers from receiving the 10 recommendations planned in the intervention, making full-sample comparisons infeasible. Specifically, 15.4\% of job seekers have no recommended offers, 35\% fewer than five offers, and 49\% fewer than ten offers. The score $\mathcal{U}$ provides a linear approximation of $\mathcal{U}^{PES}$, with the advantage of always generating 10 recommendations and excluding no job seeker.
\end{tablenotes}
\end{threeparttable}
}
\end{table}

\subsection{Population of interest} \label{sec:pop}

The eligible population are job seekers registered at \PE\ in the Auvergne-Rhône-Alpes region, of administrative category A (\emph{i.e.}, available for a job and looking for one), aged over 18 years old, and having given the PES the permission to contact them by email. Randomization was stratified by desired job type (14 modalities), the kind of support delivered by the institution (3 modalities describing the job seeker's degree of autonomy), and geographic location (level of a French \textit{département}, 12 modalities). Table~\ref{tab:balancetableexp_vad} reports balance checks for the first experiment across the five treatment arms.

\subsection{Treatments: 5 combinations of the two rankings $\mathcal{U}$ and $\mathcal{P}$}\label{sec:treatments}

We construct five ways to rank job vacancies based on the two scores, \(\mathcal{U}\) and \(\mathcal{P}\), by varying the weight given to each ranking. These hybrid rankings are implemented based solely on ranks, not on the actual values of the scores.  

We proceed as follows:  
\begin{enumerate}
    \item First, we define a ``\textit{consideration set}'' consisting of the top vacancies according to one or both rankings:  
\[
\mathcal{CS}=\{\text{Top25 }\mathcal{U}\} \cup \{\text{Top25 }\mathcal{P}\} \cup \{\{\text{Top50 }\mathcal{U}\} \cap \{\text{Top100 } \mathcal{P}\}\} \cup \{\{\text{Top100 }\mathcal{U}\}\cap\{\text{Top50 }\mathcal{P}\}\}
\]
We define \( L \) as the number of vacancies belonging to this recommendation set.\\  
\item To guard against vacancies disappearing from \PE website between the selection process and the time the email is sent, we select 15 vacancies from these \( L \). Only the top 10, according to the considered ranking and still available online at the time of sending, are included in the email.\\  

\item We rank the \( L \) vacancies in \(\mathcal{CS}\) according to \(\mathcal{P}\). The selected vacancies for the different ranking methods are as follows:  

\begin{itemize}
    \item \V : The top 15 vacancies from \(\mathcal{CS}\) ranked by \(\mathcal{P}\).
    \item \(\text{\MixQuarter}\) : The top 15 vacancies ranked by \(\mathcal{U}\) among the \(\max\{15,L/4\}\) vacancies from \(\mathcal{C}\mathcal{S}\) in the ranking of \(\mathcal{C}\mathcal{S}\) by \(\mathcal{P}\).
    \item \(\text{\MixHalf}\) : The top 15 vacancies ranked by \(\mathcal{U}\) among the \(\max\{15,L/2\}\) vacancies from \(\mathcal{C}\mathcal{S}\) in the ranking by \(\mathcal{P}\).
    \item \(\text{\MixThreeQuarter}\) : The top 15 vacancies ranked by \(\mathcal{U}\) among the \(\max\{15,3L/4\}\) vacancies from \(\mathcal{C}\mathcal{S}\) in the ranking by \(\mathcal{P}\).
    \item \(\text{\Pbs}\) : The top 15 vacancies from \(\mathcal{C}\mathcal{S}\) ranked by \(\mathcal{U}\).
\end{itemize}  
\end{enumerate}

\subsection{Design and Outcome variables}\label{sec:design} 

The first experiment was conducted in March 2022. 
Job-seekers are sent an email inviting them to complete an online survey using a link provided in the email that takes them to the survey's ``landing page.'' The landing page provides them with information about the goals of the survey and assures them that the information collected will be used for research purposes and will not affect their treatment by \PE. 

If they accept these conditions, job-seekers are first shown two job postings (the top 2 of their assigned algorithm). The  job postings  are characterized by the company, working conditions, salary, location (and distance), experience, education requirements, driver's license requirements, and a summary of the textual description of the job and the company. Job seekers are asked to rate the two job postings  in terms of i) global relevance, ii) their perception of their chances of being hired, and iii) fit with their job search criteria. They can also optionally provide natural language comments.

After rating the two job ads (which is required to proceed with the survey), job seekers are presented with an additional page that initially displays the two previous ads, but with a link to apply and three additional ads. Job seekers do not have to rate the ads on this page. They can click on the ads to view them on \PE's website (which provides more details about the ads and allows job seekers to apply if they wish). Job seekers' clicks on the ads are recorded. If they wish, job seekers can see an additional page with five more ads.

\begin{table}[!htb] \centering 
\caption{Balance check among full sample}  
\label{tab:balancetableexp_vad}
\footnotesize
\begin{tabularx}{\textwidth}{@{}l *6{>{\centering\arraybackslash}X}@{}}
\toprule
 & \V\ & \MixQuarter & \MixHalf & \MixThreeQuarter & \Pbs\ & \textit{p}\\
\midrule
\addlinespace[0.3em]
Age & 38.18 & 38.09 & 38.46 & 38.38 & 38.47 & 0.10\\
Looking for a permanent contract, full time & 0.65 & 0.65 & 0.65 & 0.65 & 0.64 & 0.66\\
Looking for a permanent contract, part time & 0.11 & 0.11 & 0.11 & 0.11 & 0.12 & 0.94\\
Looking for a temporary contract & 0.19 & 0.19 & 0.19 & 0.20 & 0.19 & 0.77\\
Education: High school & 0.25 & 0.26 & 0.26 & 0.25 & 0.26 & 0.56\\
Education: Less than high school & 0.10 & 0.10 & 0.10 & 0.10 & 0.10 & 0.71\\
Education: Vocational training & 0.28 & 0.27 & 0.28 & 0.27 & 0.27 & 0.25\\
Education: College Education & 0.37 & 0.37 & 0.37 & 0.37 & 0.37 & 0.95\\
Gender: Woman & 0.53 & 0.52 & 0.52 & 0.54 & 0.53 & 0.04\\
Level of assistance from the PES: Light & 0.25 & 0.25 & 0.25 & 0.25 & 0.25 & 1.00\\
Level of assistance from the PES: Medium & 0.55 & 0.55 & 0.55 & 0.55 & 0.55 & 1.00\\
Level of assistance from the PES: Strong & 0.19 & 0.19 & 0.19 & 0.19 & 0.19 & 1.00\\
Married & 0.42 & 0.41 & 0.43 & 0.43 & 0.42 & 0.04\\
Max. commuting time (minutes) & 23.4 & 23.8 & 23.6 & 23.6 & 23.4 & 0.28\\
No child & 0.57 & 0.58 & 0.57 & 0.57 & 0.57 & 0.95\\
Occupation targeted: Agriculture & 0.03 & 0.03 & 0.03 & 0.03 & 0.03 & 0.99\\
Occupation targeted: Art and crafts & 0.01 & 0.01 & 0.01 & 0.01 & 0.00 & 0.98\\
Occupation targeted: Banking, insurance, real est. & 0.01 & 0.01 & 0.01 & 0.01 & 0.01 & 0.95\\
Occupation targeted: Business support services & 0.16 & 0.16 & 0.16 & 0.16 & 0.16 & 0.90\\
Occupation targeted: Comm, media, digital & 0.02 & 0.02 & 0.02 & 0.02 & 0.02 & 0.99\\
Occupation targeted: Construction, public works & 0.07 & 0.06 & 0.07 & 0.06 & 0.07 & 1.00\\
Occupation targeted: Health & 0.04 & 0.04 & 0.04 & 0.04 & 0.04 & 1.00\\
Occupation targeted: Industry & 0.08 & 0.08 & 0.08 & 0.08 & 0.08 & 1.00\\
Occupation targeted: Maintenance & 0.04 & 0.04 & 0.04 & 0.04 & 0.04 & 1.00\\
Occupation targeted: Missing & 0.00 & 0.00 & 0.00 & 0.00 & 0.00 & 0.56\\
Occupation targeted: Performing arts & 0.02 & 0.02 & 0.02 & 0.02 & 0.02 & 0.99\\
Occupation targeted: Personal services & 0.19 & 0.19 & 0.19 & 0.19 & 0.19 & 1.00\\
Occupation targeted: Sales & 0.15 & 0.15 & 0.15 & 0.15 & 0.15 & 0.98\\
Occupation targeted: Tourism, leisure & 0.09 & 0.09 & 0.09 & 0.09 & 0.09 & 0.99\\
Occupation targeted: Transport & 0.10 & 0.10 & 0.10 & 0.10 & 0.10 & 0.98\\
Reservation wage (in euros) & 2702 & 2864 & 2799 & 2808 & 2838 & 0.52\\
Skill level: Higher occupation & 0.15 & 0.15 & 0.16 & 0.15 & 0.15 & 0.66\\
Skill level: Intermediate occupation & 0.70 & 0.69 & 0.69 & 0.69 & 0.69 & 0.85\\
Skill level: Lower occupation & 0.12 & 0.13 & 0.13 & 0.12 & 0.13 & 0.93\\
Skill level: Missing & 0.03 & 0.03 & 0.03 & 0.03 & 0.03 & 0.18\\
UB status: Not eligible to UB & 0.49 & 0.50 & 0.49 & 0.48 & 0.49 & 0.41\\
UB status: Receives UB & 0.51 & 0.50 & 0.51 & 0.52 & 0.51 & 0.41\\
Unemployment duration (in months) & 15.29 & 15.35 & 15.20 & 15.43 & 15.32 & 0.80\\
Work experience (in months) & 9.71 & 9.00 & 9.32 & 9.57 & 9.38 & 0.35\\
\midrule
N. Obs. & 10099 & 10092  & 10108 & 10094  & 10102  &\\
\bottomrule
\end{tabularx}
{\parbox{1\textwidth}{
\footnotesize \textit{Note:} Columns (1) to (5) characterize job-seekers by their treatment assignment and report mean values; shares of the sample for binary variables and levels for continuous variables (Age, Max. commuting time, Reservation wage, Unemployment duration, Work experience). Column \textit{p} reports the p-value from the F-test for joint significance of treatment coefficients in the regressions of each covariate on treatment assignment.}
}
\end{table}

\subsubsection{Attrition differential}

Table \ref{tab:betatesto:differentialattrition} displays the results of the regression:
\[
Y_i = \alpha + \sum_k \beta_k \{T_i = k\} + \epsilon_i
\]
among job seekers who received an email, where $T_i$ is job seeker $i$'s received treatment, and $Y_i$ corresponds to a binary indicator of having completed the survey (rated the top two ads and accessed the final page). The \V\ treatment serves as the reference category. A F-test of the joint nullity of coefficients associated to \Pbs, \MixQuarter, \MixHalf\ and \MixThreeQuarter\ yields a F-stat 1.885 (p=0.11). Accordingly, we do not attempt to model attrition differential.

\begin{table}[!htb]
\centering
\caption{Survey completion}
\label{tab:betatesto:differentialattrition}
\begin{tabular}{lc}
\toprule
Dependent variable & \textbf{Completed the survey} \\ \midrule 
\MixQuarter & \num{-0.003}   \\
& (\num{0.005})  \\
\MixHalf & \num{0.009}    \\
& (\num{0.006})* \\
\MixThreeQuarter & \num{0.001}    \\
& (\num{0.005})  \\
\Pbs & \num{-0.004}   \\
& (\num{0.005})  \\
\midrule
Strata fixed effects & Yes \\
N.Obs.                      & \num{50495}    \\
Control mean (\V) & 0.176 \\
\bottomrule
\end{tabular}
{\parbox{0.58\textwidth}{
\footnotesize \textit{Note:} The \V\ treatment group is used as the reference category. Robust standard errors are in parentheses. *, **, ***: significance at 10\%, 5\%, and 1\%. }
}
\end{table}

\section{Estimation of application and hiring probabilities from multiple algorithmic scores}\label{app:estimation_scores}

This appendix provides implementation details for the estimation strategy described in Section~\ref{sec:comparisonRS}. Throughout, we use data from Experiment~2. For each enrolled job seeker $i$ and each recommended vacancy $j\in\{1,\ldots,10\}$, we observe a vector of algorithmic scores

\[
\mathbf{S}_{i,j} \equiv
\big(
S^{\text{\textsc{Vadore}.0}}_{i,j},\,
S^{\text{\textsc{Vadore}.2}}_{i,j},\,
S^{\text{\textsc{Mix}}\sfrac{1}{2}(\text{\textsc{Vadore}.2})}_{i,j},\,
S^{\text{\textsc{Application}}}_{i,j},\,
S^{\text{\Pbs}}_{i,j}
\big).
\]
as well as an indicator of application $C_{i,j}\in\{0,1\}$ and, when $C_{i,j}=1$, an indicator of hire $H_{i,j}\in\{0,1\}$.

\paragraph{Step 1. Hiring conditional on applying.}
We estimate $p_{i,j}\equiv \mathbb{P}(H_{i,j}=1\mid C_{i,j}=1,\mathbf{S}_{i,j})$ on the subsample of recommended vacancies that received an application:
\begin{equation}
\mathbb{P}(H_{i,j}=1\mid C_{i,j}=1,\mathbf{S}_{i,j})
=
\Lambda(\alpha_h + \mathbf{S}_{i,j}'\boldsymbol{\beta}_h + \mathbf{X}_{i,j}'\boldsymbol{\gamma}_h),
\label{eq:hire_logit}
\end{equation}
where $\Lambda(t)=1/(1+e^{-t})$ and $\mathbf{X}_{i,j}$ denotes the set of controls used in the main specifications (e.g., recommendation-rank fixed effects, and any additional controls included in the empirical section). Predicted probabilities are denoted $\hat p_{i,j}$.

\paragraph{Step 2. Applications on recommended vacancies.}
We estimate $p_{a,i,j}\equiv \mathbb{P}(C_{i,j}=1\mid \mathbf{S}_{i,j})$ on the full sample of recommended vacancies:
\begin{equation}
\mathbb{P}(C_{i,j}=1\mid \mathbf{S}_{i,j})
=
\Lambda\!\Big(\alpha_a + \mathbf{S}_{i,j}'\boldsymbol{\beta}_a + \mathbf{X}_{i,j}'\boldsymbol{\gamma}_a\Big).
\label{eq:app_logit}
\end{equation}
Predicted probabilities are denoted $\hat p_{a,i,j}$.

\paragraph{Step 3. Reconstructing $\Gamma$ and related objects.}
Using the estimates from equations~\eqref{eq:hire_logit}--\eqref{eq:app_logit}, we compute
\[
\hat p_{h,i,j} \equiv \hat p_{i,j}\hat p_{a,i,j}.
\]
To reconstruct the welfare-relevant score $\Gamma_{i,j}$, we use the mapping implied by equation~\eqref{eq:gamma_plus}. In practice, we implement this by first mapping predicted application probabilities into a surplus index, and then applying the closed-form expression for $\Gamma$ under the logistic taste-shock specification and normalizing $\sigma$ to 1:
\[
\
\hat \Gamma_{i,j} = \hat p_{i,j}\ \left[- \log\!\left(1-\hat p_{a,i,j}\right)\right].
\]

\paragraph{Step 4. Counterfactual optimal recommendations.}
Let $\mathcal{J}_i$ denote the set of vacancies available to job seeker $i$ at the time of the experiment. For each $j\in\mathcal{J}_i$, we compute $\hat \Gamma_{i,j}$ and define the counterfactual top-ten set as the ten vacancies with highest $\hat \Gamma_{i,j}$. Denoting their scores by $\hat \Gamma_{i,1}^*,\ldots,\hat \Gamma_{i,10}^*$, we summarize distance to the optimum using $\hat \Gamma_{i,j}^*-\hat \Gamma_{i,j}$.

\paragraph{Step 5. Split-sample implementation.}
To avoid overfitting when using predicted probabilities to construct performance metrics, we randomly split the sample of enrolled job seekers into two halves, $S_1$ and $S_2$. We estimate equations~\eqref{eq:hire_logit}--\eqref{eq:app_logit} and construct $\hat \Gamma(\cdot)$ using $S_1$ only. We then evaluate average performance measures in $S_2$ by experimental group assignment.

\section{Supplementary tables and figures}

\footnotesize

\begin{table}[h!]
\centering
\scalebox{0.9}{
\renewcommand{\arraystretch}{1.25}
\begin{tabular}{p{0.10\textwidth} p{0.30\textwidth} p{0.30\textwidth} p{0.30\textwidth}}
\hline\hline
\toprule
& \textbf{Administrative Data} 
& \textbf{Experimental Data} 
& \textbf{Observational Data} \\
\midrule

\textbf{Contents} &
PES administrative data (Rhône-Alpes); vacancy attributes (occupation, salary, contract, location, hours, required skills, firm info, text descriptions, past applications); job seeker demographics, experience, skills, search parameters; clicks, applications, and hires (PES + DPAE).
&
Qualtrics-based experiments; each job seeker receives 10 recommendations from alternative RSs; observed clicks, applications, and hires; job seeker and vacancy characteristics; recovered ranking scores.
&
PES data for weeks 1–48 of 2019; stocks of job seekers and vacancies; 75{,}744 hires; training (weeks 1–43) and test (weeks 44–48) sets; subsamples for matching calibration, observational application model, and RSs ranking comparison (transport/logistics sector).
\\

\midrule

\textbf{Purpose} &
Train and validate RSs; construct features; measure interest, applications, and realized matches.
&
Compare algorithms in a controlled environment; evaluate behavioral responses; assess click/application/hire performance.
&
Calibrate matching probabilities; estimate application behavior; evaluate RSs out-of-sample; compare ranking performance.
\\

\bottomrule
\end{tabular}
}
\caption{Summary of the Three Types of Data Used}
\label{table:desc_data}
\end{table}


\begin{figure}[!ht]
    \centering
    \includegraphics[width=0.8\linewidth]{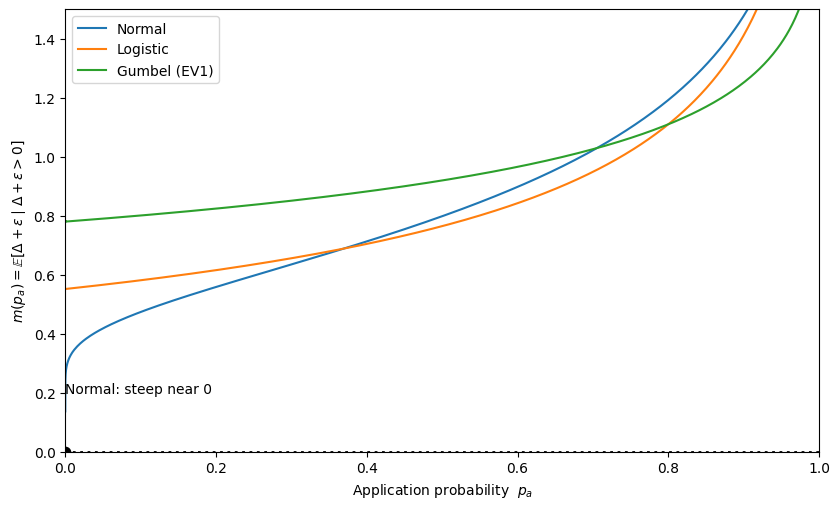}
    \caption{Expected surplus conditional on application}
    \label{graph:app:lois_epsilon}
\parbox{0.95\textwidth}{
\footnotesize \textit{Notes:} This figure plots the conditional expected surplus
$\mathbb{E}\!\left[\Delta(p,U)+\varepsilon \mid \Delta(p,U)+\varepsilon>0\right]$
as a function of the application probability
$p_a(p,U)=\mathbb{P}\!\left(\Delta(p,U)+\varepsilon>0\right)$.
The figure illustrates how the magnitude of the surplus conditional on application
varies across different distributions of the taste shock $\varepsilon$
(logistic, Gumbel (EV1), and normal), all normalized to have unit variance.}
\end{figure}

\begin{figure}[!ht]
\centering
  \centering
  \resizebox{\columnwidth}{!}{%
    \input{dessins/Muse0}
    }
  \caption{\V.0\ architecture: three embeddings are defined to model geographical, skills and general aspects of job seekers (left) and job ads (right), and compute the hiring score.}
  \label{fig:nn}
\end{figure}
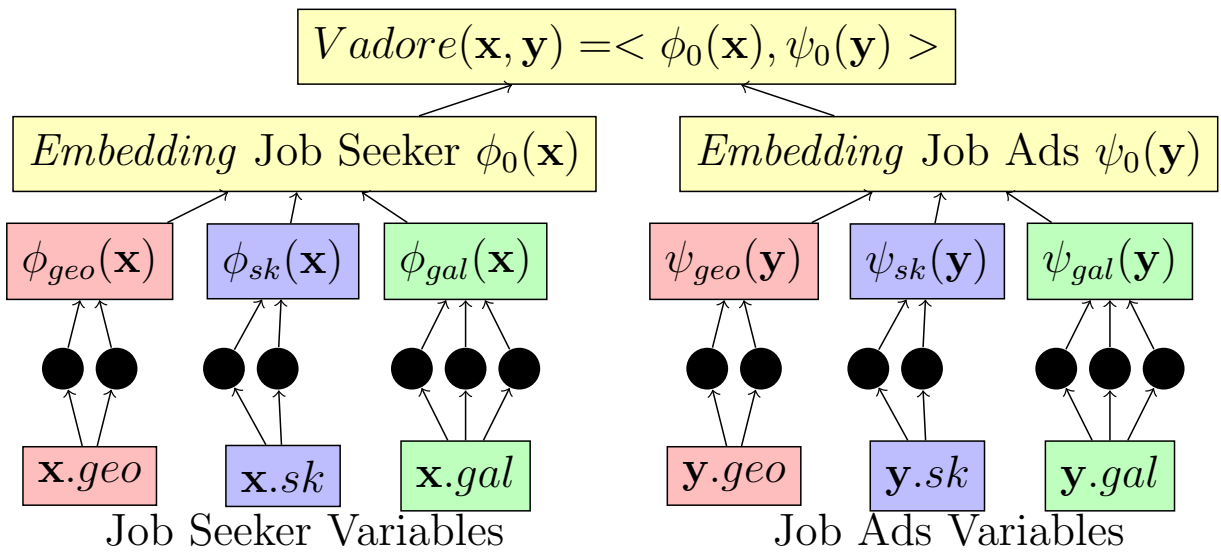
\begin{figure}[!ht]
  \centering
  \resizebox{\columnwidth}{!}{%
    \input{dessins/Muse1}
    }
  \caption{\Vu\ (below dashed line) and \Vdt\ architectures. \Vdt\ includes a second-head to model the applications, and a top head, exploiting both the standalone hiring and the application scores to predict the overall hiring score. }
  \label{fig:nn1}
\end{figure}
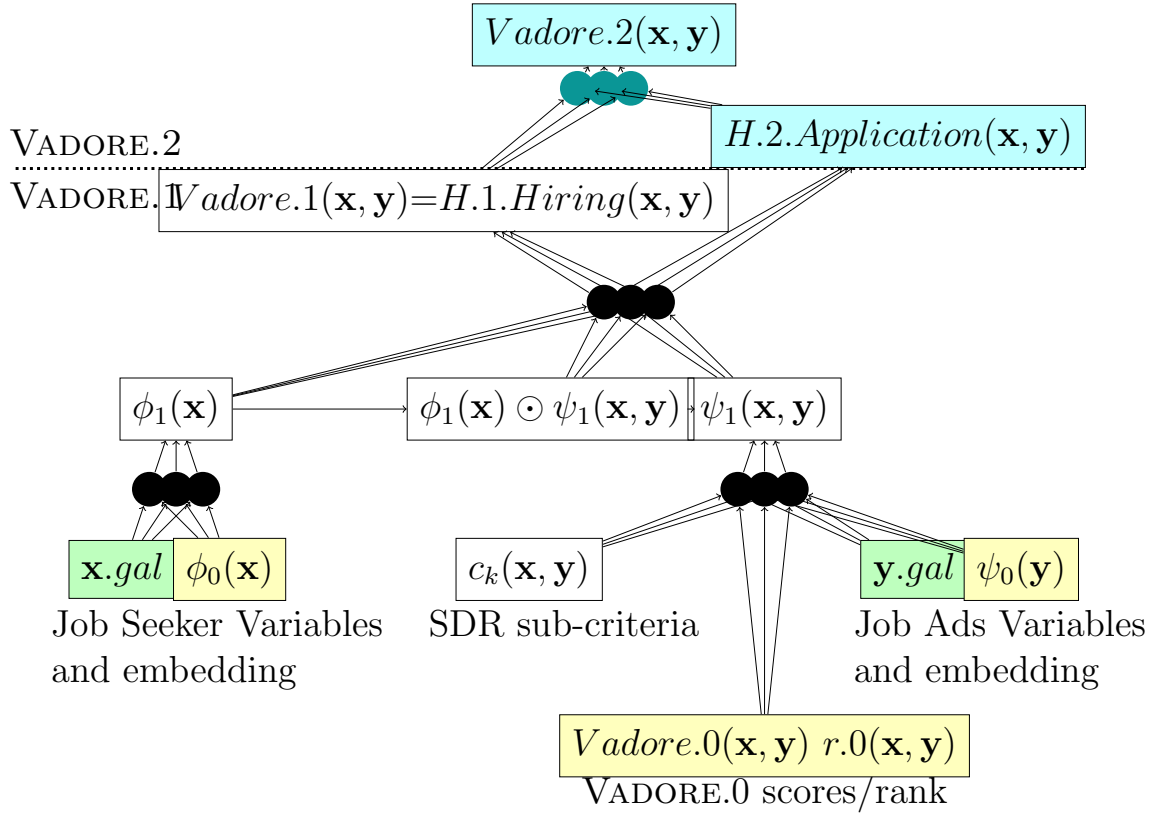

\FloatBarrier

\begin{table}[H]
\centering
\caption{First experiment— Reduced-form estimates: effects of algorithm assignment on $\mathcal{U}$ and $\mathcal{P}$ scores, clicks, applications, and subjective ratings}
\label{tab:reduced}
\scalebox{1.0}{
\begin{tabular}{lccccc}
\toprule
 & $\mathcal{U}$ & $\mathcal{P}$ & Rating & Clicks & Applications \\
 &  & ($\times$100) &  & ($\times$100) & ($\times$100) \\
\midrule
\V.0 & -0.1919*** & 4.602*** & 0.149** & 0.463** & 0.073* \\
 & (0.0039) & (0.076) & (0.054) & (0.207) & (0.040) \\
\MixQuarter & -0.1757*** & 3.863*** & 0.214*** & 0.432** & 0.053 \\
 & (0.0040) & (0.070) & (0.054) & (0.209) & (0.039) \\
\MixHalf & -0.1584*** & 3.573*** & 0.241*** & 0.640*** & 0.065 \\
 & (0.0041) & (0.068) & (0.053) & (0.211) & (0.040) \\
\MixThreeQuarter & -0.0477*** & 1.474*** & 0.095* & 0.175 & 0.038 \\
 & (0.0043) & (0.069) & (0.054) & (0.206) & (0.038) \\
\midrule
\Pbs\ (mean) & 0.773*** & 6.445*** & 4.169*** & 4.216*** & 0.455*** \\
 & (0.003) & (0.050) & (0.038) & (0.199) & (0.054) \\
\midrule
Observations & 189,470 & 189,470 & 36,668 & 189,470 & 189,470 \\
\bottomrule
\end{tabular}
}
\parbox{0.95\textwidth}{
\footnotesize \textit{Notes:} \Pbs\ is the omitted reference group. All regressions control for strata fixed effects. Standard errors clustered at the job seeker level. *, **, and *** denote significance at the 10\%, 5\%, and 1\% levels, respectively.
}
\end{table}

\begin{table}
\centering
\caption{Second experiment - Reduced-Form estimation - impact of assignment on $\mathcal{U}$ and $\mathcal{P}$ scores and clicks and applications on the 10 first ads }\label{tab:reduced2}
\		\scalebox{0.8}{			
\begin{tabular}[t]{lccccccc}
\toprule
& (1) & (2) & (3) & (4)   & (5) & (6)  & (7) \\
 & $\mathcal{U}$ &$\mathcal{P}$ & Rating & Clicks & Applications & Hirings & Hiring rates \\
 &  &  ($\times$100)  &  &  ($\times$100)  &   ($\times$100)  &  ($\times$10,000) &  (6) / (5), levels \\
\midrule
\V.0
& -0.0922*** & 3.624*** & 0.298*** & 0.747*** & 0.086 & 0.017 & 0.019 \\
& (0.0063) & (0.072) & (0.065) & (0.232) & (0.066) & (0.696) &  \\

\V .2
& -0.0545*** & 2.605*** & 0.637*** & 1.173*** & 0.285*** & 1.173$\times 10^{-4}$ & 0.010 \\
& (0.0053) & (0.047) & (0.051) & (0.178) & (0.053) & (0.541) &  \\

\MixHalf(\V .2)
& -0.0301*** & 2.324*** & 0.580*** & 1.138*** & 0.219*** & 1.064 & 0.040 \\
& (0.0053) & (0.046) & (0.051) & (0.180) & (0.052) & (0.647) &  \\

\textsc{Application}
& -0.0397*** & 2.288*** & 0.802*** & 1.279*** & 0.378*** & 0.729 & 0.022 \\
& (0.0053) & (0.046) & (0.050) & (0.180) & (0.056) & (0.618) &  \\

\textsc{XGBoost}
& -0.0045 & 2.054*** & 0.400*** & 0.998*** & 0.221** & 1.005 & 0.039 \\
& (0.0069) & (0.065) & (0.064) & (0.243) & (0.074) & (0.986) &  \\

\midrule
Average, \Pbs
& 0.429*** & 3.901*** & 4.941*** & 6.067*** & 0.148*** & 0.420 & 0.004 \\
& (0.0050) & (0.040) & (0.046) & (0.201) & (0.053) & (0.698) &  \\

\midrule
Number Obs.
& 309,730 & 309,730 & 123,875 & 309,730 & 309,730 & 309,730 &  \\
\bottomrule
\end{tabular}
		}
{\parbox{1\textwidth}{
\footnotesize \textit{Note:} The \Pbs\ treatment group is used as the reference category and we control for the ad position in the display. Column (7) represents the ratio of the number of hirings (6) over the number of applications (5) on the recommended vacancies. Standard errors are clustered at the job seeker level. *, **, ***: significance at 10\%, 5\%, and 1\%.  }
}
\end{table}

	\begin{table}
\centering
\caption{First beta-test - Reduced-Form estimation - impact of assignment on $\mathcal{U}$ and $\mathcal{P}$ scores and clicks and applications on the 10 first ads }\label{tab:reduced_all}
\		\scalebox{0.85}{
	\begin{tabular}[t]{lccccc}
				\toprule
				& $\mathcal{U}$ &$\mathcal{P}$ & Note  & Clicks & Applications\\
                   		&  &  &  x100  &    x100  &   x100 \\
				\midrule			
\V.0 & -0.1875*** & 4.519*** & 0.118 & 0.599** & 0.101*\\
 & (0.0055) & (0.108) & (0.075) & (0.300) & (0.058)\\
\MixThreeQuarter & -0.1743*** &  3.747***& 0.085  & 0.531* & 0.054\\
 & (0.0057) & (0.098) &  (0.075)& (0.297) & (0.056)\\
	\MixHalf & -0.1557*** & 3.462*** &   0.136* &  0.788** &0.068\\
 & (0.0058) & (0.095) & (0.074) & (0.298) & (0.058)\\
    \MixQuarter & -0.0485*** & 1.406*** & 0.020  & 0.101 &  -0.010\\
 & (0.0061) & (0.098) & (0.075)& (0.294) & (0.053)\\
\V.0, scores & -0.1952*** &  4.479*** &  -0.097 &0.212 & 0.011\\
 & (0.0054) & (0.108) &  (0.075)   & (0.284) & (0.054)\\
\MixThreeQuarter, scores & -0.1760*** & 3.775***& 0.065  &  0.219 & 0.017\\
 & (0.0057) & (0.010) & (0.075)   & (0.292) & (0.053)\\
	\MixHalf, Score & -0.1600*** & 3.479*** & 0.066 & 0.366 &  0.027\\
 & (0.0058) & (0.099) & (0.075)& (0.295) & (0.053)\\
    \MixQuarter, scores & -0.0457*** & 1.335***  & -0.109 & 0.136 &  0.052\\
 & (0.0061) & (0.097) & (0.076) & (0.286) & (0.054)\\
\Pbs, scores & 0.0012 & -0.204*  &  -0.275***   & -0.114   &  -0.035\\
 & (0.0061) & (0.091)  &  (0.076)  & (0.285) & (0.052)\\
					\midrule
				Average,  \Pbs		& 0.633*** &   5.584*** &  4.307***   &    3.844***  &    0.224*** \\
                		& (0.004) &  (0.065)  &  (0.052) &    (0.200)   &   (0.036) \\
				\midrule
				Number Obs. & 189,470 & 189,470 & 36,668   & 189,470 & 189,470\\
				\bottomrule
			\end{tabular}
}
{\parbox{1\textwidth}{
\footnotesize \textit{Note:} The \Pbs\ treatment group is used as the reference category. Standard errors are clustered at the job seeker level. *, **, ***: significance at 10\%, 5\%, and 1\%. }
}
\end{table}

	\begin{table}
\centering
\caption{Second experiment - Reduced-Form estimation by display condition}\label{tab:reduced3}
\		\scalebox{0.8}{			
\begin{tabular}[t]{lcccccc}
				& $\mathcal{U}$ &$\mathcal{P}$  &   Notes & Clicks & Applications & Hirings  \\
       		&    &    x100 &  & x100    &  x100 &  x10,000  \\
\midrule				
\V.0  &-0.0922*** & 3.624*** & 0.298*** & 0.747*** & 0.086 & 0.017 \\
  & (0.0063) & (0.072) & (0.0645) & (0.232) & (0.0659) & (0.696) \\
Application &  -0.0388*** & 2.276*** & 0.766*** & 1.737*** & 0.392*** & 0.518 \\
 & (0.0063) & (0.065) & (0.0626) & (0.262) & (0.0835) & (0.858) \\
\MixHalf \ with \V .2 & -0.0298*** & 2.354*** & 0.515*** & 1.588*** & 0.261*** & 0.954 \\
 & (0.0066) & (0.062) & (0.0628) & (0.262) & (0.0761) & (0.961) \\
\V .2  & -0.0523*** & 2.589*** & 0.547*** & 1.818*** & 0.267*** & -0.484 \\
 & (0.0063) & (0.066) & (0.0626) & (0.264) & (0.0776) & (0.484) \\
XGB &-0.0045 & 2.054*** & 0.400*** & 0.998*** & 0.221*** & 1.005 \\
 & (0.0069) & (0.065) & (0.0642) & (0.243) & (0.0742) & (0.986) \\
\V .2, scores &-0.0556*** & 2.639*** & 0.689*** & 0.962*** & 0.294*** & -0.484 \\
 & (0.0062) & (0.066) & (0.0625) & (0.242) & (0.0856) & (0.484) \\
Application, scores & -0.0344*** & 2.297*** & 0.804*** & 1.404*** & 0.299*** & 1.362 \\
 & (0.0063) & (0.063) & (0.0618) & (0.247) & (0.0760) & (1.042) \\
\MixHalf \ with \V .2, scores & -0.0333*** & 2.279*** & 0.585*** & 0.933*** & 0.248*** & 0.437 \\
 & (0.0064) & (0.062) & (0.0620) & (0.239) & (0.0749) & (0.812) \\
Application, explanation &-0.0428*** & 2.302*** & 0.837*** & 0.874*** & 0.391*** & 0.496 \\
 & (0.0063) & (0.063) & (0.0630) & (0.246) & (0.0884) & (0.845) \\
\MixHalf  \ with \V .2 , explanation &  -0.0301*** & 2.354*** & 0.585*** & 1.019*** & 0.174** & 2.418* \\
 & (0.0066) & (0.063) & (0.0642) & (0.250) & (0.0697) & (1.279) \\
\V .2, explanation & -0.0589*** & 2.565*** & 0.653*** & 0.987*** & 0.291*** & 0.924 \\
 & (0.0062) & (0.064) & (0.0636) & (0.234) & (0.0741) & (0.946) \\
Application,  explanation + & -0.0431*** & 2.274*** & 0.801*** & 1.104*** & 0.436*** & 0.497 \\
 & (0.0064) & (0.064) & (0.0629) & (0.255) & (0.0929) & (0.846) \\
\MixHalf  \ with \V .2 ,  explanation + &-0.0269*** & 2.312*** & 0.636*** & 1.017*** & 0.192** & 0.480 \\
 & (0.0065) & (0.062) & (0.0633) & (0.252) & (0.0809) & (0.836) \\
\V .2,  explanation + & -0.0512*** & 2.626*** & 0.657*** & 0.944*** & 0.286*** & 0.005 \\
 & (0.0063) & (0.065) & (0.0642) & (0.248) & (0.0753) & (0.688) \\        
				\hline
                Average,  \Pbs		&   0.427***  &  3.918***  & 4.941***  & 2.232***      & 0.310***  &   0.484***   \\
       		&  (0.005) &  (0.039)  & (0.046)  & (0.149)    &  (0.042) &  (0.707)  \\
				\midrule
				Num.Obs. &  309,730 & 309,730  &  123,865  & 309,730 & 309,730 & 30,973\\
				\bottomrule
			\end{tabular}

		}
{\parbox{1\textwidth}{
\footnotesize \textit{Note:} The \Pbs\ treatment group is used as the reference category. Standard errors are clustered at the job seeker level. *, **, ***: significance at 10\%, 5\%, and 1\%. }
}
\end{table}

\begin{table}
	\begin{center}
	\caption{Second beta-test - Reduced-Form estimation - impact of the assignment on $\Gamma$ for the 10 first ads }\label{tab:evaluation1}
	\scalebox{0.9}{			
\begin{tabular}{lcccc}
\toprule
  & $p$  & $p_a$ & $\Gamma$    & $p_h=p \times p_a$   \\
    & ($\times 100$) &  ($\times 100$) & ($\times 10,000$) & ($\times 10,000$) \\
    &  (1)  & (2) & (3)    &  (4)  \\
\midrule
 \Pbs & 0.431 & 0.250 & 0.323 & 0.320 \\ 
  & [0.411,0.452] & [0.242,0.259] & [0.292,0.352] & [0.290,0.349] \\ 
\V.0  & 1.375 & 0.442 & 0.848 & 0.843 \\ 
   & [1.356,1.396] & [0.433,0.451] & [0.817,0.879] & [0.813,0.874] \\ 
 \V.2 & 1.867 & 0.600 & 1.221 & 1.214 \\ 
   & [1.856,1.877] & [0.596,0.604] & [1.205,1.236] & [1.198,1.230] \\ 
\MixHalf(\V .2) & 1.535 & 0.551 & 0.990 & 0.985 \\ 
   & [1.525,1.546] & [0.546,0.555] & [0.974,1.005] & [0.969,0.999] \\ 
  \textsc{Application} & 1.355 & 0.676 & 1.116 & 1.109 \\ 
 & [1.345,1.365] & [0.672,0.680] & [1.100,1.134] & [1.094,1.126] \\ 
   \textsc{XGBoost} & 1.140 & 0.489 & 0.825 & 0.820 \\ 
  & [1.123,1.161] & [0.480,0.498] & [0.794,0.858] & [0.790,0.852] \\
  \cmidrule{2-5}
Average with $\Gamma$ & 1.939 & 0.656 & 1.349 & 1.342 \\ 
& [1.935,1.943] & [0.654,0.658] & [1.341,1.357] & [1.334,1.350] \\ 
\bottomrule
\end{tabular}	}
    \end{center}
    \vspace{0.2cm}
\footnotesize{\textit{Note:} 	Number of observations:  309,730. We use 500 split-sampling in half to estimate separately the probabilities $p$ and $p_a$ on the first part and to perform these regressions on the second one. Following \cite{chernozhukov2018generic}, we report the median over the splits of the estimated coefficients and lower/upper confidence bounds at the 95\% level, which gives the 90\% confidence intervals displayed in brackets. These take into account uncertainty conditional on the split as well as uncertainty due to split-sampling. These are the averages over the 500 splits.}
\end{table}

\begin{table}[ht]
\begin{center}
	\caption{Second beta-test - Average differences between the $\Gamma$-optimal recommendation set and the recommendation set for a given RS using the $\Gamma$ metric}\label{tab:evaluation2}
	\scalebox{1}{			
\begin{tabular}{lrll}
  \hline
 Algorithm & Estimate & CI 95\% & CI 99\% \\ 
  \hline
 \textsc{Application} & 0.23 & [0.227, 0.243] & [0.225, 0.245] \\ 
 \Pbs & 1.02 & [0.971, 1.069] & [0.957, 1.080] \\ 
\MixHalf(\V .2) & 0.35 & [0.337, 0.360] & [0.334, 0.365] \\ 
 \V.0 & 0.46 & [0.429, 0.492] & [0.422, 0.503] \\ 
 \V.2 & 0.14 & [0.130, 0.142] & [0.128, 0.144] \\ 
  \textsc{XGBoost} & 0.51 & [0.488, 0.541] & [0.483, 0.549] \\ 
   \hline
\end{tabular}}
    \end{center}
    \vspace{0.2cm}
\footnotesize{\textit{Note:} 	Number of observations:  309,730. We use 500 split-sampling in half to estimate separately the probabilities $p$ and $p_a$ on the first part and to perform these regressions on the second one. Following \cite{chernozhukov2018generic}, we compute the median over the splits of the estimated coefficients and lower/upper confidence bounds obtained using 1000 bootstrap replications at the 97.5\% and 99\% levels, which gives the 95\% and 99\% confidence intervals respectively, displayed in brackets. These take into account uncertainty conditional on the split as well as uncertainty due to split-sampling.}  These are the averages over the 500 splits.
\end{table}

\normalsize

\section{Estimating application behavior on observational data}\label{app:modelobservational}

As a robustness check, we also examine an alternative specification using observational data. We rely on data from the monitoring of job seekers’ search activity (see dedicated paragraph below). All job postings on which a job seeker has clicked are identified and stored, along with subsequent actions—particularly whether an application was submitted.

For each of these postings, we compute the indicators $\mathcal{U}$ and $\mathcal{P}$. We then estimate equation \eqref{eq:app2} directly using this observational dataset, with a logit model with or without fixed effects.
 
Results appear in Table \ref{tab:logit}. The first and second columns (Logit) and (FE-Logit) of table \ref{tab:logit} present the results with and without fixed effects, respectively, while the third column (FE-Logit unconstr.) presents results in which instead of using $\mathcal{U}$ with its PES-given weights, we estimate them. For each of the three estimates, the variable $-1/\mathcal{P}_{i,j}$ has a significant positive coefficient (which implies, as expected, that the probability of applying increases with $\mathcal{P}$). Moreover, these coefficients are very similar: 0.018 for the logit, 0.028 for the logit FE, and 0.026 for the logit FE with estimated weights. Similarly, the utility score coefficient $\mathcal{U}$ has a positive and significant coefficient in each of the first two specifications and the values are also very close, respectively 0.992 and 1.101. As stressed above, the result in the second column is especially important as it validates the interpretation we sketched at the beginning of section \ref{sec:search}, of the score $\mathcal{U}_{i,j}$ as a signal of the utility gap $U-U^*$ and of the probability $\mathcal{P}$ as a signal of the chances of success of an application, and that both scores must be taken into account to design an optimal RS. In addition, consistent with intuition, in the last column the fit between job seekers' parameters and vacancies in terms of occupation, reservation wages, skills, diplomas, and geographic mobility significantly predicts that an application is more likely. The only unexpected result here is that the fit in terms of experience in the occupation seems to enter negatively into the application decision.\footnote{The introduction of fixed effects forces to restrict the sample to so-called ``movers'' for whom at least two clicks are observed, including at least one application and one non-application. In order to track the changes due to the different specifications and the different sample, the appendix table \ref{tab:logit} compares the results of the model with fixed effects on movers (column 3) to those with uniform weights (column 2) as well as the results without fixed effects and on the whole population (column 1). Despite the sharp reduction in the number of observations used between these columns the results are close. The table shows the robustness of the result for $1/\mathcal{P}_{i,j}$. The estimated coefficients are all negative, as expected, and very close to each other.}

\begin{table}
\begin{center}
\caption{Estimates of the model of application on job postings}
\label{tab:logit}
\scalebox{0.8}{
\begin{threeparttable}
\begin{tabular}{lcccccccc}
\toprule
 & \multicolumn{2}{c}{(1)} &  & \multicolumn{2}{c}{(2)} &   & \multicolumn{2}{c}{(3)}   \\ 
  & \multicolumn{2}{c}{Logit} &  & \multicolumn{2}{c}{FE-Logit} &   & \multicolumn{2}{c}{FE-Logit unconstr.}   \\ 
 \cmidrule{2-3} \cmidrule{5-6} \cmidrule{8-9}
 & Estimate & Std. error &  & Estimate & Std. error &   & Estimate & Std. error   \\ 
\midrule
Utility score $\mathcal{U}_{i,j}$ ($\alpha$) &  0.992${}^{**}$   &  0.194&  &  1.101${}^{***}$    &   0.155 &  &  &   \\ 
Occupation &   & & &   &  &  &0.582${}^{***}$ & 0.104    \\
Skills &   & & & &    &  &  0.175${}^{*}$ & 0.114  \\
Reservation wage &   & & & &   &  & 0.236${}^{***}$ & 0.082  \\
Languages &   & & & &   &  &   -0.010 & 0.229  \\
Experience in occ.  &   & & & &    &  &  -1.017${}^{***}$ & 0.339    \\
Diploma &   & & & &   &  &  0.288${}^{**}$  & 0.118    \\ 
Driving license &   & & &  &   &  & 0.106 & 0.097   \\
Geographic mobility  &   & & &  &   &  & 0.625${}^{***}$ & 0.214   \\
Duration  &   & & &  &   &  & 0.139 & 0.068  \\
Type of contract  &   & & & &    &  &0.015 & 0.004   \\
 Inverse of $\mathcal{P}_{i,j}$ ($\beta$)  &-0.018${}^{**}$ & 0.007  &  & -0.028${}^{***}$ & 0.004   &  & -0.026${}^{***}$ & 0.004    \\ 
\cmidrule{2-3} \cmidrule{5-6} \cmidrule{8-9}
Avg. indiv. Fixed effects   & -1.908  & 0.179  & & -1.388  & 0.047  &  & -1.372  & 0.04  \\ 
\bottomrule
\end{tabular}
\begin{tablenotes}
 \footnotesize
 \item Estimation of equation \eqref{eq:app2} modeling applications as a fixed effect logit model. 
 \item Notes: Our sample is the set of all clicks on job postings monitored at the PES for job seekers in the transportation and logistics sector during week 44 of 2019, leading to an application and hiring or not. Fixed effect estimation keeps 70,557 observations for 8,105 job seekers, and 869 of them applying at least once. Thus, 17,865 observations are kept for estimation.   Significance levels: $<1\%: {}^{***}$, $<5\%: {}^{**}$, $<10\%: {}^{*}$. 
 \end{tablenotes}
 \end{threeparttable}
 }
 \end{center}
\end{table}

\paragraph{Observational data used in the valuation exercise.}
For the analysis using observational data in this paper, we consider this market from weeks 1 to 48 of 2019. There, we  use data on 1,181,902 (or 516,776) unique job seeker search sessions (or job ads); and on average, 610,986 job seekers (or 129,642 job ads) are active in a given week. We observe 75,744 successful matches in the data. Observations from week 1 to 43 of 2019 are used as a training set (representing 66,914 matches) for the two RSs; while weeks 44 to 48 (representing 8,830 matches) are used as a test set to evaluate the quality of recommendations. Sample sizes and restrictions for the experiments are detailed below.

More precisely, for the calibration of the matching probability (Appendix \ref{sec:calibration}), we use time-ordered sequences of all applications made on the PES website together with the outcomes for all individuals without occupational restrictions in the test set (\emph{i.e.}, from weeks 44 to 48 of 2019). This amounts to 85,639 job seekers, 207,544 applications and 8,830 hires. For the estimation of the model of application using observational data (Appendix \ref{app:modelobservational}), we use all clicks recorded on vacancies posted on the PES website with outcomes (application or not) for job seekers belonging to the transportation and logistics sector during week 44 of 2019, whatever the outcome (application or hiring or not). Estimation keeps 70,557 observations for 8,105 job seekers, and 869 of them applying at least once. Finally, for the comparison of the rankings of the vacancies between RSs (Appendix \ref{app:comparison}), we use the sample of 60,299 job seekers whose main sector is transportation and logistics in the test sample, where 18,873 vacancies are available during that period in this sector.

\end{document}

%% file: dessins/Muse0.tex
\begin{tikzpicture}[scale=1, every node/.style={scale=1}]

\definecolor{geocolor}{RGB}{255,190,190}
\definecolor{skillcolor}{RGB}{190,190,255}
\definecolor{othercolor}{RGB}{190,255,190}
\definecolor{V0color}{RGB}{255,255,190}
\definecolor{VNcolor}{RGB}{190,255,255}
\definecolor{V1color}{RGB}{255,190,255}
\definecolor{VNarrowcolor}{RGB}{10,150,150}
\definecolor{myblue}{RGB}{80,80,160}
\definecolor{mygreen}{RGB}{80,160,80}
\tikzstyle{point1}=[fill=myblue, circle, draw, inner sep=0.05cm]
\tikzstyle{point2}=[fill=mygreen, circle, draw, inner sep=0.05cm]

\node (XT) at (2, -0.5) {Job Seeker Variables};
\node[draw,fill=geocolor] (X1) at (0, 0) {$\x.geo$};
\node[draw,fill=geocolor] (XS1) at (0, 2) {$\phi_{geo}(\x)$};
\node[fill, circle] (X11) at (-0.25, 1) {};
\node[fill, circle] (X12) at (0.25, 1) {};
\node[draw,fill=skillcolor] (X2) at (1.8, 0) {$\x.sk$};
\node[draw,fill=skillcolor] (XS2) at (1.8, 2) {$\phi_{sk}(\x)$};
\node[fill, circle] (X21) at (1.25, 1) {};
\node[fill, circle] (X22) at (1.75, 1) {};
\node[draw,fill=othercolor] (X3) at (3.5, 0) {$\x.gal$};
\node[draw,fill=othercolor] (XS3) at (3.5, 2) {$\phi_{gal}(\x)$};
\node[fill, circle] (X31) at (3, 1) {};
\node[fill, circle] (X32) at (3.5, 1) {};
\node[fill, circle] (X33) at (4, 1) {};
\node[draw,fill=V0color] (X) at (2, 3) {\textit{Embedding} Job Seeker $\phi_0(\x)$};

\node (YT) at (8, -0.5) {Job Ads Variables};
\node[draw,fill=geocolor] (Y1) at (6, 0) {$\y.geo$};
\node[draw,fill=geocolor] (YS1) at (6, 2) {$\psi_{geo}(\y)$};
\node[fill, circle] (Y11) at (5.75, 1) {};
\node[fill, circle] (Y12) at (6.25, 1) {};
\node[draw,fill=skillcolor] (Y2) at (7.8, 0) {$\y.sk$};
\node[draw,fill=skillcolor] (YS2) at (7.8, 2) {$\psi_{sk}(\y)$};
\node[fill, circle] (Y21) at (7.25, 1) {};
\node[fill, circle] (Y22) at (7.75, 1) {};
\node[draw,fill=othercolor] (Y3) at (9.5, 0) {$\y.gal$};
\node[draw,fill=othercolor] (YS3) at (9.5, 2) {$\psi_{gal}(\y)$};
\node[fill, circle] (Y31) at (9, 1) {};
\node[fill, circle] (Y32) at (9.5, 1) {};
\node[fill, circle] (Y33) at (10, 1) {};
\node[draw,fill=V0color] (Y) at (8, 3) {\textit{Embedding} Job Ads $\psi_0(\y)$};

\node[draw,fill=V0color] (XY) at (5, 4) {$\Vo (\x, \y) = <\phi_0(\x), \psi_0(\y) >  $};

\draw [->] (X1) -> (X11);
\draw [->] (X1) -> (X12);
\draw [->]  (X11) -> (XS1);
\draw [->]  (X12) -> (XS1);
\draw [->] (XS1) -> (X);

\draw [->] (X2) -> (X21);
\draw [->] (X2) -> (X22);
\draw [->]  (X21) -> (XS2);
\draw [->]  (X22) -> (XS2);
\draw [->] (XS2) -> (X);

\draw [->] (X3) -> (X31);
\draw [->] (X3) -> (X32);
\draw [->] (X3) -> (X33);
\draw [->] (X31) -> (XS3);
\draw [->] (X32) -> (XS3);
\draw [->] (X33) -> (XS3);
\draw [->] (XS3) -> (X);

\draw [->] (Y1) -> (Y11);
\draw [->] (Y1) -> (Y12);
\draw [->]  (Y11) -> (YS1);
\draw [->]  (Y12) -> (YS1);
\draw [->] (YS1) -> (Y);

\draw [->] (Y2) -> (Y21);
\draw [->] (Y2) -> (Y22);
\draw [->]  (Y21) -> (YS2);
\draw [->]  (Y22) -> (YS2);
\draw [->] (YS2) -> (Y);

\draw [->] (Y3) -> (Y31);
\draw [->] (Y3) -> (Y32);
\draw [->] (Y3) -> (Y33);
\draw [->] (Y31) -> (YS3);
\draw [->] (Y32) -> (YS3);
\draw [->] (Y33) -> (YS3);
\draw [->] (YS3) -> (Y);

\draw [->] (X) -> (XY);
\draw [->] (Y) -> (XY);

\end{tikzpicture}

%% file: dessins/Muse1.tex
\begin{tikzpicture}[scale=0.9, every node/.style={scale=1.5}]

\definecolor{geocolor}{RGB}{255,190,190}
\definecolor{skillcolor}{RGB}{190,190,255}
\definecolor{othercolor}{RGB}{190,255,190}
\definecolor{V0color}{RGB}{255,255,190}
\definecolor{VNcolor}{RGB}{190,255,255}
\definecolor{V1color}{RGB}{255,190,255}
\definecolor{VNarrowcolor}{RGB}{10,150,150}
\definecolor{myblue}{RGB}{80,80,160}
\definecolor{mygreen}{RGB}{80,160,80}
\tikzstyle{point1}=[fill=myblue, circle, draw, inner sep=0.05cm]
\tikzstyle{point2}=[fill=mygreen, circle, draw, inner sep=0.05cm]

\node[draw,fill=othercolor] (X1) at (-2, -3) {$\x.gal$};
\node (XT) at (0, -4.5) [text width = 4cm] {Job Seeker Variables and embedding};
\node[draw,fill=V0color] (X2) at (0, -3) {$\phi_0(\x)$};
\node[draw] (XS1) at (-1, 0) {$\phi_1(\x)$};
\node[fill, circle] (X11) at (-1.5, -1.5) {};
\node[fill, circle] (X12) at (-1, -1.5) {};
\node[fill, circle] (X13) at (-0.5, -1.5) {};
\node[draw] (X) at (6, 0) {$\phi_1(\x) \odot \psi_1(\x,\y)$};

\node (YT) at (15, -4.5) [text width = 4cm] {Job Ads Variables and embedding};
\node (YT2) at (6.25, -4) {SDR sub-criteria};
\node (YT3) at (10, -7.2) {\V.0\ scores/rank};
\node[draw] (Y1) at (5.60, -3) {$c_k(\x,\y)$};
\node[draw,fill=V0color] (Y2) at (10, -6.3) {$\Vo .0 (\x, \y)$ $r.0(\x,\y)$};
\node[draw,fill=othercolor] (Y3) at (12.8, -3) {$\y.gal$};
\node[draw,fill=V0color] (Y4) at (14.8, -3) {$\psi_0(\y)$};
\node[draw] (YS3) at (10, 0) {$\psi_1(\x,\y)$};
\node[fill, circle] (Y31) at (9.5, -1.5) {};
\node[fill, circle] (Y32) at (10, -1.5) {};
\node[fill, circle] (Y33) at (10.5, -1.5) {};

\node[fill, circle] (W1) at (7.5, 2) {};
\node[fill, circle] (W2) at (8, 2) {};
\node[fill, circle] (W3) at (7, 2) {};

\node[draw] (XYT1) at (4, 3.9) {$\Vu (\x, \y)$=$\Hu (\x, \y)$};
\node[draw,fill=VNcolor] (XYT2) at (12.5, 5.1) {$\Hd (\x, \y)$};
\node[draw,fill=VNcolor] (XY) at (7, 7) {$\Vdt (\x, \y)$};

\node (VDT) at (-2.5, 5) {\Vdt\ };
\node (VDT2) at (-2.5, 4) {\Vu\ };
\draw[line width=0.5 mm,dotted] (-4, 4.5) -- (16, 4.5);

\node[fill, circle,fill=VNarrowcolor] (XY1) at (7, 6) {};
\node[fill, circle,fill=VNarrowcolor] (XY2) at (7.5, 6) {};
\node[fill, circle,fill=VNarrowcolor] (XY3) at (6.5, 6) {};

\draw [->] (X1) -> (X11);
\draw [->] (X1) -> (X12);
\draw [->] (X1) -> (X13);
\draw [->] (X2) -> (X11);
\draw [->] (X2) -> (X12);
\draw [->] (X2) -> (X13);
\draw [->]  (X11) -> (XS1);
\draw [->]  (X12) -> (XS1);
\draw [->]  (X13) -> (XS1);
\draw [->] (XS1) -> (X);

\draw [->] (Y1) -> (Y31);
\draw [->] (Y1) -> (Y32);
\draw [->] (Y1) -> (Y33);
\draw [->] (Y2) -> (Y31);
\draw [->] (Y2) -> (Y32);
\draw [->] (Y2) -> (Y33);
\draw [->] (Y3) -> (Y31);
\draw [->] (Y3) -> (Y32);
\draw [->] (Y3) -> (Y33);
\draw [->] (Y4) -> (Y31);
\draw [->] (Y4) -> (Y32);
\draw [->] (Y4) -> (Y33);
\draw [->] (Y31) -> (YS3);
\draw [->] (Y32) -> (YS3);
\draw [->] (Y33) -> (YS3);
\draw [->] (YS3) -> (X);

\draw [->] (YS3) -> (W1);
\draw [->] (YS3) -> (W2);
\draw [->] (YS3) -> (W3);
\draw [->] (XS1) -> (W1);
\draw [->] (XS1) -> (W2);
\draw [->] (XS1) -> (W3);
\draw [->] (X) -> (W1);
\draw [->] (X) -> (W2);
\draw [->] (X) -> (W3);
\draw [->] (W1) -> (XYT1);
\draw [->] (W2) -> (XYT1);
\draw [->] (W3) -> (XYT1);
\draw [->] (W1) -> (XYT2);
\draw [->] (W2) -> (XYT2);
\draw [->] (W3) -> (XYT2);
\draw [->] (XYT2) -> (XY1);
\draw [->] (XYT2) -> (XY2);
\draw [->] (XYT2) -> (XY3);
\draw [->] (XYT1) -> (XY1);
\draw [->] (XYT1) -> (XY2);
\draw [->] (XYT1) -> (XY3);i
\draw [->] (XY1) -> (XY);
\draw [->] (XY2) -> (XY);
\draw [->] (XY3) -> (XY);

\end{tikzpicture}